%% file: _main.tex
\documentclass[12pt]{extarticle}
\usepackage[linenumbers,preprint]{imsart}

\usepackage[margin=1.5in]{geometry}
\usepackage{amsmath}
\usepackage{amssymb,amsthm,mathtools}
\usepackage{mathtools}
\usepackage{bm}
\usepackage[unicode,colorlinks,citecolor=blue,urlcolor=black,linkcolor=blue,pdfborder={0 0 0}]{hyperref}
\usepackage[capitalize,sort,compress]{cleveref} %
\usepackage{comment}
\usepackage{graphicx}
\usepackage{enumitem}
\usepackage{framed}
\usepackage{dcolumn}
\usepackage{chngcntr}
\usepackage{dsfont} %

\usepackage{url} %
\usepackage[usenames]{color} %
\usepackage{subfig}
\usepackage[sort]{natbib}

\usepackage{datetime}
\usepackage[normalem]{ulem}
\usepackage[linesnumbered,ruled,vlined]{algorithm2e}

\usepackage{autonum}
\usepackage{microtype} %
\usepackage{booktabs} %
\usepackage{graphicx}
\usepackage{thmtools}

\newtheorem{assumption}{Assumption}
\newtheorem{theorem}{Theorem}
\newtheorem{proposition}{Proposition}
\newtheorem{lemma}{Lemma}
\newtheorem{example}{Example}
\newtheorem{definition}{Definition}
\newtheorem{remark}{Remark}

\crefname{theorem}{Theorem}{Theorems}
\crefname{lemma}{Lemma}{Lemmas}
\crefname{proposition}{Proposition}{Propositions}
\crefname{corollary}{Corollary}{Corollaries}
\crefname{conjecture}{Conjecture}{Conjectures}
\crefname{definition}{Definition}{Definitions}
\crefname{assumption}{Assumption}{Assumptions}
\crefname{example}{Example}{Examples}
\crefname{remark}{Remark}{Remarks}
\crefname{algocf}{Algorithm}{Algorithms}

\AtBeginDocument{ %
\def\[#1\]{\begin{align}#1\end{align}}
}

\input{shortcuts}

\graphicspath{{./figures/}}
\begin{document}

\begin{frontmatter}
    \title{Robust Model Selection for Discovery of Latent Mechanistic Processes}
	\runtitle{~Robust Model Selection for Mechanistic Discovery}

	\runauthor{J.\ Li, N.\ Nguyen, M. Lai, et al.}

	\begin{aug}
		\author[MS]{\fnms{Jiawei} \snm{Li}\ead[label=jl]{jwli@bu.edu}\thanksref{A}},
		\author[SE]{\fnms{Nguyen} \snm{Nguyen}\ead[label=nn]{nguyenpn@bu.edu}\thanksref{A}},
		\author[CDS]{\fnms{Meng} \snm{Lai}\ead[label=ml]{menglai@bu.edu}\thanksref{A}}, \\
		\author[SE,CDS]{\fnms{Ioannis Ch.} \snm{Paschalidis}\ead[label=yp]{yannisp@bu.edu}}
		\and
		\author[MS,CDS]{\fnms{Jonathan H.} \snm{Huggins}\ead[label=jh]{huggins@bu.edu}\thanksref{B}}
		\address[MS]{Department of Mathematics \& Statistics, Boston University, USA \printead{jl,jh}}
		\address[SE]{Division of Systems Engineering, Boston University, USA \printead{nn,yp}}
		\address[CDS]{Faculty of Computing \& Data Sciences, Boston University, USA \printead{ml}}
	\end{aug}
    \thankstext{A}{Equal contribution.}
     \thankstext{B}{Corresponding author.}

	\maketitle
	\vspace{1em} 
	\input{abstract}

\end{frontmatter}

\doublespacing

\input{content}

\singlespacing 

\bibliographystyle{imsart-nameyear}
\bibliography{paper-ref}

\newpage
\appendix

\input{appendix}

\end{document}

%% file: shortcuts.tex
\newcommand{\methodname}{\textsf{ACDC}\xspace}

\definecolor{WowColor}{rgb}{.75,0,.75}
\definecolor{SubtleColor}{rgb}{0,0,.50}

\newcounter{margincounter}

\newcommand{\charalphabetmacro}[3]{
	\def\mydeffoo##1{\expandafter\def\csname #1##1\endcsname{#2{##1}}}
	\def\mydefallfoo##1{\ifx##1\mydefallfoo\else\mydeffoo##1\expandafter\mydefallfoo\fi}
	\expandafter \mydefallfoo #3\mydefallfoo
}

\newcommand{\texalphabetmacro}[3]{
	\def\mydeftex##1{\expandafter\def\csname #1##1\endcsname{#2{\csname ##1\endcsname}}}
	\def\mydefalltex##1{\ifx\mydefalltex##1\else\mydeftex{##1}%
			\lowercase{\mydeftex{##1}}\expandafter\mydefalltex\fi}
	\expandafter \mydefalltex #3\mydefalltex
}

\newcommand{\upperCaseRomanLetters}{ABCDEFGHIJKLMNOPQRSTUVWXYZ}
\newcommand{\lowerCaseRomanLetters}{abcdefghijklmnopqrstuvwxyz}
\newcommand{\lowerCaseRomanLettersNoMHT}{abcdefgijklnopqrsuvwxyz}

\newcommand{\lowerCaseRomanLettersNoMF}{abcdeghijklnopqrstuvwxyz}

\newcommand{\lowerCaseGreekLetters}{{alpha}{beta}{gamma}{delta}{epsilon}{zeta}{eta}{theta}{iota}{kappa}{lambda}{mu}{nu}{xi}{omicron}{pi}{rho}{sigma}{tau}{upsilon}{phi}{chi}{psi}{omega}}
\newcommand{\lowerCaseGreekLettersNoEta}{{alpha}{beta}{gamma}{delta}{epsilon}{zeta}{theta}{iota}{kappa}{lambda}{mu}{nu}{xi}{omicron}{pi}{rho}{sigma}{tau}{upsilon}{phi}{chi}{psi}{omega}}
\newcommand{\upperCaseGreekLettersInLaTeX}{{Gamma}{Delta}{Theta}{Lambda}{Xi}{Pi}{Sigma}{Upsilon}{Phi}{Psi}{Omega}}

\charalphabetmacro{bar}{\bar}{\upperCaseRomanLetters}
\charalphabetmacro{bar}{\bar}{\lowerCaseRomanLetters}

\charalphabetmacro{b}{\boldorbar}{\upperCaseRomanLetters}
\charalphabetmacro{b}{\boldorbar}{\lowerCaseRomanLettersNoMF} 

\texalphabetmacro{b}{\boldorbar}{\lowerCaseGreekLettersNoEta} 

\texalphabetmacro{b}{\boldorbar}{\upperCaseGreekLettersInLaTeX}

\charalphabetmacro{mc}{\mathcal}{\upperCaseRomanLetters}

\newcommand{\widehatmathcal}[1]{\widehat{\mathcal{#1}}}
\charalphabetmacro{hmc}{\widehatmathcal}{\upperCaseRomanLetters}

\charalphabetmacro{h}{\widehat}{\upperCaseRomanLetters}
\charalphabetmacro{h}{\widehat}{\lowerCaseRomanLettersNoMHT}
\texalphabetmacro{h}{\widehat}{\lowerCaseGreekLetters}
\texalphabetmacro{h}{\widehat}{\upperCaseGreekLettersInLaTeX}

\newcommand{\boldhat}[1]{\mathbf{\widehat{#1}}}
\charalphabetmacro{bh}{\boldhat}{\upperCaseRomanLetters}
\charalphabetmacro{bh}{\boldhat}{\lowerCaseRomanLetters}
\texalphabetmacro{bh}{\boldhat}{\lowerCaseGreekLetters}
\texalphabetmacro{bh}{\boldhat}{\upperCaseGreekLettersInLaTeX}

\charalphabetmacro{wh}{\widehat}{\upperCaseRomanLetters}
\charalphabetmacro{wh}{\widehat}{\lowerCaseRomanLettersNoMHT}

\charalphabetmacro{td}{\tilde}{\upperCaseRomanLetters}
\charalphabetmacro{td}{\tilde}{\lowerCaseRomanLetters}
\texalphabetmacro{td}{\tilde}{\lowerCaseGreekLetters}
\texalphabetmacro{td}{\tilde}{\upperCaseGreekLettersInLaTeX}

\newcommand{\defas}{:=}

\newcommand{\stk}[2]{\ensuremath{\stackrel{\text{#2}}{#1}}}

\newcommand{\model}[1]{\mcM^{(#1)}}
\newcommand{\paramSpace}[1]{\Theta^{(#1)}}

\newcommand{\opt}{\star}

\newcommand{\datadist}{P_{o}}
\newcommand{\distDisc}{\mcD_{\mathrm{dist}}}
\newcommand{\compDisc}{\mcD_{\mathrm{comp}}}
\newcommand{\compDiscEst}{\widehat{\mcD}_{\mathrm{comp}}}
\newcommand{\Int}[1]{\operatorname{Int}({#1})}

\def\norm#1{\left\|{#1}\right\|} 

\newcommand{\kl}[2]{\mathrm{KL}(#1 \mid #2)}

\newcommand{\klest}[2]{\mathrm{\widehat{KL}}(#1 \mid #2)}
\newcommand{\klestsub}[4]{\mathrm{\widehat{KL}^{#1}}_{#2}(#3 \mid #4)}

\newcommand{\data}[1]{x_{#1}}
\newcommand{\numobs}{N}
\newcommand{\numcomps}{K}

\newcommand{\allparam}{\theta}
\newcommand{\param}{\phi}

\newcommand{\discr}[2]{\mcD(#1\mid #2)}
\newcommand{\discrest}[2]{\widehat\mcD(#1\mid #2)}

\newcommand{\blmetric}{d_\mathrm{BL}}

\renewcommand{\Pr}{\text{pr}}
\newcommand{\dee}{\mathrm{d}}
\newcommand{\op}{o_{P}}
\newcommand{\BLnorm}[1]{\Vert #1 \Vert_{\mathrm{BL}}}

\def\argmax{\operatornamewithlimits{arg\,max}}
\def\argmin{\operatornamewithlimits{arg\,min}}
\newcommand{\distNamed}[1]{{\sf{#1}}}
\newcommand{\distCat}{\distNamed{Categorical}}

\newcommand{\distPoiss}{\distNamed{Poiss}}

\newcommand{\distNorm}{\mathcal{N}}
\newcommand{\distSNorm}{\mathcal{SN}}
\newcommand{\distMulti}{\distNamed{Multi}}
\newcommand{\distGamma}{\distNamed{Gamma}}

\newcommand{\reals}{\ensuremath{\mathbb{R}}}

\newcommand{\nats}{\ensuremath{\mathbb{N}}}

\newcommand{\distas}{\sim}
\newcommand{\distiid}{\stk{\distas}{iid}}
\newcommand{\distind}{\stk{\distas}{ind}}

\newcommand{\E}{\mathbb{E}}	
\renewcommand{\Pr}{\mathbb{P}}	
\newcommand{\ind}{\mathds{1}} 

\newcommand{\lrp}[1]{\left(#1\right)}
\newcommand{\lrb}[1]{\left[#1\right]}
\newcommand{\lrc}[1]{\left\{#1\right\}}

\newcommand{\abs}[1]{\left\lvert #1\right\rvert} 

\newcommand{\veps}{\varepsilon}

\newcommand{\Norm}{\mathcal{N}}

\newcommand{\Unif}{\distNamed{Unif}}

\newcommand{\Poiss}{\distNamed{Poiss}}

\newcommand{\transpose}{^{\mathrm{T}}}

\newcommand{\inv}{^{-1}}

\newcommand{\mop}[1]{\mathop{\mathrm{#1}}}
\newcommand*\blank{{\mkern 2mu\cdot\mkern 2mu}}

\let\oldP\P
\renewcommand\P{\relax\ifmmode\mathbb{P}\else\oldP\fi} 
\newcommand{\Var}{\operatorname{Var}}	

\newcommand{\convp}{\overset{p}{\to}}
\newcommand{\convd}{\overset{d}{\to}}

\let\oldd\d
\renewcommand\d{\relax\ifmmode\mathrm{d}\else\oldd\fi} 

%% file: abstract.tex
\begin{abstract}
When learning interpretable latent structures using model-based approaches,
even small deviations from modeling assumptions can lead to inferential results that are 
not mechanistically meaningful. 
In this work, we consider latent structures that consist of $K_o$ mechanistic processes, where 
$K_o$ is unknown.
When the model is misspecified, likelihood-based model selection methods can substantially overestimate $K_o$ 
while more robust nonparametric methods can be overly conservative.
Hence, there is a need for approaches that combine the sensitivity of likelihood-based methods with the robustness of nonparametric ones.
We formalize this objective in terms of a  \emph{robust model selection consistency} property, 
which is based on a component-level discrepancy measure that captures the mechanistic 
structure of the model. 
We then propose the \emph{accumulated cutoff discrepancy criterion} (\methodname), which leverages plug-in estimates of component-level discrepancies. 
To apply \methodname, we develop mechanistically meaningful component-level discrepancies for a general
class of latent variable models that includes unsupervised and supervised variants of probabilistic matrix factorization and mixture modeling. 
We show that \methodname is robustly consistent when applied to unsupervised matrix factorization and mixture models. 
Numerical results demonstrate that in practice our approach reliably identifies a mechanistically meaningful number 
of latent processes in numerous illustrative applications, outperforming existing methods. 
\end{abstract}


%% file: content.tex
\section{Introduction}
\label{sec:intro}

Characterizing latent structures with meaningful real-world interpretations is a common task in scientific applications of statistical methods.
This task often amounts to discovering unobserved physical ``processes'' that generate observable quantities.
For example, processes might correspond to subpopulations that cannot be directly observed such as
types of cells \citep{Gorsky:2020,Prabhakaran:2016}, behavioral genotypes \citep{Stevens:2019}, or
groups with canonical patterns of IQ development \citep{Bauer:2007}.
Processes could correspond to a variety of scientifically important objects such cell programs \citep{Kotliar_Identify_Cell_Idendity_Activity_NMF_2019,Buettner_FscLVM_ScalableVersatile_FA_2017,Risso_General_Flexible_Signal_Extract_2018},
mutational processes in tumors %
\citep{Levitin_DeNovo_Gene_Signature_Identification_2019,Kinker_Pan_Cancer_2020,Seplyarskiy_PopulationSequencingData_2021},
or material types
\citep{Fevotte_NonlinearHyperspectralUnmixing_2015,Rajabi_SpectralUnmixingHyperspectral_2015}.

In practice, it is necessary to not only characterize each latent process but also determine \emph{how many} such processes there are.
This requires solving a model selection problem for a sequence of model families $\model{K} = \{ P_\theta : \theta \in \Theta_K\}$, $K = 1, 2, \dots,$ where $K$ denotes the number of latent processes in the model.
For example, $K$ could be the number of components in a mixture model or the number of factors in a factor analysis model.
Given some observed data $\data{1},\dots, \data{N} \distiid P_o$ that were generated by the output of $K_{o}$ latent processes,
the goal is to recover $K_{o}$ and a model $\widetilde{P}_o \in \model{K_{o}}$ such that $\widetilde{P}_o  \approx P_{o}$.
Likelihood-based model selection methods provide consistent estimation of $K_o$ when $\model{K_o}$ is \emph{well-specified} (that is, $P_o \in \model{K_o}$ but $P_o \notin \model{K}$ if $K < K_o$).
However, if $\model{K_o}$ is \emph{misspecified} (that is, $P_o \notin \model{K_o}$), then these methods do not work as intended
\citep{Cai:2021,Guha:2021,Fruhwurth:2006,Miller:2019,Xue:2024}.
The reason for this failure is that they optimize for some measure of \emph{predictive performance} (e.g., some form of expected log loss).
Therefore, as $N \to \infty$, these methods will select a sequence of models that converge to the distribution $P_\star \in \model{\infty} = \bigcup_{K=1}^\infty \model{k}$ that has minimal Kullback--Leibler divergence between $P_o$ and all $P \in \model{\infty}$.
Since typically $P_\star \notin \model{K}$ for any finite $K$, when using a predictive method for
model selection, as the number of observations $N$ increases, rather than obtaining better estimates, the opposite occurs:
the distribution $P_o$ is better estimated by adding spurious latent structures that compensate for the shortcomings of the parametric model.
This problem is known as \emph{overfitting}  \citep{Cai:2021}. %

However, when trying discover real-world processes, the statistical problem is no longer predictive in nature, but rather \emph{explanatory}  \citep{Shmueli:2010}.
The aim is to infer meaningful constructs that capture the underlying causal structure 
-- in this case the latent processes -- even if doing so results in a fitted model with reduced
predictive power.
An alternative set of approaches are nonparametric or semi-parametric in nature. 
With weaker assumptions, they can often be more robust.
However, because they do not rely on parametric assumptions, they tend to be less sensitive and can (but don't always) underestimate the number of latent components. 
These nonparametric methods also lack the generality of the likelihood-based approaches. 
Hence, there is a need for robust model selection procedures that have the generality of likelihood-based methods like Akaike, Bayesian, and deviance information criteria \citep{Akaike:1974,Schwarz:1978,Spiegelhalter:2002} while
retaining the robustness to parametric assumptions provided by the nonparametric methods. 
Notably, information criteria remain widely used due to their simplicity and the ease with
which they can be incorporated into data analysis workflows -- for example, when a user wants to use an existing (perhaps specialized) method to estimate the parameters of each model $\model{K}$ ($K = 1, 2, \dots$).
We defer detailed discussions of related work to \cref{sec:related-work,sec:mixture-model-applications,sec:pmf-applications}. %

\paragraph{Summary of Contributions.}
In this paper, we make the following contributions to the development of more reliable
and robust methods for model selection for discovering real-world latent processes:
\begin{enumerate}
    \item We define a formal notion of \textit{robust model selection consistency}
    and argue that it captures key features any robust model selection method 
    should satisfy. 
	\item We propose the \emph{accumulated cutoff discrepancy criterion} (\methodname), as a simple, flexible approach to robust model selection. 
	\item We show how to apply \methodname to a broad class of models
	in which the observations are determined by combining unobserved outputs from $K$ latent processes.
	\item We prove that \methodname provides robust model selection 
    consistency for mixture models and probabilistic matrix factorization models. 
	\item We illustrate the advantages of \methodname through numerical experiments with 
	    simulated and real data, including demonstrating state-of-the-art performance for
        cell discovery using single-cell RNA sequencing data.\footnote{Code to reproduce all experiments is available on GitHub: \url{https://github.com/TARPS-group/robust-model-selection-for-discovery}.} %
\end{enumerate}

\section{Methodology}
\label{sec:methods}

We first define robust model selection consistency, which naturally 
leads to a plug-in procedure, the accumulated cutoff discrepancy criterion (\methodname). 
We then describe how to apply \methodname for a broad class of latent
variable models.
In this section, we use mixture modeling as a running example to illustrate ideas. 
Let $\mcF = \{ F_{\phi} \mid \phi \in \Phi \}$ denote the component distribution family
and let $\eta \in \Delta_{\numcomps}$ denote the component weights, 
where $\Delta_\numcomps = \{ \eta \in \reals_{+}^{\numcomps} \mid \sum_{k=1}^{\numcomps} \eta_{k} = 1 \}$ 
is the $(\numcomps-1)$-dimensional probability simplex. 
Denote the parameter set for the $\numcomps$-component mixture distributions by 
$\Theta^{(\numcomps)} = \Delta_{\numcomps} \times \Phi^{\numcomps}$, so the 
mixture model distribution family is 
\[
\model{K} = \textstyle \left\{ P_\theta = \sum_{k=1}^K \eta_k F_{\phi_k} : \theta = (\eta, \phi_1,\dots,\phi_K) \in \Theta^{(\numcomps)}\right\}. 
\]
We can also write the generative process of the mixture model in terms of
latent variables $z_{n} \in \{1,\dots,K\}$ that indicate which component observation $n$ belongs to:
\[
z_n \mid \theta &\distas \distCat(\eta) &
x_n \mid z_n = k, \theta &\distas F_{\phi_{k}}. 
\]
Since we are interested in isolating the contribution of each component, it is this latent variable
representation that will be most relevant. 
We discuss applications to other models in \cref{sec:framework,sec:pmf-applications}.

\subsection{Robust Model Selection Consistency} \label{sec:robust-consistency}

Generalizing the mixture model setting, consider a sequence of models $\model{1}, \model{2}, \dots, \model{K}, \dots$, where $K$ captures how many distinct latent components are generating the observed data.
Assume that $\model{K} = \{ P_{\theta} \mid \theta \in \paramSpace{K} \}$, where $\paramSpace{K}$ is the parameter space
and $P_{\theta} \in \mcP(\mathbb{X})$, the space of probability measures on the observation space $\mathbb{X}$. 
The objective is to identify the true number of processes $K_o$.
Fix a \emph{distribution-level discrepancy} $\distDisc$ on probability measures that will quantify the fit between $P_{\theta}$ and the data-generating distribution $P_{o}$.
We do not assume a unique minimizing parameter since, at the very least, 
the component indices in latent variable models are non-identifiable. 
Let $\Theta_{\star}^{(K)}(P_{o}) \defas \argmin_{\theta \in\paramSpace{K}} \distDisc(P_{o} \mid P_{\theta})$ 
denote the set of minimizing parameters.
Alternatively, a practitioner might choose a parameter estimation procedure that is not model-based, in which case it might converge to a parameter value in some other set, which we also denote by $\Theta_{\star}^{(K)}(P_{o})$.

The challenge in the misspecified setting is that for any $\theta_{\star}^{(K)} \in \Theta_{\star}^{(K)}(P_{o})$, typically $\distDisc(P_{o} \mid P_{\theta_{\star}^{(K)}})$ is not minimized at $K = K_{o}$.
In fact, in our settings of interest $\model{K} \subsetneq \model{K+1}$ but $P_{o} \notin \model{K}$ for any $K$, so
$\distDisc(P_{o} \mid P_{\theta_{\star}^{(K)}})$ is monotonically decreasing as $K$ increases. 
To correctly recover $K_{o}$, the user must specify how much $P_{\theta}$ can deviate from $P_{o}$ while remaining an acceptable approximation.
Therefore, we introduce a second discrepancy which measures how well the \emph{components} of $P_{o}$ and $P_{\theta}$ match.

Since the components of $P_{o}$ are unknown, the component contributions must be estimated based on model $P_{\theta}$ but 
using the distribution of data from $P_{o}$. 
Let the \emph{component-level realized discrepancy} $\compDisc^{(K)}(\theta, k, P_{o})$ quantify
how close the inferred component $k \in \{1,\dots,K\}$ from $P_{o}$ is to component $k$ of the model
$P_{\theta}^{(K)}$. 
To quantify the overall degree of component-level misspecification of $P_o$ with true number of components $K_o$,
define the \emph{worst-case component-wise discrepancy}
\[
 \rho(P_{o}, K_{o}) \defas  \sup_{\theta \in \Theta_{\star}^{(K_{o})}(P_{o})} \max_{k \in [K_{o}]} \compDisc^{(K_{o})}(\theta, k, P_{o}).
\]
For example, in the mixture model case we can construct a component-level realized discrepancy by inferring
the component of $P_o$ that would correspond to each mixture component.
That is, if we assign each observation from $P_o$ according to the conditional component probabilities 
$p(k \mid x, \theta) = \eta_{k} \frac{\dee F_{\phi_k}}{\dee P_{\theta}}(x)$, 
then the inferred $k$th component of $P_o$ is 
\begin{align}
	\widetilde{F}^{(\theta)}_{ok} = \frac{p(k \mid x, \theta)}{\int p(k \mid y, \theta) P_o(\dee y)} P_o. 
\end{align}
Given a choice of discrepancy measure $\mcD$ (e.g., $\distDisc$), we can set 
$
\compDisc^{(K)}(\theta, k, P_{o}) = \mcD(\widetilde{F}^{(\theta)}_{ok} \mid F_{\phi_k}).
$

\begin{definition}[Robust model selection consistency] \label{def:robust-consistency}
	Fix a function $\kappa : \reals_{+} \times \nats \to \reals_{+}$. 
	A model selection procedure $\widehat K(\data{1:N}, \rho) \in \nats$ is \emph{$\kappa$-robustly consistent for $\Theta_{\star}$ and $\compDisc$} if,
	for any data-generating distribution $P_o$ and true component number $K_o$ that satisfies 
	\[ \label{eq:mismatch-condition}
		\inf_{\theta \in \Theta_{\star}^{(K)}(P_{o})} \distDisc(P_{o} \mid P_{\theta}) \ge \kappa(\rho(P_{o}, K_{o}), K)
		\quad\text{for all $K \in \{1,\dots,K_{o}-1\}$},
	\]
	it holds that, for $\data{1}, \data{2}, \dots \distiid P_{o}$, 
	\[
		\Pr\big\{\widehat K(\data{1:N}, \rho(P_{o}, K_{o})) = K_o \big\} \overset{N \to \infty}{\longrightarrow} 1.
	\]
\end{definition}
To interpret \cref{def:robust-consistency} and the role of the function $\kappa$, 
it is helpful to compare robust model selection consistency to classical model selection consistency.
\Cref{fig:consistency-illustration} provides a cartoon illustration of the differences. 
Classical model consistency typically requires that  
(1) $P_o \in \model{K_o}$ and (2) $P_{o} \notin \model{K}$ for $K < K_o$. 
Robust consistency weakens the first condition by allowing for a worst-case discrepancy $\rho(P_{o}, K_{o}) \ne 0$,
rather than needing $\rho(P_{o}, K_{o}) = 0$.
However, robust consistency strengthens the second by requiring a gap between $P_o$ 
and all models for $K < K_o$.
The size of this gap specified in \cref{eq:mismatch-condition} in terms of $\kappa$. 
Hence, we call $\kappa$ the \emph{gap function}. 
It would be natural to ask that $\kappa(\rho, K) = \rho$, although this may not always be possible. 

\begin{figure}[tp]
	\centering
	\includegraphics[width=.7\textwidth,trim=0in 6.2in 6.2in 0in,clip]{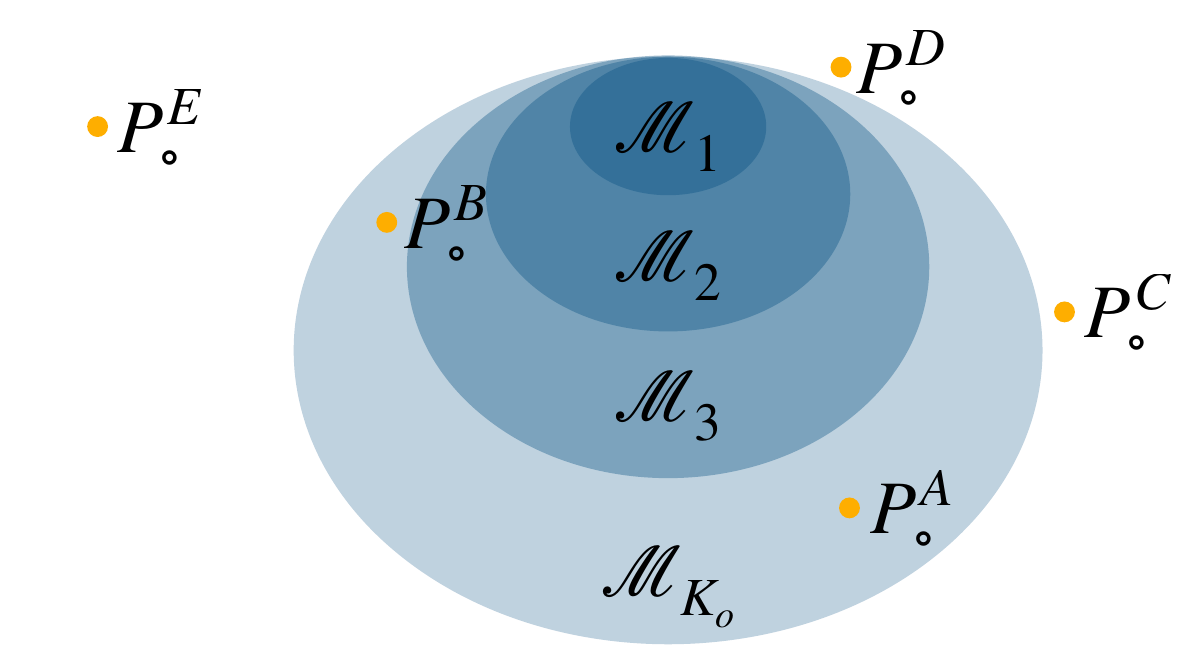}
	\caption{Cartoon illustration of the differences between traditional and robust model selection consistency
		where the true model corresponds to $K_{o} = 4$, with the nested models indicated by the gray ovals.
		We contrast five possible data-generating distributions, $P_{o}^{A}, \dots, P_{o}^{E}$, indicated by
		the gold points.
		Since $P_{o}^{A}, P_{o}^{B} \in \model{K_{o}}$ but are not in $\model{K}$ for $K < K_o = 4$,
		for these models $K_{o}$ can be consistently estimated.
		However,  since $P_{o}^{C}, P_{o}^{D}, P_{o}^{E} \notin \model{K_{o}}$, for these distributions
		$K_{o}$ cannot be estimated consistently in the traditional sense.
		On the other hand $P_{o}^{A}, P_{o}^{B}, P_{o}^{C}$, and $P_{o}^{D}$
		are close to $\model{K_{o}}$,
		so $K_{o}$ could potentially be robustly and consistently estimated in these four cases.
		However, $P_{o}^{B}$ and $P_{o}^{D}$ are also close to $\model{3}$, so robustly consistent
		estimation of $K_{o}$ is feasible only for $P_{o}^{A}$ and $P_{o}^{C}$.
		Since $P_{o}^{E}$ is far from $\model{K_{o}}$, $K_{o}$ would not be consistently estimable --
		either in the traditional or robust sense.
	}
	\label{fig:consistency-illustration}
\end{figure}

\subsection{A Plug-in Procedure}
\label{sec:method}

Inspired by \cref{def:robust-consistency}, we propose a simple plug-in procedure for model selection,.
Assuming $\rho_o = \rho(P_o, K_o)$ were known, we would want to find the smallest value of $K$ such that 
for all $k \in \{1,\dots,K\}$, we have $\compDisc^{(K)}(\theta_\star, k, P_{o}) \le \rho_o$ for $\theta_\star \in \Theta_{\star}^{(K_{o})}(P_{o})$.
However, since $P_{o}$ and $\Theta_{\star}^{(K_{o})}(P_{o})$ are unavailable, we propose to instead 
use the empirical distribution $\widehat{P}_o = \numobs^{-1} \sum_{n=1}^N \delta_{x_n}$ (here $\delta_x$ denotes
the Dirac measure at $x$) and a point estimator $\widehat\theta^{(K)}$.
Hence, we obtain the plug-in estimator $\compDiscEst^{(K,k)} = \compDisc^{(K)}(\widehat\theta^{(K)}, k, \widehat{P}_o)$.
Since values of $\compDiscEst^{(K,k)} < \rho_o$ are not important from a model selection perspective, 
we truncate the estimator by replacing it with $\max(0, \compDiscEst^{(K,k)} - \rho)$, where $\rho$ is 
a best estimate of $\rho_o$. 
We can view taking this maximum as serving a similar role to how the coarsened posterior conditions on the discrepancy having a known upper bound. 

Rather than taking the maximum over the component-wise discrepancy estimates, for better robustness
to noisy estimates we use a sum, which results in a robust model selection loss 
\[ \label{eq:robust-loss}
	\mathcal{R}^\rho(x_{1:n}, K) = \sum_{k=1}^K \max(0, \compDiscEst^{(K,k)} - \rho),
\]
where for notational simplicity we have left the dependence of $\compDiscEst^{(K,k)}$ on $x_{1:n}$ 
(as a function of $\widehat{P}_o$ and $\widehat\theta^{(K)}$) implicit. 
Since the loss is the sum (i.e., accumulation) of discrepancies that have been truncated (i.e., cut off) at $\rho$,
we call \cref{eq:robust-loss} the \emph{accumulated cutoff discrepancy criterion} (\methodname).
The corresponding robust model estimator is
\[
\widehat{K}^\rho(x_{1:n}) = \min\{\argmin_K \mathcal{R}^\rho(x_{1:n}, K)\}.
\]
Since $\argmin_{\numcomps}$ may return a set of values if the loss is equal to
zero for more than one value of $\numcomps$, it is necessary to include an
additional $\min$ operation to select the smallest value from the set.
We provide a number of methods for determining $\rho$ in \cref{sec:choosing-rho}. 

\subsection{Modeling Framework} \label{sec:framework}

\begin{figure}[tp]
	\centering
	\includegraphics[width=.8\textwidth,page=1]{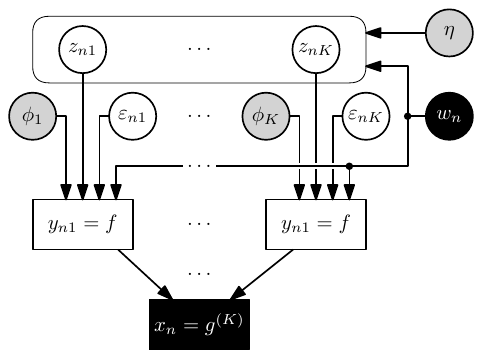}
	\caption{Graphical representation of the general form for model $\model{K}$ for a single observation $\data{n}$.
		Circles denote random variables while squares denote deterministic variables.
		A gray background indicates global parameters while a black background indicates an observed quantity.}
	\label{fig:model}
\end{figure}

Guaranteeing robust model consistency requires specific choices of model and component-level discrepancy.
In this paper, we consider a flexible modeling framework in which the observed data are the result of $K$ distinct latent sources (\cref{fig:model}). 
This framework will lead to a natural choice of component-level discrepancy, of which the mixture model
discrepancy proposed \cref{sec:robust-consistency} is a particular case. 

We allow each observation $x_n \in \mathcal{X}$ to have sample-specific covariates $w_n \in \mcW$.
Assume that $x_n$ depends on process-level contributions $y_{n1},\dots,y_{nK} \in \mathcal{Y}$ via the deterministic function $g^{(K)} \colon \mathcal{Y}^{\otimes K} \to \mathcal{X}$:
\[
	x_n = g^{(K)}(y_{n1}, \dots, y_{nK}).
\]
For example, we could have $g^{(K)}(y_1,\dots,y_K) = \sum_{k=1}^K y_k$ or $g^{(K)}(y_1,\dots,y_K) = \max_{k} y_{k}$.
The process-level contributions are in turn
determined by an observation-specific latent variable $z_n  = (z_{n1},\dots,z_{nK}) \in \mcZ^{(K)} \subseteq \mcZ^{\otimes K}$ and the process-specific parameters $\phi_1,\dots,\phi_\numcomps \in \Phi$.
We assume both $\mathcal{Y}$ and $\mcZ$ contain a \emph{null value} $\emptyset$, which indicates no contribution.
Specifically, we assume $g^{(K)}$ has the following \emph{no contribution property}: for all $y_1,\dots,y_{K-1} \in \mathcal{Y}$,
\[
	g^{(K)}(y_{1}, \dots, y_{K-1}, \emptyset) = g^{(K-1)}(y_{1}, \dots, y_{K-1}).
\]
Given a deterministic function
$f \colon \mcW \times \mcZ \times \Phi \times \reals \to \mathcal{Y}$
and independent noise random variables
$
	\varepsilon_{nk} \distiid G
$,
the component-level contributions are given by
\[
	y_{nk} = \begin{cases}
		\emptyset,                                  & \text{if $z_{nk} = \emptyset$,} \\
		f(w_{n}, z_{nk}, \phi_k, \varepsilon_{nk}), & \text{otherwise}.
	\end{cases}
\]
If there are no covariates, we drop the dependence on $w_n$ and write $f(z_{nk}, \phi_k, \varepsilon_{nk})$ instead.
Typically $\mcZ \subseteq \reals$ and $z_{nk}$ represents the activity level of the $k$th process for the $n$th observation.
In such cases, usually $\emptyset = 0$.
For a global parameter $\eta \in \mathcal{E}^{(K)}$, we assume the observation-specific latent variables are independent but may depend on the sample-specific covariates:
\[
	z_n \mid \eta, w_n \distind H^{(K)}_{\eta,w_n}.
\]
If there are no covariates, we drop the dependence on $w_{n}$ and write $H^{(K)}_{\eta}$ instead.

Our framework captures many common model types: 

\begin{example}[Mixture Modeling]\label{ex:mix-model}
	We can recover a general mixture model by taking $H_\eta^{(\numcomps)} = \distCat(\eta)$, 
	so $z_{nk} \in \{0,1\}$ and $\sum_{k=1}^K z_{nk} = 1$.
	Given a mixture component distribution family $\mcF = \{ F_{\phi} \mid \phi \in \Phi \}$,
	define $f$ such that, for  $\varepsilon \distas G$, it holds that
    $f(0, \phi, \varepsilon) = 0$ and $f(1, \phi, \varepsilon) \distas F_\phi$. 
	Finally, take $g^{(K)}$ to be the summation operator.
	Hence, if $z_{nk} = 1$, then $x_n = y_{nk} \distas F_{\phi_k}$.
\end{example}
\begin{example}[Mixture Model with Varying Component Probabilities] \label{ex:mix-model-varying}
	When covariates are available for each observation, the mixture model can be generalized to
	allow the mixture probabilities to depend on observed covariates \citep{Jaspers_BayesianEstimationMixtureCovariate_2018,Huang_MixtureRegressionModels_2012}.
	We can recover this model by using the same setup as \cref{ex:mix-model}
	but instead letting $H_{\eta,w_n}^{(\numcomps)} = \distCat(h(\eta, w_{n}))$ for some fixed function
	$h \colon \mcE^{(K)} \times \mcW \to \Delta_{\numcomps}$.
\end{example}
\begin{example}[Probabilistic Matrix Factorization]\label{ex:pmf-formulations}
	For probabilistic matrix factorization (PMF), $x_n \in \reals^{D}$.
	Let $\mcZ \subseteq \reals$ and $\Phi \subseteq \reals^{D}$.
	We assume that $\mcF = \{ F_{\mu} \mid \mu \in \reals^{D} \}$ is a location family of distributions satisfying $\int x F_\mu(\dee x) = \mu$ for all $\mu \in \reals^{D}$.
	Let $f$ satisfy $f(z, \phi, \varepsilon) \distas F_{z \phi}$ if  $\varepsilon \distas G$.
	For example, in nonnegative matrix factorization, $F_{\mu} = \Poiss(\mu)$ while in classical factor analysis $F_{\mu} = \Norm(\mu, \sigma^{2})$,
	where $\sigma^{2}$ can also be learned.
	Finally, take $g^{(K)}$ to be the summation operator.
\end{example}
Our framework is similarly applicable to supervised variants of probabilistic matrix factorization models, functional clustering problems, and a variety
of other latent variable models \citep{Carvalho:2008,Chiou:2007,Cunningham:2014,West:2003,Blei:2007,Dunson:2000}.

\subsection{Constructing the Component-level Discrepancy}

To define the component-level discrepancy, it is tempting to directly apply the approach we took for mixture models
and quantify the difference between the distributions of $y_{n1},\dots,y_{nK}$ when $x_n \distas P_o$ and the modeled distributions of $y_{n1},\dots,y_{nK}$.
However, the distributions of $y_{1k}, y_{2k}, \dots, y_{nk}, \dots$ may be different due to the sample-specific dependence on $w_n$ and $z_{nk}$.
To address this issue, we instead consider the discrepancy between the conditional distribution of the noise variables $\varepsilon_{nk}$ given $x_{1:N}$ and the assumed noise distribution $G$.
However, we must exclude $\varepsilon_{nk}$ if $y^{(K)}_{nk} = \emptyset$ because then it is no longer part of the graphical model (see \cref{fig:model}).
Dropping the dependence on $n$, the inferred distribution of the $k$th noise variable is
\[
\widetilde{G}_k^{(\theta)} = \int \mcL(\veps_{k} \mid y_k \ne \emptyset, x, w, \theta) P_o(\d x, \d w),
\]
where $\mcL(\veps_{k} \mid \cdots)$ denotes a conditional law of $\varepsilon_k$. 
Hence, define the discrepancy for the $k$th component by $\compDisc^{(K)}(\theta, k, P_o) = \discr{\widetilde G_k^{(\theta)}}{G}$.

To define an empirical version of $\compDisc^{(K)}(\theta, k, P_o)$, let
$\widehat G_{nk}^{(K)} = \mcL(\varepsilon_{nk} \mid x_n, w_n, \widehat\theta^{(K)})$, 
define the usage indicators $u_{nk}^{(K)} = \ind(y_{nk}^{(K)} \ne \emptyset)$, and denote the number of observations for component $k$ as $N_k^{(K)} = \sum_{n=1}^N u_{nk}^{(K)}$.
Then the empirical distribution of the noise variables for component $k$ is
\[
	\widehat G_k^{(K)} = \frac{1}{N_k^{(K)}} \sum_{n=1}^N u_{nk}^{(K)} \widehat G_{nk}^{(K)}.
\]
Hence, an estimate of discrepancy for the $k$th component is given by 
\[ \label{eq:general-Dcomp}
\compDiscEst^{(K,k)} = \discr{\widehat G_k^{(K)}}{G}.
\]

In some scenarios, it is necessary to replace the divergence with an estimator because 
$\discr{\widehat G_k^{(K)}}{G}$ is undefined (e.g., the KL divergence) or not efficiently computable (e.g., the Wasserstein distance).
In \cref{sec:kl-estimation}, we discuss how best to estimate the KL divergence in practice and provide supporting
consistency theory by modifying the two-sample KL divergence estimation theory developed in \citet{Wang:2009} to the one-sample estimation setting applicable to \methodname. 
In \cref{sec:sinkhorn}, we discuss using the entropy-regularized Wasserstein distance (the Sinkhorn distance) 
as a stable, efficiently computable alternative to the Wasserstein distance that scales to high dimensions.

\subsection{Choosing $\rho$} \label{sec:choosing-rho}

The value of $\rho$ is problem dependent, as it quantifies the maximum amount of model misspecification of each process.
We propose two complementary approaches to selecting $\rho$ that
take advantage of the fact that the robust loss is a piecewise linear function of $\rho$.
Therefore, given a fitted model for each candidate $\numcomps$, we can easily compute the loss for all values of $\rho$.

\paragraph{Using domain knowledge.}
The first approach aims to leverage domain knowledge. %
Specifically, it is frequently the case that some related datasets are available with ``ground-truth'' labels either
through manual labeling or via \emph{in silico} mixing of data where group labels are directly observed \citep[see, e.g.,][]{Souto:2008}.
In such cases, an empirically optimal $\rho$ value %
for one or more such datasets with ground-truth labels
can be determined by maximizing a problem-appropriate accuracy metric such as F-measure.
Because $\rho$ quantifies the divergence between the true process distributions and the model estimates,
we expect the values found using this approach will generalize to new datasets that are reasonably similar.
We illustrate this approach in \cref{sec:flow-cytometry}.
Alternatively, if real data with ground truth $K_o$ is unavailable, plausible simulated data could 
be used to calibrate $\rho$ instead \citep{Xue:2024}.

\paragraph{A generally applicable approach.}
For applications where there are no related datasets with ground-truth labels available, we
propose a second approach.  %
After estimating the model parameters for each fixed $\numcomps$
and computing all process-wise divergences, we plot the loss as a function of $\rho$ for each $\numcomps \in \{\numcomps_{\min},\dots,\numcomps_{\max}\}$.
For readability, introduce a small positive constant $\lambda \ll 1$ and plot $\mathcal{R}^{\rho}(\data{1:\numobs}, K) + \lambda \numcomps$
so it is possible to distinguish the lines when the loss is exactly zero.
The optimal model is determined by identifying the number of processes which is best over a substantial range of $\rho$ values,
with $\rho$ as small as possible.
The idea behind this selection rule is to identify the first instance of stability, indicating that allowing just a small amount of additional misspecification (by increasing $\rho$) doesn't notably improve the loss.
However, subsequent stable regions that appear afterward might introduce too much tolerance, potentially resulting in model underfitting.
This approach is similar in spirit to the one introduced for heuristically selecting the $\alpha$ parameter for the coarsened posterior \citep{Miller:2019}.

\paragraph{Automation.}
Building on the same heuristic intuition, we can automate the selection of $\rho$ by defining a minimum width $\Delta_{\min}$ for which an interval can be recognized as the stability region. Specifically,
we keep track of the smallest $\rho$ value, $\rho_\mathrm{start}$, for which a $K$-specific penalized loss becomes minimal among all other $K$-specific losses.
We then identify the next $\rho$ value, $\rho_\mathrm{end}$, at which this same loss curve is no longer the minimum.
The difference between $\Delta = \rho_\mathrm{end} - \rho_\mathrm{start}$ defines an interval of stability.
If $\Delta \ge \Delta_{\min}$ (i.e., $\Delta$ has the predefined minimum width), it is recognized as a stability interval, the corresponding value of $K$ is chosen to be $\widehat K$. Otherwise, $\rho_\mathrm{end}$ becomes $\rho_\mathrm{start}$
and repeat the procedure to compute the new $\rho_\mathrm{end}$ and $\Delta$, check if $\Delta \ge \Delta_{\min}$, and so forth.
The value of $\Delta_{\min}$ should be set based on preliminary manual experiments to estimate a suitable stability region width for automated selection in larger batches. This approach allows users to adjust the interval threshold to balance the tradeoff between avoiding underfitting ($\Delta_{\min}$ sufficiently small) and ensuring appropriately conservative and stable model selection ($\Delta_{\min}$ sufficiently large).
We demonstrate the utility of this approach in \cref{sec:scRNA}. %

\subsection{Related Work} \label{sec:related-work}

There is limited work on general-purpose approaches to robust model selection with the goal
of ensuring interpretability.
Recent work on robustifying likelihoods to small model perturbations offer one promising 
strategy \citep{Chakraborty:2023,Dewaskar:2023,Wu:2024}.
However, these methods aim to replace existing parameter estimation methods rather than augment them
-- which is a key goal of the present work.
Perhaps most closely related to our work is \citet{Miller:2019}, which proposes an elegant robust Bayesian model selection
procedure that employs a technique they call \emph{coarsening}.
Unlike the standard posterior, which assumes the data were generated from the assumed model (that is, it conditions on $\data{} = \data{\mathrm{obs}}$), the \emph{coarsened posterior} conditions on the ``true model data'' being close to the observed data
(that is, it conditions on the estimated Kullback--Leibler divergence between $\data{}$ and $\data{\mathrm{obs}}$ being less than some threshold $\gamma$) -- hence, sacrificing predictive power for greater robustness.
However, using the coarsened posterior approach has a potentially high computational cost because
it requires running Markov chain Monte Carlo dozens of times to heuristically determine a suitable robustness threshold \citep{Miller:2019,Xue:2024}.
In addition, our experiments show that, while coarsening offers good robustness in many scenarios, it does not have have any formal correctness guarantees
and can fail in simple situations (see \cref{sec:coarsening-limitations}).

\section{Application to Mixture Models}
\label{sec:mixture-model-applications}

In the next two sections, we illustrate the use of \methodname in some representative applications while 
also providing theoretical support for our approach by showing that \methodname is robustly consistent.  
For readability, theorems are stated informally in the main text.
Formal statements of assumptions and results are deferred to \cref{sec:theory}. 

For all experiments, unless stated otherwise we use the KL divergence as the discrepancy measure 
$\mcD$ for the calculation of the estimated component-level discrepancy 
$\compDiscEst^{(K,k)}$ defined in \cref{eq:general-Dcomp}.
We compare to BIC since, like \methodname, it only requires a point estimate for each value of $K$.
Alternative criteria like AIC and DIC would give similar results.
However, since BIC has a larger penalty than AIC (which DIC generalizes), it will be more conservative and hence tend to choose smaller values of $K$.
See, for example, \citet{Miller:2019,Xue:2024,Cai:2021} for numerical examples showing that Bayesian 
model selection does not resolve the overfitting problem.

In this section, we consider mixture modeling, for which \methodname is robustly consistent:
\begin{theorem} \label{thm:mixture-model-robust-consistency}
\methodname using $\compDiscEst^{(K,k)}$ defined in \cref{eq:general-Dcomp} is $\kappa$-robustly consistent for mixture models if the underlying component discrepancy $\mcD$ is the KL divergence, Wasserstein distance, or maximum mean discrepancy (MMD). 
For the KL divergence, $\kappa(\rho, K) = \sqrt{\rho}$ and $\distDisc = d_\mathrm{BL}$, the bounded Lipschitz distance. 
For Wasserstein distance and MMD, $\kappa(\rho, K) = \rho$ and $\distDisc = \mcD$. 
\end{theorem}

In addition to BIC, we compare \methodname to three nonparametric, mixture model-specific selection methods that rely on within-cluster dispersion. 
For $K$ clusters, the \emph{total within-cluster sum of squares} $\mathrm{WCSS}(K)$ is defined as
\[
    \mathrm{WCSS}(K) = \sum_{k=1}^{K} \sum_{x \in C_k} \| x - \mu_k \|^2,
\]
where $C_k$ is the set of observations assigned to cluster $k$ 
and $\mu_k$ is its centroid. 
The first approach is the \emph{elbow method}, which identifies the point beyond which additional clusters no longer show significant improvement in dispersion reduction \citep{Thorndike_1953}.
The \emph{silhouette coefficient}, on the other hand, assesses how well each point fits within its assigned cluster relative to its distance to the nearest neighboring cluster \citep{siluet_coef}. It favors clusters with high cohesion where points are close to others in the same cluster, and high separation where clusters are well isolated from one another. 
Finally, the \emph{gap statistic} compares the observed clustering dispersion in the data to what is expected under a null reference where data comes from uniform distribution \citep{gap_stats}. The selection is based on the number of clusters that maximizes the deviation from its baseline. 
For all mixture model experiments, we use the EM algorithm to obtain parameter estimates.
Unless stated otherwise, we use the bias-corrected KL estimator for \methodname (see \cref{sec:kl-estimation} for details).

Given estimates $\hat{K}_1,\dots,\hat{K}_T$ for $T$ datasets with 
true $K$ values of, respectively, $K_1, \dots, K_T$, we measure performance in terms of the mean absolute deviation $\operatorname{MAE} = T^{-1} \sum_{t=1}^T |\hat{K}_t - K_t|$, the mean 0-1 loss $\operatorname{0-1} = T^{-1} \sum_{t=1}^T \ind(\hat{K}_t \ne K_t)$, and the median of the signed deviations $\{ \hat{K}_t - K_t \}_{t=1}^T$.

\subsection{Simulation Study}
\label{sec:high-dim-simulation}

We first investigate the accuracy of \methodname under a variety of 
conditions on the level of misspecification, the relative sizes of the mixture components, 
and the data dimension.
We generate data $x_1, \ldots,x_{\numobs} \in \reals^{D}$ from the mixture distribution
$P_{o} = \sum_{k=1}^{\numcomps_{o}}\eta_{ok}\distSNorm(m_{ok}, \Sigma_{ok}, \gamma_{ok})$,
where $\distSNorm(m, \Sigma, \gamma)$ denotes a skew normal distribution with location vector $m \in \reals^{D}$, scale matrix $\Sigma \in \reals^{D \times D}$,
and shape $\gamma \in \reals^{D}$, which controls the skewness.
When $D > 1$, we introduce weak correlations by letting $\Sigma_{ok} = \Sigma$,
where $\Sigma_{ij}=\exp\{-(i-j)^2/\sigma^2\}$ and $\sigma$ controls the strength of the correlation.
See \cref{sec:simulation-gauss} for full experimental details. 

\paragraph{High-dimensional data.} 
We suggest using the KL divergence as a default discrepancy choice.
But, as discussed in \cref{sec:kl-estimation}, the $k$-nearest-neighbor estimator becomes less accurate with increasing data dimension.
While a general solution is unlikely to exist, we illustrate one approach to address this challenge.
Specifically, if we believe the coordinates are likely to be only weakly correlated, we can employ the $k$-nearest-neighbor
method on each coordinate by assuming the coordinates are independent.
We set $D = 25$, $N=10\,000$, $K_{o}=3$, and $\sigma=0.6$.
As shown in \cref{fig:high-dim}(left), the wide and stable region corresponds to the true correct number of components.
\cref{fig:high-dim}(right) illustrates the value of model-based clustering, particularly in high dimensions: 2-dimensional
projections of the data give the appearance of there being four clusters in total, when in fact there are only three.
\begin{figure}[tp]
	\centering
	\subfloat{\label{fig:high-dim-stare-loss}\includegraphics[width=.48\textwidth]{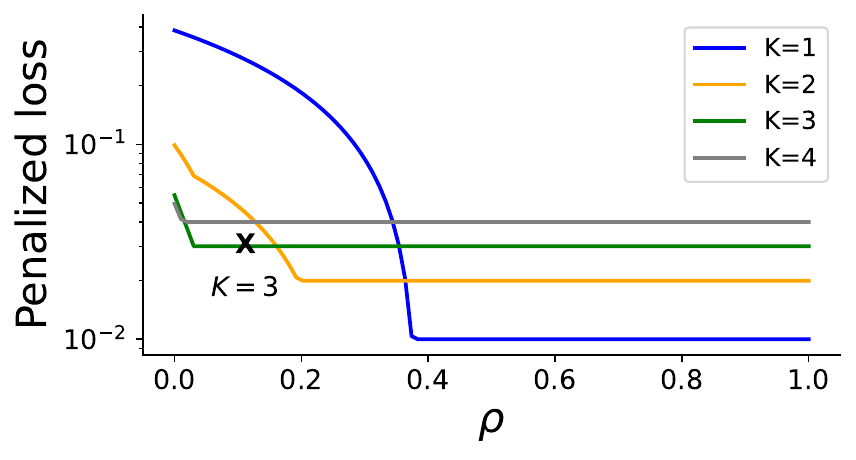}}
	\subfloat{\label{fig:high-dim-true}\includegraphics[width=.48\textwidth]{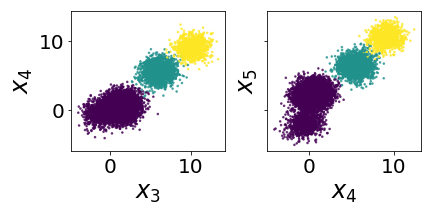}}

	\caption{Application of \methodname clustering simulated high dimensional data (\cref{sec:high-dim-simulation}).
		\textbf{Left:} The penalized loss plot for $\numcomps \in \{1,2,3,4\}$.
		Since the loss for $K=3$ is very close to equal to or much less than the loss for $K \ne 3$, it shows the greatest
		stability, resulting in $\widehat K = 3$ (indicated by the ``$X$'').
		\textbf{Right:} Selected two-dimensional projections of the data.}
	\label{fig:high-dim}
\end{figure}

\begin{figure}[tp]
	\centering
    \subfloat[low-dimensional]{\label{fig:sim-low-dim}\includegraphics[width=0.5\textwidth]{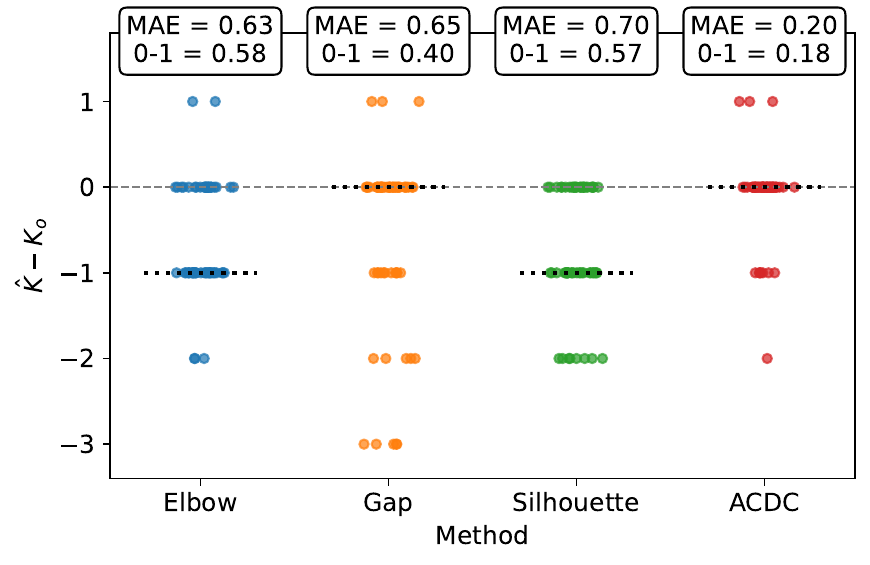}}
    \subfloat[high-dimensional]{\label{fig:sim-high-dim}\includegraphics[width=0.5\textwidth]{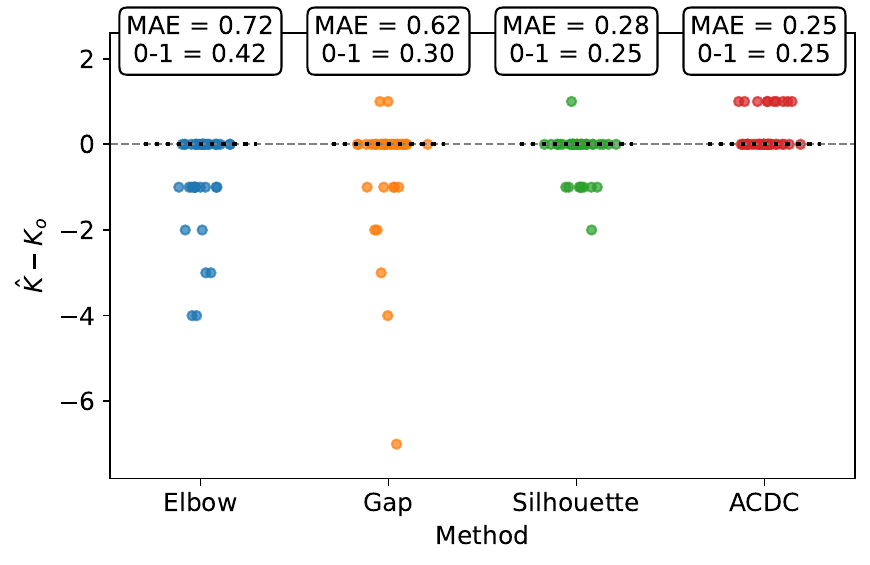}}
	\caption{Comparison of ACDC with common model selection criteria (Elbow, Gap, and Silhouette) in low-and high-dimensional settings. Each point shows the deviation between estimated and true cluster counts ($\hat{K}-K_o$) of one dataset. Dotted black lines indicate the median error for each method. Mean absolute error (MAE) and 0–1 loss quantify the model selection accuracy.}
	\label{fig:multiGMM_alphaVar}
\end{figure}

\paragraph{Comparison to other methods.} To assess the relative performance of \methodname, we compare it to the three nonparametric model selection methods. 
In the low-dimensional setting ($D \in \{2,3\}$), \methodname achieves the best overall performance, with the smallest MAE of $0.20$ and mean 0--1 loss of $0.18$, and a median deviation equal to zero (\cref{fig:sim-low-dim}).
By contrast, the competing methods exhibit MAE values of $\ge 0.60$ and mean 0--1 losses of $\ge 0.40$. 
Both the elbow and silhouette criteria show a tendency to underestimate the number of components.
The gap statistic shows the highest variability and significant underestimation. 
In the high-dimensional setting (\cref{fig:sim-high-dim}), \methodname has similar performance to the silhouette criteria while significantly outperforming 
the elbow and gap methods, particularly in terms of MAE. 
For the low and high-dimensional settings, all differences between the MAE values of
\methodname and the three alternative methods were statistically significant except for \methodname versus the silhouette coefficient applied to the high-dimensional data (Wilcoxon signed-rank test; see \cref{tab:wilcoxon-pvals} for exact p-values).

\subsection{Cell Type Discovery using Flow Cytometry}
\label{sec:flow-cytometry}

As a first investigation of \methodname applied to real data, we follow the setup 
of \citet{Miller:2019} using 12 flow cytometry datasets originally from a
longitudinal study of graft-versus-host disease %
in patients undergoing blood marrow transplantation \citep{Brinkman:2007}.
Flow cytometry is a technique used to analyze properties of individual cells in a biological material \citep{flow_cytometry}.
Typically, flow cytometry data consists of 3--20 properties of tens of thousands cells.
Cells from distinct populations tend to fall into clusters and discovering cell
types by identifying clusters is of primary interest in practice.
Usually scientists identify clusters manually, which is labor-intensive and subjective.
Therefore, clustering methods that provide interpretable groups of cells is invaluable.
Nevertheless, for our experiments we treat the manual cluster assignments as the ground truth, which was feasible because the datasets have dimension $D=4$.

To calibrate $\rho$, we follow the domain knowledge-aided approach described in \cref{sec:choosing-rho}. 
For comparability with the results of \citet{Miller:2019}, we use F-measure to quantify clustering accuracy and calibrate $\rho$ using the first 6 datasets.
As shown in \cref{fig:flowcyt-train}, the training datasets 1--6 have a nearly identical trend of
clustering accuracy as a function of $\rho$.
The averaged F-measure achieves the maximum when $\rho \approx 1.16$, which is a point of maximum F-measure for all 6 datasets.
The consistency of \methodname compares favorably to using coarsening, where drastically different $\alpha$ values maximize the F-measure \citep[Figure 5]{Miller:2019}.
These results provide evidence that our approach is taking better advantage of the common structure and degree of misspecification across datasets.

\begin{figure}[tp]
	\centering
	\includegraphics[width=80mm]{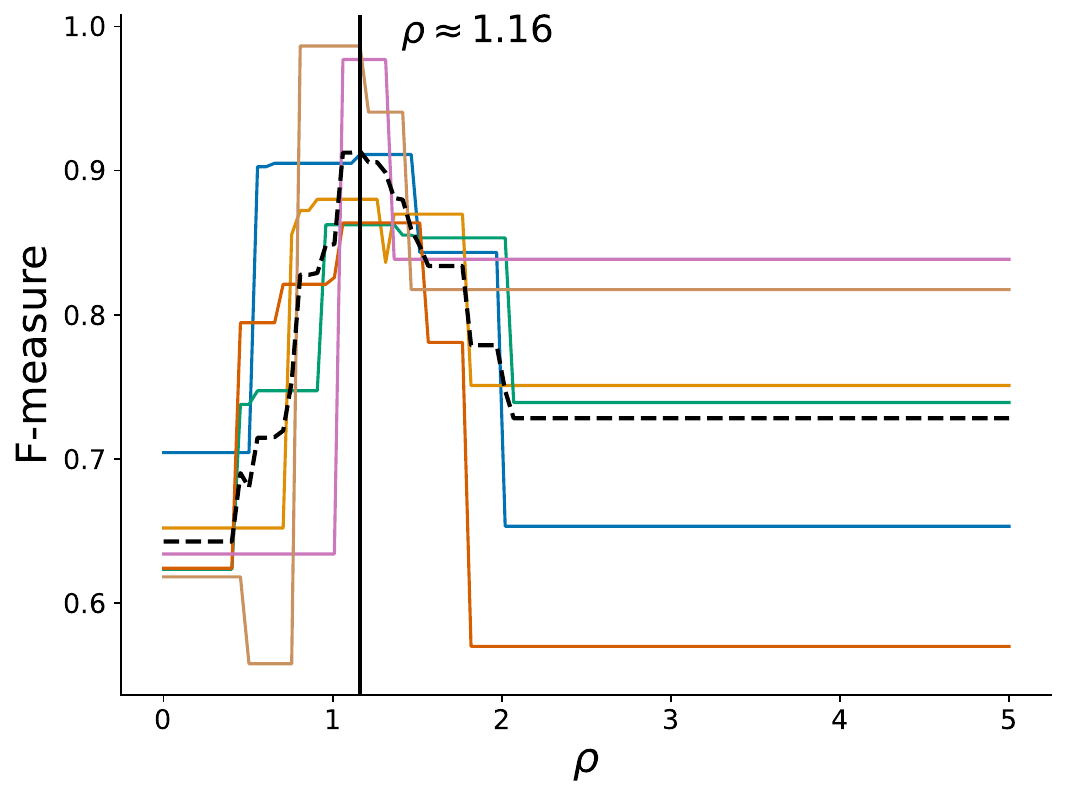}
	\caption{Selecting an optimal value of $\rho$ using the first 6 flow cytometry datasets (\cref{sec:flow-cytometry}).
		The solid lines show $\rho$ vs.\ F-measure for datasets 1--6.
		The black dashed line indicates averaged F-measure over the training datasets.
		The vertical black line shows the value $\rho = 1.16$ that maximizes the averaged F-measure.}
	\label{fig:flowcyt-train}
\end{figure}

\begin{table}[tp]
	\centering
	\def~{\hphantom{0}}
	\caption{F-measures on flow cytometry test datasets 7--12} {
		\begin{tabular}{ccccccc}
			\hline
			                                                   & 7             & 8             & 9             & 10            & 11            & 12            \\ \hline
			ACDC                                 & 0.63          & {0.92} & {0.94} & {0.99} & {0.99} & {0.98} \\
			Coarsened & {0.67} & 0.88          & {0.93} & {0.99} & {0.99} & {0.99} \\ 
            Elbow                               & 0.91          & 0.91          & 0.93          & 0.99          & 0.99          & 0.98          \\ 
			Gap                                 & 0.73          & 0.91          & 0.93          & 0.99          & 0.98          & 0.98          \\ 
			Silhouette                          & 0.63          & 0.74          & 0.94          & 0.99          & 0.96          & 0.98          \\ \hline
		\end{tabular}}
	\label{tab:flowcyto-test}
\end{table}

For test datasets 7--12, we pick the value of $\numcomps$ based on a value of $\rho$ that is as close as possible to the estimated $\rho$ value of $1.16$ while also being stable for a range of $\rho$ values.
As shown in \cref{tab:flowcyto-test}, \methodname provides the same average accuracy as the coarsened posterior
while being substantially more computationally efficient, despite using a much slower programming language for implementation (\methodname took 2 hours using Python while the coarsened posterior took 30 hours using Julia). 
\methodname performs comparably to the nonparametric methods, although 
the elbow method performs the best, particularly on dataset 7. 

\subsection{Cell Type Discovery using Single-Cell RNA Sequencing Data} 
\label{sec:scRNA}

There are two limitations to our flow cytometry experiments: the cell types are based on manual labeling 
(which maybe be incorrect) and the data is low dimensional. 
To more fully test the capabilities of \methodname on high-dimensional datasets with reliable ground truth, we consider
the common task of identifying cell types from single-cell RNA sequencing (scRNA-seq) data. 
Current practice relies heavily on heuristics to determine the number of cell types, which are known
to be inaccurate \citep{grabski}.

We evaluate \methodname and the alternative methods on the Tabula Muris dataset \citep{mice}, 
which profiles nearly 100,000 mouse cells across 20 tissues with curated ground-truth cell types.
We applied all methods to 80 data subsets with equally sized clusters (see \cref{tab:subsampled-datasets}).
All datasets undergo the same preprocessing steps described in \cref{sec:rna-data}.
Following common practice, we use log-transformed gene expression counts 
and fit a Gaussian mixture model. 
The model is clearly misspecified, motivating the need for a robust clustering approach.
The cell clustering performance is assessed by (1) the accuracy of estimating the true number of cell types and (2) the agreement between estimated and true labels, quantified using using the Adjusted Rand Index (ARI) and Adjusted Mutual Information (AMI) \citep{ari,ami}.

\begin{figure}[t]
	\centering
	\subfloat{\includegraphics[width=.55\textwidth]{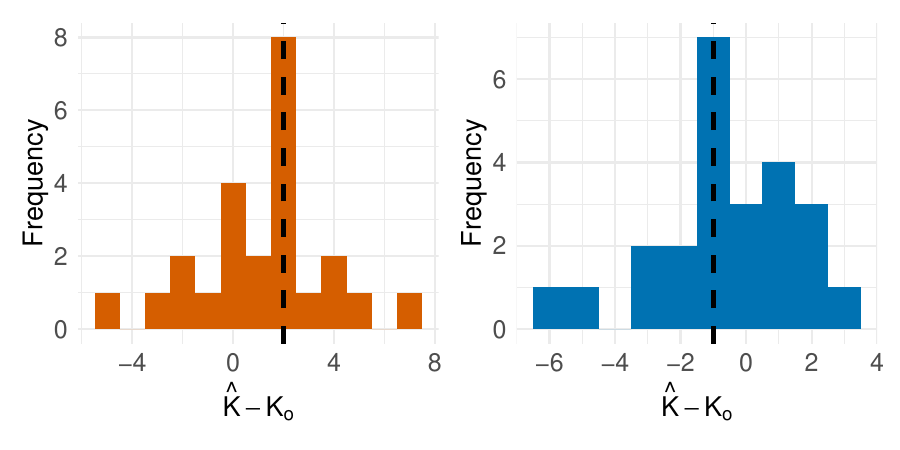}}
	\subfloat{\includegraphics[width=.44\textwidth]{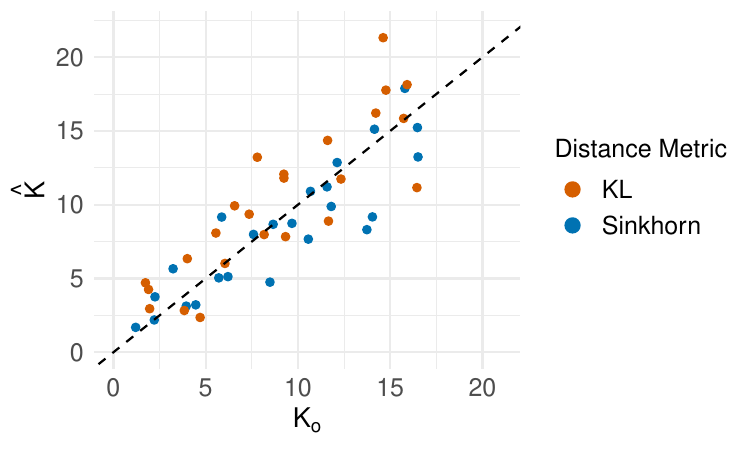}}

	\caption{Comparison of KL to Sinkhorn as divergence estimators for scRNA-seq clustering. Difference in $K$ estimates using \methodname with KL divergence \textbf{(left)} 
    and Wasserstein distance \textbf{(middle}.
    The dashed vertical lines indicate the medians.
	\textbf{Right:} True number of cell types vs. estimated number of cell types (jittered).
	}
	\label{fig:sh_kl_comp}
\end{figure}
\paragraph{Discrepancy Choice.} 
One benefit of \methodname is that the user can select a discrepancy measure $\mcD$ 
that best captures the underlying structure of their data.
While the KL divergence is our recommended default, the Wasserstein distance can better align with the underlying metric structure of the data.
In single-cell clustering, the cosine distance captures directional relationships between gene expression profiles.
This is particularly useful with high-dimensional scRNA-seq data where the magnitude of total gene counts in each cell can vary significantly, while the expression proportions across genes remain informative. 
We use the unbalanced Sinkhorn distance to estimate the Wasserstein distance (see \cref{sec:case-study-details} for details).
As illustrated in \cref{fig:sh_kl_comp}, both choices of discrepancy lead to good estimates of the true number of cell types, with the KL version having a slight bias toward overestimation (median difference of $2$ from the ground-truth $K_o$) 
and the Wasserstein version showing a slight bias toward underestimation (median difference of $-1$).
These results demonstrate the robustness of our default choice (KL divergence), 
while also highlighting the benefits of being able to improve performance by making a problem-specific choice (the Wasserstein distance).
For consistency with our other experiments,
we continue to use the KL divergence for the remaining experiments. 

\paragraph{Automation.}
To facilitate large-scale analyses and reduce manual intervention, we also assess the automated selection of the regularization parameter $\rho$ using the stability-based procedure introduced in \cref{sec:choosing-rho}.
As shown in Fig.~\ref{fig:auto_comp_unif},
the manually selected $K$ achieves slightly better ARI values
(ranging from $-0.27$ to $0.37$), while automated $K$ selection results in higher AMI values
(ranging from $-0.19$ to $0.33$).
To assess whether the differences are significant, we conduct paired t-tests on the AMI and ARI scores. For AMI, the mean
difference between the two methods is $0.0027$  (95\% CI: $[-0.0095, 0.0148])$.
For ARI, the mean difference is $0.0084$ (95\% CI: $[-0.0069, 0.0237])$.
In terms of estimation accuracy,  manual tuning of $\rho$ yields slightly lower mean absolute error (MAE $=2.16$ vs.\ $2.43$) and 0--1 loss ($0.84$ vs.\ $0.88$) compared to the automated approach.
The mean difference in absolute error between the manual and the automated procedures ($|\hat{K}_{manual} -\hat{K}_{auto} |$) is $-0.26$ (95\% CI: [$-0.57$, $0.05$]).
Overall, the difference between the automated and manual approaches -- in terms of both clustering accuracy and estimation accuracy -- is minimal. 

\begin{figure}[tp]
	\centering
    \includegraphics[width=0.8\textwidth]{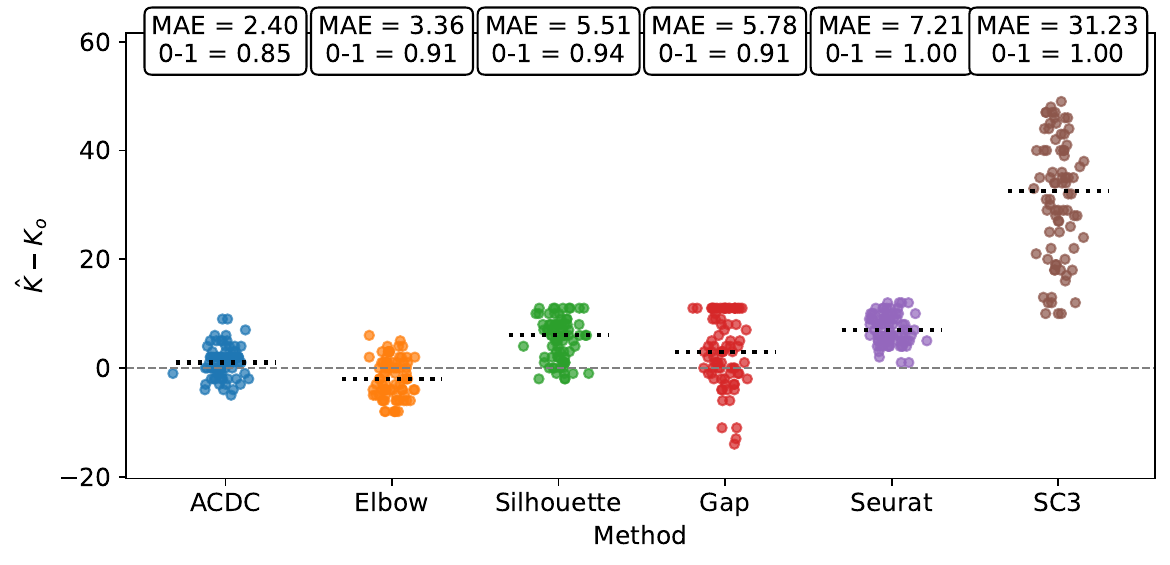}
	\caption{Comparison of ACDC with common model selection criteria (Elbow, Gap, and Silhouette) and specialized tools for scRNA-seq clustering (Seurat and SC3) on 80 scRNA-seq datasets. Each point shows the deviation between estimated and true numbers of cell types ($\hat K - K_o$). Dotted black lines indicate the median error for each method. Mean absolute error (MAE) and 0--1 loss quantify the model selection accuracy.}
    \label{fig:rna-seq-mm-select}
\end{figure}

\paragraph{Comparison to other methods.} 
There are a number of specialized tools for scRNA-seq clustering. 
So, in addition to comparing \methodname to Elbow, Gap, and Silhouette (all applied to the same parameter estimates), we also benchmark 
two widely used scRNA-seq clustering pipelines, Seurat  \citep{seurat} and SC3  \citep[Single-Cell Consensus Clustering;][]{sc3}; 
see \cref{sec:existing-tools} for further details about these tools. 
As shown in \cref{fig:rna-seq-mm-select}, 
\methodname achieves the best cluster estimation accuracy (MAE $=2.40$, 0–1 loss $=0.85$).
To quantify the difference in estimation accuracy, we conduct paired t-tests on the absolute estimation errors between \methodname and each competing method.
The mean differences in absolute error between \methodname and competing methods are $-0.96$ for Elbow (95\% CI: [$-1.59$, $-0.33$]), 
$-3.11$ for Silhouette (95\% CI: [$-3.97$, $-2.26$]), 
$-3.38$ for Gap (95\% CI: [$-4.41$, $-2.34$]), 
$-4.81$ for Seurat (95\% CI: [$-5.50$, $-4.12$]), 
and $-28.83$ for SC3 (95\% CI: [$-31.28$, $-26.37$]), indicating \methodname outperforms all other methods. 
In terms of estimation bias, the gap (median $=3.0$) and silhouette (median $=6.0$) criteria exhibit a stronger tendency to overestimate the number of clusters, while the elbow method (median $=-2.0$) tends to underestimate. All three methods show median deviations farther from zero than those of \methodname (median $=1.0$).
Seurat (median $=7.0$) produces consistently inflated estimates and SC3 (median $=32.5$) substantially overshoots $K_o$ in all datasets. 
In terms of clustering agreement, \methodname remains highly competitive, achieving ARI and AMI values comparable to Seurat (see \cref{fig:rna-seq-mm-ami-ari}).

\begin{figure}[h!]
	\centering
	\includegraphics[width=.8\textwidth]{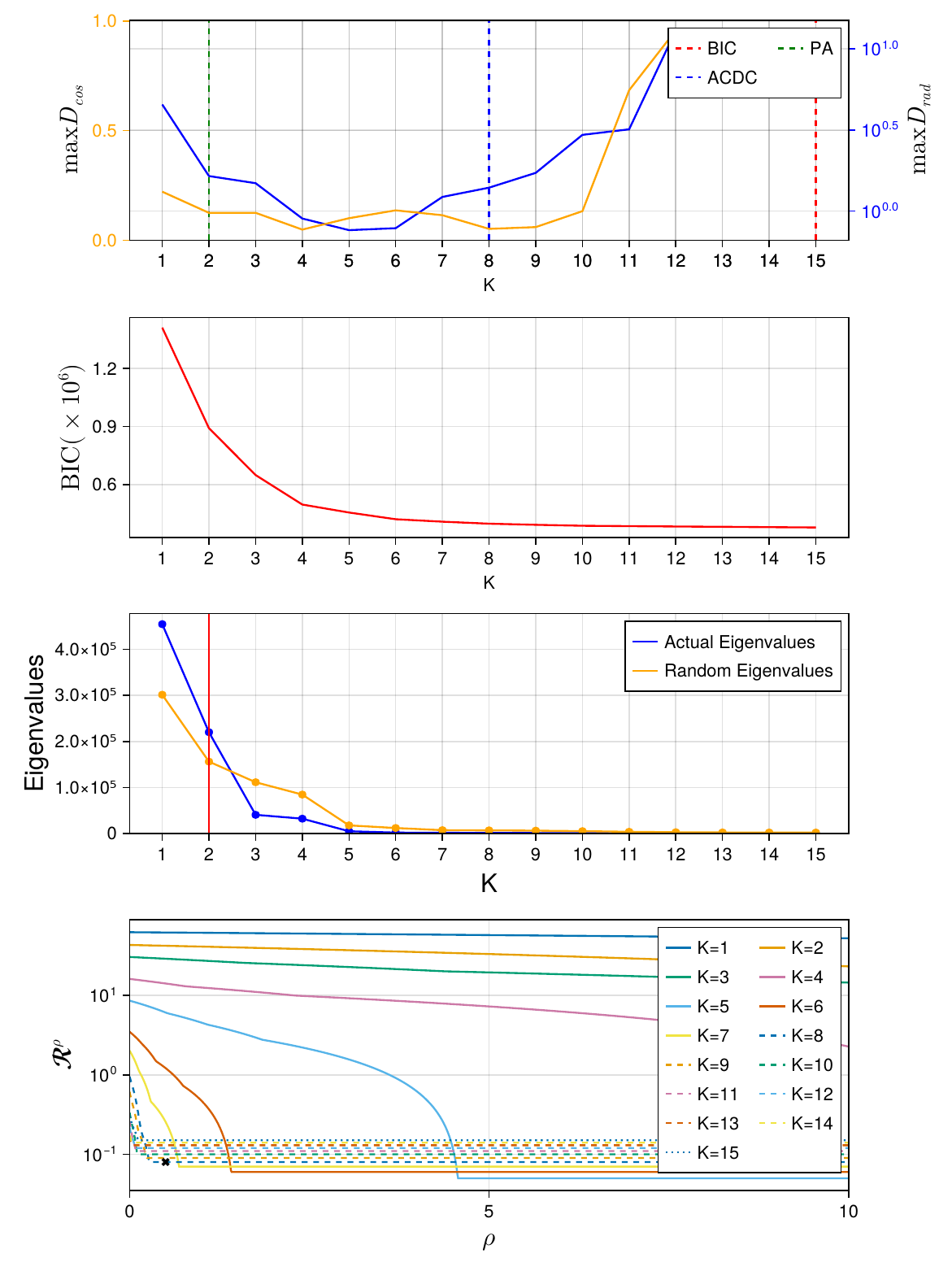}
	\caption{
	Estimation quality of mutational signature discovery for
	perturbed synthetic breast cancer data (\cref{sec:mutsigs}).
	\textbf{(top)} Errors of signature and loadings estimates.
	\textbf{(middle top)} $K$ versus value of BIC. BIC selects $\widehat K=K_{\max}$.
	\textbf{(middle bottom)} Scree plot generated from the dataset, indicating that $K = 2$ is the optimal choice of $K$.
	\textbf{(bottom)} Structurally aware loss for $K\in \{1,\dots,K\}$, the wide region with the smallest $\rho$ is marked by the cross mark, indicating $K=7$ or $K=8$ is the most appropriate choice of $K$.
	}
	\label{fig:mutsig_result}
\end{figure}

\section{Application to Probabilistic Matrix Factorization} \label{sec:pmf-applications}

Next, we apply \methodname to probabilistic matrix factorization problems. 
As with mixture models, \methodname is robustly consistent: 
\begin{theorem} \label{thm:pmf-robust-consistency}
\methodname using $\compDiscEst^{(K,k)}$ defined in \cref{eq:general-Dcomp} is $\kappa$-robustly consistent for probabilistic matrix factorization if the underlying component discrepancy $\mcD$ is the KL divergence,  $\distDisc = d_\mathrm{BL}$, and $\kappa(\rho, K) = \sqrt{K\rho/2}$. 
\end{theorem}

In our numerical experiments, we take $G = \Unif([0,1]^{D})$, the uniform distribution on the $D$-dimensional hypercube.
This choice is without loss of generality when using KL divergence as the discrepancy measure since it is invariant to diffeomorphisms of the noise variables ${\veps}_{nk}$.
Moreover, it leads to a universal choice of $f(z, \phi, \cdot) = F^{-1}_{z\phi}$, the inverse CDF.
However, in specific scenarios other choices for $G$ and $f$ might be preferred due to considerations such as ease of implementation or stability of KL divergence estimation. 
For example, in the Gaussian case, we could take $G = \Norm({0}, I)$ and $f(z, \phi, \veps) = z\phi + \sigma \veps$. 

In addition to BIC, we compare \methodname to \emph{parallel analysis} (PA) \citep{Horn_ParallelAnalysis_1965,Buja_RemarksParallelAnalysis_1992}, which is a 
commonly used nonparametric method for model selection for probabilistic matrix factorization. 
Parallel analysis is carried out by generating the scree plot (ordered PCA eigenvalues) of the data against that of randomly generated matrices of the same size. These random matrices are generated by independently permuting each row of the data matrix. 
The results of \citet{Dobriban_PA-for-FA_2020} show that, in general, PA can be conservative 
and miss factors with a low signal-to-noise ratio. 
For BIC, to account for permutation invariance of the parameters, we use the formula
\[
\operatorname{BIC}\Big(x_{1:N},z^{(K)}_{1:K,1:N},\phi^{(K)}_{1:K}\Big) = K\log(N) - 2\log\left\{p\big(x_{1:N}~\big|~ z^{(K)}_{1:K,1:N},\phi^{(K)}_{1:K}\big)\right\}+2\log(K!).
\]

While other robust model selection approaches for matrix factorization exist, they all have limitations
that lead us to not include them in our empirical comparison. 
The method of \citet{Liu_Support-Union_2019} is limited to Gaussian nonnegative matrix factorization (NMF).
\citet{Pelizzola_NegBin-NMF_2023} aim to address the problem of robustness for the case of 
Poisson NMF using two different approaches: a negative binomial instead of a Poisson likelihood to improve the model's ability to handle overdispersed data, and a testing routine inspired by cross validation.
While this approach provides reasonable results, using the negative binomial only targets a very specific type of data--model mismatch, and the cross validation approach does not have any correctness guarantees.
\citet{Bai_DeterminingNumberFactors_2002} propose an information criterion-based approach for factor analysis
and provide an asymptotic consistency result for the case where the input dimension tends towards infinity. 
However, their main result
applies only to the Gaussian NMF model with principle component analysis (PCA) as the estimation method.
Particularly in NMF applications, another common approach is to evaluate the stability of the 
NMF solution across multiple runs. 
The \emph{cophenetic correlation coefficient} \citep{brunet_CCC_2004a} is one such example. 
A similar stability-based principle is adopted by the widely-used toolset SigProfilerExtractor \citep{islam_sigprofiler_extractor_2022}, which uses a consensus bootstrap approach to improve stability. 
However, these approaches are computationally costly due to its reliance on repeated NMF executions. Furthermore, the results of \citet{Xue:2024} demonstrate empirically that the consensus bootstrap methodology -- much like PA -- tends to be conservative, underestimating the true number of factors. 

\subsection{Mutational Process Discovery} \label{sec:mutsigs}

Exposure to, and presence of, carcinogenic processes such as UV radiation, tobacco smoke, defective DNA repair mechanisms, and naturally occurring biochemical reactions, generate characteristic patterns of somatic mutations known as \emph{mutational signatures} \citep{Alexandrov_mut-sig-NMF_2013,Nik-zainal_MutationalProcessesMolding_2012}.
Mutational signature-based analyses have contributed to novel insights in cancer research \citep[e.g.,][]{Alexandrov_mut-sig-NMF_2013,Nik-zainal_MutationalProcessesMolding_2012,pcawgSV:2020,pcawg2020} and are leading to emerging translations in clinical settings \citep[e.g.,][]{Chakravarty:2021}.
The most widely used approach to signature discovery is to fit a Poisson non-negative matrix factorization (NMF) model. %
For the $n$th tumor sample, the data consist of a count vector ${x}_{n}\in\nats^{D}$,
where $D$ is the number of mutation types being considered (e.g., there are $D=96$ single-base substitution types with a trinucleotide context).
The number of mutations of type $d$ in sample $n$ due to mutational process $k$ is given by
$y^{(K)}_{nkd} \sim \Poiss({\phi}^{(K)}_{kd}z^{(K)}_{nk})$ and hence the total number of mutations of type $d$ in sample $n$ is $x_{nd} = \sum_{k=1}^K y^{(K)}_{nkd}$.

Because it is nearly impossible to obtain ground-truth signatures for real data, we use
simulated breast cancer data
based on the COSMIC v2 catalog and the pan-cancer analysis of whole genomes (PCAWG),
following the procedure of \citet{Xue:2024}.
See \cref{sec:mutsigs-details} for experimental details.

As is standard practice, we quantify the signature recovery error using the cosine difference $
	D_{\mathrm{cos}}({\phi},{\phi}_{\star})=1-\langle {\widetilde\phi},{\widetilde\phi}_{\star} \rangle$, where for any vector ${v}$, we define ${\widetilde v} = {v} / \norm{{v}}_{2}$.
For the exposures, we quantify the error using the relative average difference
$
	D_{\mathrm{rad}}({z}, {z}_{\star})= {\abs{\overline{z}-\overline{z}_{\star}}}/{\overline{z}_{\star}},
$
where for a vector ${v} \in \reals^N$, $\overline{v} = N^{-1}\sum_{n=1}^N v_n$.
To evaluate the quality of an estimate as a whole, we perform bipartite matching against the ground truth by minimizing the metric
\[
	D^{(K)}(\sigma)=\sum^{K}_{k=1}\lrb{D_{\mathrm{cos}}\lrp{{\phi}^{(K)}_{k},{\phi}_{o\sigma(k)}}
		+0.1\, \tanh\, D_{\mathrm{rad}}\lrp{{z}^{(K)}_{k},{z}_{o\sigma(k)}}},
\]
where $\sigma \colon [K_o] \to [K]$ denotes an injective matching function.
We bound and down-weight $D_{\mathrm{rad}}$ using the $0.1 \tanh$ transform so that signature 
reconstruction accuracy is the main determinant for matching, with exposure accuracy acting a
tiebreaker when the signature accuracies are similar. 
Given the optimal matching $\sigma_{\star} = \argmin_\sigma D^{(K)}(\sigma)$, the accuracy scores are defined
as the worst-case cosine and relative average errors, given by
$
	L^{(K)}_{{\phi}} = \max\{D_{\mathrm{cos}}({\phi}^{(K)}_{k},{\phi}_{o\sigma_{\star}(k)}) : k=1,\dots,K \},
$
and
$
	L^{(K)}_{z} = \max\{D_{\mathrm{rad}}(\overline{z}^{(K)}_{k},\overline{z}_{o\sigma_{\star}(k)}) :  k=1,\dots,K\}.
$

\Cref{fig:mutsig_result,fig:mutsig_result_appendix_1} show that, across all four datasets,
BIC selects $\widehat{K} = K_{\max}$ -- even when the data is well specified.
On the other hand, PA consistently underestimates $K_o$, estimating $\hat K = 2$. Finally, \methodname selects $\widehat K = 7$ or $8$, which correspond to the some of the largest values of $K$ for which the parameter estimates still have reasonably small error and meaningful decomposition.
These results suggest \methodname is a promising alternative to existing approaches for selecting the number of mutational signatures, which all suffer from some combination of high computational cost and lack of statistical rigor \citep{Xue:2024,pcawg2020}.

\subsection{Materials Discovery using Hyperspectral Imaging} \label{sec:hyperspectral}

Hyperspectral remote sensing data is used in applications such as environmental monitoring and city planning \citep{Brook_Dust_over_Green_Canopy_2016,Ji_Estimatng_Vegetation_Fractional_Cover_2016,Lin_RetrievingHydrousMinerals_2017}.
Hyperspectral data is collected as an image in which each pixel specifies the intensity of light at each observed wavelength.
However, due to low spatial resolution, each pixel can be a mixture of materials,
each reflecting different amounts of light at each wavelength.
\emph{Hyperspectral unmixing} refers to the unsupervised extraction of spectral signatures corresponding to materials (called \emph{end-members}) and the abundances of these materials from each pixel.
While it is reasonable to assume the intensities are observed with Gaussian noise, the contributions from each material obviously cannot be negative.
However, incorporating this non-negativity constraint into the model is non-trivial (e.g., if using off-the-shelf parameter estimation methods),
and so the constraint is often ignored.
Thus, to illustrate the benefits of our approach in enabling principled model selection
in a setting with clear misspecification,
we will use a Gaussian factor analysis model.
We apply the model to a $307\times307$ hyperspectral image of an urban area,
with each pixel representing a plot 2m $\times$ 2m in size.\footnote{\url{https://rslab.ut.ac.ir/data}}
After discarding certain wavelengths due to dense water vapor and atmospheric effects, each pixel consists of 162 channels with wavelengths ranging from $400$nm to $2500$nm.
There are three versions of ground truth, containing 4, 5, and 6 end-members respectively.
The 6 end-member version includes an additional material named ``metal'', which, through manual inspection, was found to contribute to only a small part of the image.
As a result, none of the NMF algorithms we tested accurately recovered this end-member.
Therefore, we use the ground truth with 5 end-members for this experiment (visualized in \cref{fig:hyprunmix_gt}(a))

\begin{figure}[t!]
	\centering
	\includegraphics[width=.8\textwidth]{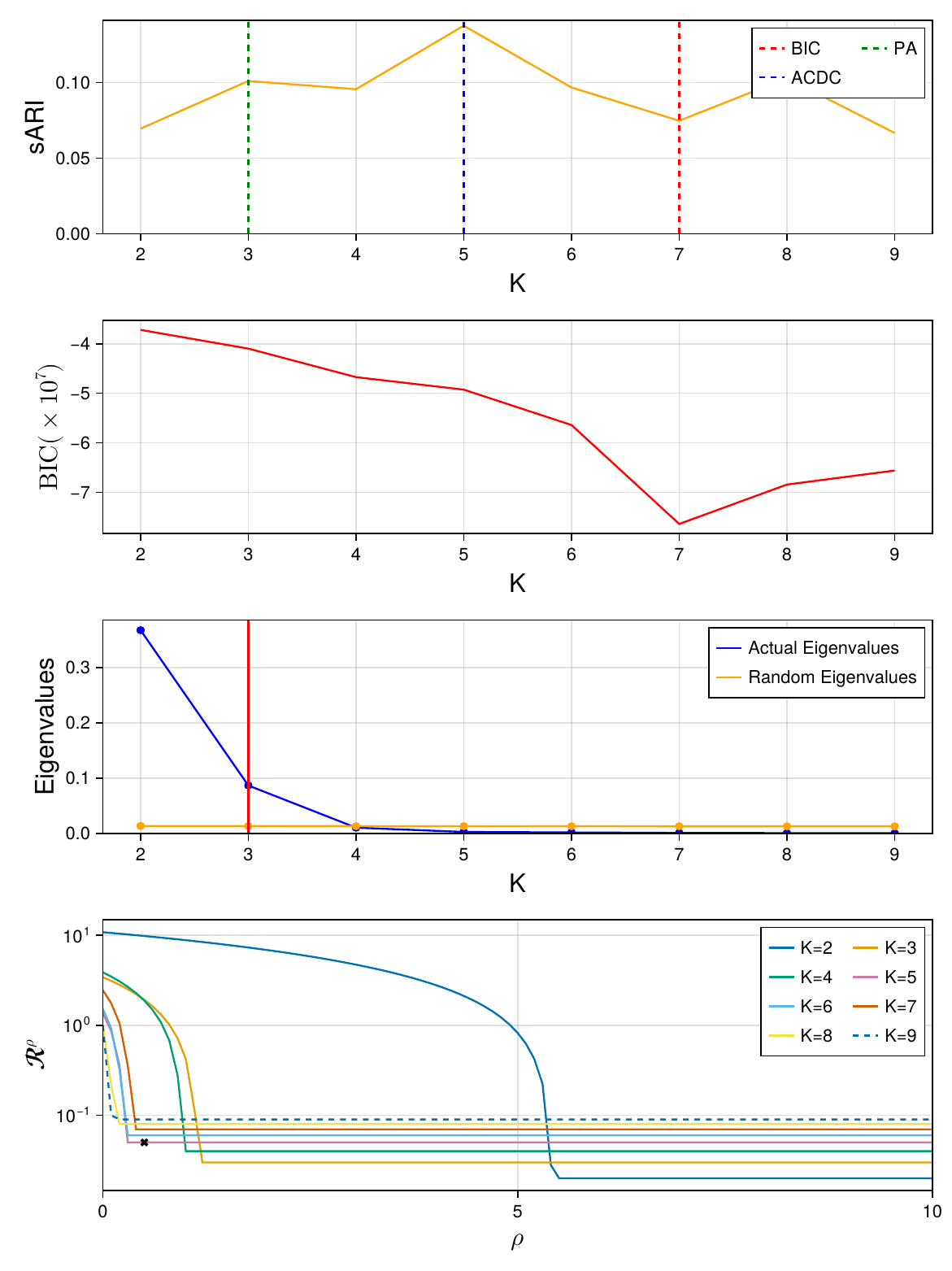}
	\label{fig:urban_results_rho_k}
	\caption{
		Model selection of end-members for hyperspectral urban dataset (\cref{sec:hyperspectral}).
		\textbf{(top)} Quality of solution measured by sARI, showing $K=5$ is the solution with  material distribution that best matches the ground truth.
		\textbf{(middle top)} Quality of solution ranked using BIC, indicating $K=7$ is the best fitting solution.
        \textbf{(middle bottom)} Scree plot generated from the dataset, indicating that $K=3$ is the optimal choice of $K$.
		\textbf{(bottom)} Structurally aware loss for $K=2,\dots,9$, the wide region with the smallest $\rho$ is marked by the cross mark, indicating $K=5$ is the most appropriate choice of $K$.
	}
	\label{fig:urban_results}
\end{figure}

We obtain point estimates of the factorization using a modified implementation of the coordinate descent method \citep{Cichocki-Phan_coorddesc_2009}.
We judge the quality of the estimates by quantifying how close the inferred material abundances are
to the ground truth
using the soft adjusted Rand index \citep[sARI;][]{Flynt_sARI_2019}.
\Cref{fig:urban_results} shows that \methodname selects $\widehat{K} = K_{o} = 5$,
which also maximizes the sARI. %
BIC overfits, selecting $\widehat{K} = 7$, while PA underfits, selecting $\widehat{K} = 3$.
\Cref{fig:hyprunmix_gt}(b)--(e) provides a qualitative understanding of the result.
The small gap between $K=3$ and $K=4$ can be explained through the addition of the ``grass'' material that is
very similar to the already included ``tree'' material. 
The large gap between $K=4$ and $K=5$ can be explained by the addition of the very different ``asphalt'' material. 
Finally, $K=6$ only adds a ``residual'' material compared to $K=5$, which can be interpreted as a different shade of ``grass''
and is clearly an artifact of overfitting.

\section{Discussion} \label{sec:discussion}

In this paper, we have developed a general theoretical and methodological framework for ensuring 
mechanistic interpretability when selecting the number of latent components in a latent variable model. 
As two applications, we showed that our \methodname method is robustly consistent and empirically 
effective for mixture models and probababilistic matrix factorization. 
These results suggest \methodname is a promising general-purpose method for robust model selection. 

Our work also opens up many avenues for future research. 
In applications where related labeled datasets are not available, we are only able to provide a heuristic method for calibrating the degree of misspecification, as quantified by $\rho$.
Ideally, we would like to have a more rigorous criteria.
However, given the nonparametric nature of misspecification, we suspect that a fully general solution does not exist; for example, the coarsened posterior similarly requires a heuristic calibration step.
One alternative to explore in the future is the simulation-based calibration method from \citet{Xue:2024}, which was developed for the coarsened posterior but could easily be used with \methodname as well.
However, generally speaking, a user must have \emph{some} prior knowledge about the degree of model misspecification -- and believe that the misspecification will be reasonably small as measured by the chosen discrepancy $\mcD$.
Moreover, if the degree of misspecification is very large, we should not expect any robust model selection procedure to
work well (cf.\ $P_{o}^{E}$ in \cref{fig:consistency-illustration}).

It would be valuable to further develop our robust consistency theory; for example  
similar results could be proved for other common model classes. 
It would also be useful to characterize how quickly $\Pr(\widehat{K} = K_o)$ converges to 1 
and to quantify the variability of $\widehat{K}$.
Another interesting direction for future work is to apply \methodname to other model classes such as supervised factor analysis and extend it to apply to models outside of the framework we developed in
\cref{sec:framework} -- for example, (nonlinear) variational autoencoders \citep{Kingma2014,Kingma2019VAEs} and semiparametric matrix factorization models \citep{Anandkumar2014uc,rohe2023vintage-7f4}.

%% file: appendix.tex
\counterwithin{equation}{section}
\counterwithin{figure}{section}
\renewcommand{\thefigure}{\Alph{section}.\arabic{figure}}
\counterwithin{table}{section}
\renewcommand{\thetable}{\Alph{section}.\arabic{table}}
\counterwithin{theorem}{section}
\renewcommand{\thetheorem}{\Alph{section}.\arabic{theorem}}
\counterwithin{lemma}{section}
\renewcommand{\thelemma}{\Alph{section}.\arabic{lemma}}
\counterwithin{assumption}{section}
\renewcommand{\theassumption}{\Alph{section}.\arabic{assumption}}
\counterwithin{definition}{section}
\renewcommand{\thedefinition}{\Alph{section}.\arabic{definition}}
\counterwithin{proposition}{section}
\renewcommand{\theproposition}{\Alph{section}.\arabic{proposition}}

\begin{center}
    \LARGE \textbf{Supplementary Materials}
\end{center}

\section{Robust Consistency of \methodname} 
\label{sec:theory}

We now show that, under reasonable assumptions, \methodname is robustly consistent.
We first cover the mixture modeling case, then probabilistic matrix factorization.

\subsection{General Theory for Mixture Modeling}

We will develop a general theory, then discuss how \cref{thm:mixture-model-robust-consistency} follows 
as a corollary.
Proofs of all mixture model-related results are in \cref{sec:mixture-model-consistency}. 
First, we consider the requirements for the distribution-level discrepancy $\distDisc(\cdot, \cdot)$
(assumed to be a metric), the discrepancy $\discr{\cdot}{\cdot}$ used to construct the component-level
discrepancy, and the discrepancy estimator $\discrest{\cdot}{\cdot}$
-- noting that sometimes it will be possible to take $\discrest{\cdot}{\cdot} = \discr{\cdot}{\cdot}$.
\begin{assumption} 	\label{assump:metric-discr-conditions}
	For $y_{\numobs n} \in \mathcal{X}$ $(\numobs = 1,2,\dots; n = 1,\dots,\numobs)$,
	define the empirical distribution $\widehat{P}_{\numobs} = \numobs^{-1}\sum_{n=1}^{\numobs}\delta_{y_{N,n}}$
	and assume $\widehat{P}_{\numobs} \rightarrow P$ in distribution.
	The distribution-level and component-level discrepancies satisfy the following conditions:
	\begin{enumerate}[label=\textnormal{(\alph*)}]
		\item \emph{The distribution-level discrepancy metric detects empirical convergence:} $\distDisc(\widehat{P}_\numobs, P) \to 0$ as $\numobs \to \infty$.
		\item  \emph{The distribution-level discrepancy metric is jointly convex in its arguments:} for all $w \in (0, 1)$ and distributions $P$, $P'$, $Q$, $Q'$,
		      \[
			      \distDisc(wP + (1 - w)P', wQ + (1 - w)Q') \le w\,\distDisc(P, Q) + (1 - w)\,\distDisc(P', Q').
		      \]
		\item  \emph{The discrepancy estimator is consistent:} For any distributions $P, Q$,
		      if $\discr{P}{Q} < \infty$ and $\widehat P_{N} \to P$ in distribution,
		      then $\discrest{\widehat{P}_{\numobs}}{Q} \to \discr{P}{Q}$ as $\numobs \to \infty$.
		\item \emph{Smoothness of the discrepancy estimator:} The map $\param \mapsto \discrest{\widehat{P}_{\numobs}}{F_{\param}}$ is continuous.
		\item  \emph{The discrepancy bounds the metric:} There exists a continuous,
		      non-decreasing function $\widetilde{\kappa}: \reals \to \reals$ such that
		      $\distDisc(P,Q) \leq \widetilde{\kappa}(\discr{P}{Q})$ for all distributions $P, Q$.
		\item \emph{The distribution-level discrepancy metric between components is finite:} For all $\param, \param' \in \Theta$,
		      it holds that $\distDisc(F_{\param}, F_{\param'}) < \infty$.
	\end{enumerate}
\end{assumption}

\begin{remark}[Discussion of \cref{assump:metric-discr-conditions}] A wide variety of metrics satisfy \cref{assump:metric-discr-conditions}(a), including the bounded Lipschitz metric,
the Kolmogorov metric, maximum mean discrepancies with sufficiently regular bounded kernels, and the Wasserstein metric with a bounded cost function \citep{Vaart:1996,Sriperumbudur:2010,SimonGabriel:2018,Villani:2009}.
\Cref{assump:metric-discr-conditions}(b) also holds for a range of metrics.
For example, it is easy to show that all integral probability metrics -- which includes the bounded Lipschitz
metric, maximum mean discrepancy, and 1-Wasserstein distance -- are jointly convex (see \cref{lem:ipm-joint-convex}).
\cref{assump:metric-discr-conditions}(c) is a natural requirement that the limiting divergence can be estimated consistently.
Such estimators are well-studied for many common discrepancies.
\cref{assump:metric-discr-conditions}(d) will typically hold as long as the map $\phi \mapsto F_{\phi}$ is
well-behaved.
For example, for the Kullback--Leibler divergence estimators described in \cref{sec:kl-estimation}
and standard maximum mean discrepancy estimators \citep{Gretton:2012,Krause:2023},
when $F_{\phi}$ admits a density $f_{\phi}$, it suffices for the map $\phi \mapsto f_{\phi}(x)$
to be continuous for $P_{o}$-almost every $x$.
\cref{assump:metric-discr-conditions}(e) is not overly restrictive and we discuss some relevant examples in \cref{sec:divergence-choice} below.
\cref{assump:metric-discr-conditions}(f) trivially holds for bounded metrics such as the bounded Lipschitz metric and integral probability measures with uniformly bounded test functions.
\end{remark}

Our second assumption requires the parameter estimation procedure to be sufficiently well-behaved,
in the sense that, for each fixed number of components $\numcomps \le \numcomps_{o}$, it should consistently estimate an asymptotic parameter $\allparam_{\star}^{(K)}$.
We do not make any explicit assumptions that such parameters are, in any sense, ``optimal.''

\begin{assumption}\label{assump:inference-regularity}
	For each $\numcomps \in \{1,\dots, \numcomps_{o}\}$, there exists $\theta_{\star}^{(K)} \in \paramSpace{K}$ such that
	$\widehat\theta^{(K)} \to \theta_{\star}^{(K)}$ in probability as $N \rightarrow \infty$,
	after possibly reordering of components.
\end{assumption}

\begin{remark}[Discussion of \cref{assump:inference-regularity}]
\Cref{assump:inference-regularity} holds for most reasonable algorithms, including expectation--maximization, point estimates based on the posterior distribution,
and variational inference \citep{Balakrishnan:2017,Walker:2001,Wang:2019}.
Note that we assume that consistency holds for parameters in the equivalence class induced by component reordering, although we keep
this equivalence implicit.
\Cref{assump:inference-regularity} implies that the empirical data distribution of the $k$th component,
$\widehat{F}^{(K)}_k$, %
converges to a limiting distribution
\begin{align}
	\widetilde{F}^{(K)}_{k}(\dee x) & = \frac{p_{\star}^{(K)}(k\mid x)P_o(\dee x) }{\int p_{\star}^{(K)}(k\mid y)P_o(\dee y)}, \label{eq:phat-def}
\end{align}
where $p_{\star}^{(K)}(k \mid x) = \eta^{(K)}_{\star k}f_{\star k}^{(K)}(x)/p^{(K)}_{\star}(x)$ is the
conditional component probability under the limiting model distribution
 and $f_{\star k}^{(K)}$ denotes the density of $F_{\param^{(K)}_{\star k}}$.
\end{remark}

We can now state our main result about the robust consistency of \methodname for mixture modeling. 
\begin{theorem} \label{thm:main}
	For mixture modeling, if \cref{assump:metric-discr-conditions,assump:inference-regularity} hold, then 
    \methodname is $\kappa$-robustly consistent for $\kappa(\rho, K) = \widetilde{\kappa}(\rho)$.
\end{theorem}
\begin{proof}[Proof idea] To prove \cref{thm:main} we establish two facts.
	First, we show that the loss satisfies
	\[
		\Pr\{\mcR^{\rho}(\allparam^{(\numcomps_{o})}; z^{(\numcomps_{o})}, \data{1:\numobs}) = 0 \}   \overset{N \to \infty}{\longrightarrow}  1.
	\]
	Second, we verify that for $\numcomps < \numcomps_o$,
	$
		\mcR^{\rho}(\allparam^{(\numcomps)}; z^{(\numcomps)}, \data{1:\numobs}) \to \infty
	$
	in probability for $N \to \infty$.
	Therefore, for $\numcomps > \numcomps_o$, asymptotically
	$\mcR^{\rho}(\allparam^{(\numcomps)}; z^{(\numcomps)}, \data{1:\numobs})  \ge \mcR^{\rho}(\allparam^{(\numcomps_{o})}; z^{(\numcomps_{o})}, \data{1:\numobs})$, so the minimum is asymptotically attained at $\widehat{\numcomps}=\numcomps_o$.
\end{proof}

%
%
%
%
%
%
%

%
%
%
%
%
%
%
%
%
%
%
%
%
%
%
%
%
%
%
%
%
%
%
%
%
%

\subsection{Results for Specific Divergences} \label{sec:divergence-choice}

We now verify \cref{thm:mixture-model-robust-consistency} by showing that  \cref{assump:metric-discr-conditions}(a,b,e) is satisfied when $\mcD$ is the  KL divergence, maximum mean discrepancy \citep[MMD; ][]{Sriperumbudur:2010}), or 1-Wasserstein distance \citep{Villani:2009}) -- with an appropriate choice of $\distDisc$ is made.
Before stating these results, we recall the following useful class of metrics on probability measures.
\begin{definition}[Integral probability metric]
	Given a collection $\mcH$ of real-valued functions on $\mcX$,
	for any probability measures $P$ and $Q$ on $\mcX$,
	the corresponding \emph{integral probability metric} is defined as
	\[ \label{eq:IPM}
		d_{\mcH}(P,Q) = \sup_{h \in \mcH}\left|\int h(x)P(\dee x) - \int h(y) Q(\dee y)\right|.
	\]
\end{definition}
\paragraph{Kullback--Leibler divergence.}
Assume that $\mcX$ is equipped with a metric $m$ and define the \emph{bounded Lipschitz norm}
$\BLnorm{h} = \Vert h \Vert_{\infty} + \Vert h \Vert_{L}$, where $\Vert{h}\Vert_{L} = \sup_{x \ne y}|h(x) - h(y)|/m(x, y)$
and $\Vert{h}\Vert_{\infty} = \sup_{x} |h(x)|$.
Letting $\mcH = \mcH_{\mathrm{BL}} = \{h : \Vert h \Vert_{\mathrm{BL}}\le 1\}$ gives the bounded Lipschitz metric
$\blmetric = d_{\mcH_{\mathrm{BL}}}$.
If $\mcD$ is the KL divergence, we can choose $d$ to be the bounded Lipschitz metric.
\begin{proposition} 	\label{coro:KL}
	If $\distDisc(P, Q) = \blmetric(P,Q)$ and $\mcD(P \mid Q) = \kl{P}{Q}$, then
	\cref{assump:metric-discr-conditions}(a,b,e) holds with $\widetilde\kappa(\rho) = (\rho/2)^{1/2}$.
\end{proposition}

\paragraph{Maximum mean discrepancy (MMD).}
Let $\mcK $ denote a kernel.
Denote the reproducing kernel Hilbert space with positive definite kernel $\mcK \colon \mcX \times \mcX \to \reals$ as $\mcH_{\mcK}$.
Denote the inner product of $\mcH_{\mcK}$ by $\langle{\cdot},{\cdot}\rangle_{\mcK}$ and the norm by $\|{\cdot}\|_{\mcK}$.
Letting $\mcH = \mcB_{\mcK} = \{ h \in \mcH_{\mcK} :  \|h\|_{\mcK} \le 1\}$, the unit ball, gives the maximum mean discrepancy $d_{\mathrm{MMD}} = d_{\mcB_{\mcK}}$.
If we choose the divergence to be an MMD, $\distDisc$ can be the same MMD.
\begin{proposition} 	\label{coro:MMD}
	If $\mcK$ is chosen such that $d_{\mathrm{MMD}}$ metrizes weak convergence, and
	$\distDisc(P, Q) =  \discr{P}{Q} = d_{\mathrm{MMD}}(P,Q)$,
	then \cref{assump:metric-discr-conditions}(a,b,e) holds with $\widetilde\kappa(\rho) = \rho$.
\end{proposition}

For conditions under which $d_{\mathrm{MMD}}$ metrizes weak convergence, see
\citet{Sriperumbudur:2010,Simon:2020}.

\paragraph{1-Wasserstein distance.}
Setting $\mathcal{H} = \mathcal{H}_{\text{Lip}} = \{ h : \| h \|_L \leq 1 \}$ gives the 1-Wasserstein distance $d_W = d_{\mathcal{H}_{\text{Lip}}}$.
As with the MMD, we can choose the divergence and $\distDisc$ to be the same:
\begin{proposition}
	\label{prop:wasserstein}
	If $\distDisc(P, Q) = \discr{P}{Q} = d_W(P, Q)$, then \cref{assump:metric-discr-conditions}(a,b,e) holds with $\kappa(\rho) = \rho$.
\end{proposition}

\subsection{Theory for Probabilistic Matrix Factorization}

Due to the more complex dependence structures in probabilistic matrix factorization, we focus 
on the case when $\mcD$ is the KL divergence, as stated informally in \cref{thm:pmf-robust-consistency}. 
For a given $K$, denote the empirical distribution of estimated coefficients by $\widehat{H}^{(K,N)} = N\inv\sum^{N}_{n=1}\delta_{\widehat{z}^{(K)}_{n}}$.
\begin{assumption}\label{assump:consistency}
	For each $K \in \{1,\dots,K_{o}\}$, there exists ${\phi}^{(K)}_{\star} \in\Phi^{(K)}$ and distribution $H^{(K)}_{\star}$ such that ${\widehat{\phi}}^{(K,N)} \convp {\phi}^{(K)}_{\star}$ and $\widehat{H}^{(K,N)}\convd H^{(K)}_{\star}$ as $N\to \infty$, after possible reordering of components.
\end{assumption}

\cref{assump:consistency} is analogous to \cref{assump:inference-regularity} in the mixture model case. This should also holds for most applications under mild conditions
\citep{Zhao_ConvergenceAnalysisMU_2017,Devarajan_NMFdualdivergence_2019,Fu_IdentifiabilityNMF_2018,Anderson_AsymptoticFA_1988,Anderson_AsymptoticPCA_1963}. 

Let $\frac{\d{y}}{\d{\veps}}(\blank,{\phi},z)$ be the Jacobian matrix of the mapping ${\veps}\mapsto f\lrp{z,{\phi},{\veps}}$,
and $\mathcal{E}({y},{\phi},z)=\lrc{{\veps}\mid f(z,{\phi},{\veps})={y})}$. 
Let
\[
Q({y_{nk}}\mid{\phi}_{k},z_{nk},\widetilde{G}_{k}) %
  = \int_{\mathcal{E}\lrp{{y_{nk}},{\phi_{k}},z_{nk}}}
	\widetilde{G}_{k}({\veps})
	\det\lrp{\frac{\d{{y}}}{\d{\veps}}\lrp{{\veps},{\phi_{k}},z_{nk}}}\inv
\]
denote the conditional distribution of ${y}^{(K)}_{nk}$ 
given $z^{(K)}_{nk}$, ${\phi}^{(K)}_{k}$ and, noise distribution $\widetilde{G}$,
and let 
\[
  Q^{(K)}\lrp{y_{n} \mid {\phi} , z_{n} , \widetilde{G}_{1:K}}=
  \prod_{k=1}^{K}Q({y}_{nk}\mid{\phi}_{k},z_{nk},\widetilde{G}_{k}). 
\]
We can now define the conditional distributions for the limiting model as 
\[
  Q^{(K)}_{\star k}({y_{nk}}\mid z_{nk}) = Q\lrp{{y_{nk}}\mid{\phi}^{(K)}_{\star k},z_{nk},G}
  \quad\text{and}\quad
  Q^{(K)}_{\star}({y}_{n}\mid z_{n})=\prod^{K}_{k=1}Q^{(K)}_{\star k}({y}_{nk}\mid z_{nk}),
\]
and the empirical conditional distributions as 
\[
  \widehat{Q}^{(K,N)}_{k}({y_{n}}\mid z_{n}) = Q\lrp{{y_{n}}\mid{\widehat{\phi}}^{(K,N)}_{k},z,\widehat{G}^{(K,N)}_{k}}
    \quad\text{and}\quad
  \widehat{Q}^{(K,N)}(y_{n}\mid z_{n}) =\prod^{K}_{k=1}\widehat{Q}^{(K,N)}_{k}({y}_{k}\mid z_{k}).
\]

\begin{assumption}\label{assump:sample_approx_able}
  The data-generating distribution and model satisfy the following conditions: 
	\begin{enumerate}[label={(\alph*)}]
\item For all $K \in \{1,\dots,K_{o}\}$, $k \in \{1,\dots,K\}$, and $z \in \mcZ$, 
the distributions $Q^{(K)}_{\star k}(\blank\mid z)$ has the same support on $\mathcal{Y}$.
    \item For all $K \in \{1,\dots,K_{o}\}$, $k \in \{1,\dots,K\}$, $z \in \mcZ$, and distribution $Q'$ on $\mcY$, 
      \[
        \lim_{{\phi}\to{\phi}^{(K)}_{k\star}}\widehat{\mathcal{D}}(Q', Q(\blank\mid{\phi}, z, G))
        = \widehat{\mathcal{D}}(Q', Q^{(K)}_{\star k}(\blank\mid z)).
      \]
    \item There exists a function $C: \mathcal{Y}^{K} \to \reals$ such that for all ${y}_{1:K}\in\mathcal{Y}^{K}$, if
     $z_{1:K}\sim H^{(K)}_{\star}$, then 
      \[
      	\Var\lrp{Q^{(K)}_{\star}\lrp{{y}_{1:K}\mid z_{1:K}}} 
   	 \leq C({y}_{1:K})\cdot\mathbb{E}\lrb{Q^{(K)}_{\star}\lrp{{y}_{1:K}\mid z_{1:K}}},
     \]
    and
    \[
    \int_{\mathcal{X}^{\otimes K}_{D}}C({y}_{1:K}) \, d{y}_{1:K} < \infty.
    \]
	\end{enumerate}
\end{assumption}
\cref{assump:sample_approx_able} ensures that the data generating distribution can be sufficiently approximated via empirical sampling. 
The first two parts are mild regularity conditions while \cref{assump:sample_approx_able}(c) is applicable to the PMF models used in \cref{sec:pmf-applications} (see \cref{sec:sample_approx_able_example} for details).

The following theorem and \cref{coro:KL} immediately imply \cref{thm:pmf-robust-consistency}. 
\begin{theorem} \label{thm:main-pmf}
	For probabilistic matrix factorization, if \cref{assump:metric-discr-conditions,assump:consistency,assump:sample_approx_able} hold with $\discr{\cdot}{\cdot} = \kl{\cdot}{\cdot}$, then 
    \methodname is $\kappa$-robustly consistent for $\kappa(\rho, K) = \widetilde{\kappa}(K \rho)$. 
\end{theorem}
See \cref{sec:PMF-consistency-proof} for a proof.

\section{Proofs of Mixture Model Results} \label{sec:mixture-model-consistency}
\subsection{Notation}
We write $\op(g(\numobs))$ to denote a random function $f$ that satisfies $f(\numobs) / g(\numobs) \to 0$ in probability for $\numobs \to \infty$.
Let $\widehat\param_{k}^{(\numcomps,\numobs)}$ denote the $k$th component parameter estimate using $\data{1:\numobs}$
for the mixture model with $\numcomps$ components.
More generally, we replace superscript $(\numcomps)$ with $(\numcomps,\numobs)$ to make the dependence on $\numobs$ explicit.
Note that, with probability 1, $N^{(\numcomps,\numobs)}_k \rightarrow \infty$ as $\numobs \rightarrow \infty$.
For simplicity, we introduce the shorthand notation
$F^{(\numcomps,\numobs)}_{k} = F_{\widehat\param^{(\numcomps,\numobs)}_k}$,
$\mcR^{\rho}_{\numobs} (\allparam^{(\numcomps)}) = \mcR^{\rho}(\allparam^{(\numcomps)}; z^{(\numcomps,\numobs)}, \data{1:\numobs})$, and
$\mcR^{\rho,\lambda}_{\numobs}(\allparam^{(\numcomps)}) = \mathcal{R}^{\rho,\lambda}(\allparam^{(\numcomps)}; z^{(\numcomps,\numobs)}, \data{1:\numobs})$.
We also write $F^{(K)}_{\star k} = F_{\param^{(K)}_{\star k}}$ and $G^{(K)}_{\star k} = G_{\allparam^{(K)}_{\star k}}$, and similarly for their densities.

Define the conditional component probabilities based on optimal model distribution and true generating distribution respectively as
\[
	q^{(\numcomps)}_{\star}(k \mid x)= \frac{\eta^{(\numcomps)}_{\star k}
  f_{\star k}^{(\numcomps)}(x)}{p^{(\numcomps)}_{\star}(x)}, \qquad
	q_{o}(k \mid x) = \frac{\eta_{ok}f_{ok}(x)}{p_{o}(x)}.
\]
For conditional probabilities of model distribution, we denote as 
\[
\widehat{q}^{(\numcomps, \numobs)}(k \mid x) = \widehat{\eta}^{(\numcomps,\numobs)}_{k}f_{k}^{(\numcomps,\numobs)}(x) / p^{(\numcomps,\numobs)}(x).
\]
Note that $\widehat{q}^{(\numcomps,\numobs)} (k\mid x) \rightarrow q^{(\numcomps)}_{\star}(k \mid x)$ as $\numobs \rightarrow \infty$ with probability 1.

\subsection{Proof of \cref{thm:main}} \label{pf:main-thm}

We show that (1) for $\numcomps = \numcomps_o$,  $ \mcR^{\rho}_{\numobs}(\allparam^{(\numcomps)}_{\star}) \rightarrow 0$ in probability as $\numobs \rightarrow \infty$, and (2) for $\numcomps < \numcomps_o$, $\mcR^{\rho}_{\numobs}(\allparam^{(\numcomps)}_{\star})\rightarrow \infty$ in probability as $\numobs \rightarrow \infty$.
The theorem conclusion follows immediately from these two results since the loss is lower bounded by zero,
so asymptotically $K_o$ will be estimated to be the smallest value of $\numcomps$ which has a loss of zero.

\textbf{Proof of part (1).} %
Fix $\numcomps=\numcomps_{o}$.
It follows from \cref{assump:metric-discr-conditions}(c,d) and \cref{assump:inference-regularity} that $\discrest{\widehat{F}^{(\numcomps_{o},\numobs)}_{k}}{F^{(\numcomps,\numobs)}_{k} }  \to \discr{\widetilde{F}_{k}}{F^{(\numcomps_{o})}_{\opt k}}$ in probability.
Hence, it follows that there exists $\varepsilon>0$ such that
\[
	\discrest{\widehat{F}^{(\numcomps_{o},\numobs)}_{k}}{F^{(\numcomps,\numobs)}_{k} } < \rho - \varepsilon + \op(1).
\]
Using this inequality, it follows that
\begin{align}
	\mcR^{\rho}_{\numobs}(\allparam^{(\numcomps_o}_{\star})
	 & =\sum_{k=1}^{\numcomps_o}	N^{(\numcomps_{o},\numobs)}_k \max\{0, \discrest{\widehat{F}_{k}^{(\numcomps_o, \numobs)}}{F^{(\numcomps_{o},\numobs)}_{k}}  - \rho\} \\
	 & \le \sum_{k=1}^{\numcomps_o}	N^{(\numcomps_{o},\numobs)}_k\max\{0, -\varepsilon+\op(1)\}.
\end{align}
Hence, we can conclude that $\lim_{\numobs \rightarrow\infty}\Pr[ \mcR^{\rho}_{\numobs}(\allparam^{(\numcomps_o}_{\star}) =0] = 1$. %

\textbf{Proof of part (2).}
Now we consider the case of $\numcomps < \numcomps_o$.
Note the empirical distribution can be written as $\widehat{F}^{(\numcomps, \numobs)} = \sum_{k=1}^{\numcomps}\widehat{\eta}_k^{(\numcomps, \numobs)} \widehat{F}_k^{(\numcomps, \numobs)}$.
By dominated convergence, we know that for $\numobs \to \infty$,
\begin{equation}
	\widehat{\eta}_k^{(\numcomps,\numobs)} = \int \widehat{q}^{(\numcomps,\numobs)} (k\mid x')\datadist(\dee x') \rightarrow  \int q^{(\numcomps)}_{\star}(k\mid x')\datadist(\dee x') = \eta^{(\numcomps)}_{\star k}, \label{eq:pi-conv}
\end{equation}
where convergence is in probability.

For the purpose of contradiction, assume that $\discr{\widetilde{F}_{k}}{F^{(\numcomps)}_{\opt k}} \le \rho$ for all $k \in \{1,\ldots, \numcomps\}$. 
Then we have
\begin{align}
	 \distDisc\left(P^{(\numcomps)}_{\star}, P_{o}\right) 
	 & \le \distDisc\left(P^{(\numcomps)}_{\star}, \widehat{F}^{(\numcomps, \numobs)}\right) + \distDisc\left( \widehat{F}^{(\numcomps, \numobs)}, \datadist\right)\label{eq:trigin-ineq} \\
	 & = \distDisc\left(\sum_{k = 1}^{\numcomps}  \eta^{(\numcomps)}_{\star k}F^{(\numcomps)}_{\star k}, \sum_{k=1}^{\numcomps}\widehat{\eta}_k^{(\numcomps, \numobs)} \widehat{F}_k^{(\numcomps, \numobs)}\right) + \op(1) \label{eq:assump1a}\\
	 & \le \distDisc\left(\sum_{k = 1}^{\numcomps}  \eta^{(\numcomps)}_{\star k}F^{(\numcomps)}_{\star k}, \sum_{k = 1}^{\numcomps}  \widehat{\eta}_k^{(\numcomps, \numobs)} F^{(\numcomps)}_{\star k} \right) \\
   &\qquad + \distDisc\left(\sum_{k = 1}^{\numcomps}  \widehat{\eta}_k^{(\numcomps, \numobs)} F^{(\numcomps)}_{\star k}, \sum_{k=1}^{\numcomps}\widehat{\eta}_k^{(\numcomps, \numobs)} \widehat{F}_k^{(\numcomps, \numobs)}\right)
   + \op(1), \label{eq:decomp}
\end{align}
where \cref{eq:assump1a} follows by  \cref{assump:metric-discr-conditions}(a).
Define $\eta_k^{\min} = \min\{\eta^{(\numcomps)}_{\star k}, \widehat{\eta}^{(\numcomps, \numobs)}_k\}$ and $\bar{\eta}=1-\sum_{k=1}^{\numcomps}\eta_k^{\min}$.
Let $\Vert \cdot \Vert_{1}$ denote the $\ell^1$-norm.
For the first term in \cref{eq:decomp}, we can write
\begin{align}
	 & \distDisc\left(\sum_{k = 1}^{\numcomps} \eta^{(\numcomps)}_{\star k}F^{(\numcomps)}_{\star k}, \sum_{k = 1}^{\numcomps}  \widehat{\eta}_k^{(\numcomps, \numobs)} F^{(\numcomps)}_{\star k} \right)\\
	 & = \distDisc\left(\sum_{k = 1}^{\numcomps} \eta_k^{\min}F^{(\numcomps)}_{\star k} + \bar{\eta}\sum_{k = 1}^{\numcomps}  \frac{\eta^{(\numcomps)}_{\star k}-\eta_k^{\min}}{\bar{\eta}}F^{(\numcomps)}_{\star k}, \sum_{k = 1}^{\numcomps}  \eta_k^{\min}F^{(\numcomps)}_{\star k}+ \bar{\eta}\sum_{k = 1}^{\numcomps}  \frac{\widehat{\eta}_k^{(\numcomps, \numobs)} -  \eta_k^{\min}}{\bar{\eta}}F^{(\numcomps)}_{\star k} \right)  \\
	 & \le \sum_{k = 1}^{\numcomps} \eta_k^{\min}\distDisc\left(F^{(\numcomps)}_{\star k}, F^{(\numcomps)}_{\star k}\right) \\
   &\qquad+ \bar{\eta}\distDisc\left(\sum_{k = 1}^{\numcomps}  \frac{\eta^{(\numcomps)}_{\star k}-\eta_k^{\min}}{\bar{\eta}}F^{(\numcomps)}_{\star k}, \sum_{k = 1}^{\numcomps}  \frac{\widehat{\eta}_k^{(\numcomps, \numobs)} -  \eta_k^{\min}}{\bar{\eta}}F^{(\numcomps)}_{\star k}\right)\label{eq:joint-convexity}\\
	 & \le \Vert \eta^{(\numcomps)}_{\star} - \widehat{\eta}^{(\numcomps,\numobs)} \Vert_{1} \distDisc\left(\sum_{k = 1}^{\numcomps}  \frac{\eta^{(\numcomps)}_{\star k}-\eta_k^{\min}}{\bar{\eta}}F^{(\numcomps)}_{\star k}, \sum_{k = 1}^{\numcomps}  \frac{\widehat{\eta}_k^{(\numcomps, \numobs)} -  \eta_k^{\min}}{\bar{\eta}}F^{(\numcomps)}_{\star k}\right)\label{eq:pi-l1norm-bound-pibar} \\
	 & = \Vert \eta^{(\numcomps)}_{\star} - \widehat{\eta}^{(\numcomps,\numobs)}\Vert_{1} \\
	 & \qquad \times  \distDisc\left( \sum_{k = 1}^{\numcomps} \sum_{\ell = 1}^{\numcomps} \frac{(\eta^{(\numcomps)}_{\star k}-\eta_k^{\min})(\widehat{\eta}_\ell^{(\numcomps, \numobs)} -  \eta_\ell^{\min}) }{\bar{\eta}^2}  F^{(\numcomps)}_{\star k}, \sum_{k = 1}^{\numcomps} \sum_{\ell = 1}^{\numcomps} \frac{(\eta^{(\numcomps)}_{\star k}-\eta_k^{\min})(\widehat{\eta}_\ell^{(\numcomps, \numobs)} -  \eta_\ell^{\min}) }{\bar{\eta}^2} F^{(\numcomps)}_{\star \ell}\right) \\
	 & \le \Vert \eta^{(\numcomps)}_{\star} - \widehat{\eta}^{(\numcomps,\numobs)}\Vert_{1} \sum_{k = 1}^{\numcomps} \sum_{\ell = 1}^{\numcomps} \frac{(\eta^{(\numcomps)}_{\star k}-\eta_k^{\min})(\widehat{\eta}_\ell^{(\numcomps, \numobs)} -  \eta_\ell^{\min}) }{\bar{\eta}^2} \distDisc\left(F^{(\numcomps)}_{\star k}, F^{(\numcomps)}_{\star \ell}\right)\\
	 & \le \Vert \eta^{(\numcomps)}_{\star} - \widehat{\eta}^{(\numcomps,\numobs)}\Vert_{1} \max_{\substack{k,\ell\in \{1,\ldots,K\} \\ 
   k\neq \ell}} \distDisc\left(F^{(\numcomps)}_{\star k}, F^{(\numcomps)}_{\star \ell}\right) \\
	 & = \op(1), \label{eq:first-bound}
\end{align}
where \cref{eq:joint-convexity} uses  \cref{assump:metric-discr-conditions}(b), \cref{eq:pi-l1norm-bound-pibar} follows by the fact that $\bar{\eta}=1- \sum_{k=1}^{\numcomps}\eta_k^{\min} \le  \sum_{k=1}^{\numcomps}\eta_k^{\max}-\sum_{k=1}^{\numcomps}\eta_k^{\min}=\Vert \eta^{(\numcomps)}_{\star} - \widehat{\eta}^{(\numcomps,\numobs)} \Vert_{1}$, and \cref{eq:first-bound} follows by  \cref{assump:metric-discr-conditions}(f) and \cref{eq:pi-conv}.

The second term in \cref{eq:decomp} can be upper bounded as
\begin{align}
	\distDisc\left(\sum_{k = 1}^{\numcomps}  \widehat{\eta}_k^{(\numcomps, \numobs)} F^{(\numcomps)}_{\star k}, \sum_{k=1}^{\numcomps}\widehat{\eta}_k^{(\numcomps, \numobs)} \widehat{F}_k^{(\numcomps, \numobs)}\right)
	 & \le   \sum_{k = 1}^{\numcomps}  \widehat{\eta}_k^{(\numcomps, \numobs)}\distDisc\left( F^{(\numcomps)}_{\star k},  \widehat{F}_k^{(\numcomps, \numobs)}\right)  \label{eq:assump1c}          \\
	 & \le   \sum_{k = 1}^{\numcomps}  \widehat{\eta}_k^{(\numcomps, \numobs)}\kappa\left(\discr{\widehat{F}_k^{(\numcomps, \numobs)}}{F^{(\numcomps)}_{\star k}}\right)  \label{eq:assump1d} \\
	 & \le  \sum_{k = 1}^{\numcomps}  \widehat{\eta}_k^{(\numcomps, \numobs)}\kappa\left(\rho + \op(1) \right) \label{eq:prop}                                                                     \\
	 & = \kappa(\rho) + \op(1) \label{eq:second-bound},
\end{align}
where %
\cref{eq:assump1c} follows by  \cref{assump:metric-discr-conditions}(b),
\cref{eq:assump1d} follows by  \cref{assump:metric-discr-conditions}(e), and
\cref{eq:prop} follows by our assumption for purposes of contradiction, and  \cref{eq:second-bound} follows by the continuity of $\kappa$. %

Plugging \cref{eq:first-bound,eq:second-bound} into \cref{eq:decomp} yields the final inequality $\distDisc\left(P^{(\numcomps)}_{\star}, P_{o}\right) \le \kappa(\rho) + \op(1)$, which contradicts the 
definition of $\rho$.
Therefore, there must exist $\ell$ such that $\discr{\widetilde{F}_{\ell}}{F^{(\numcomps)}_{\star \ell}} > \rho$.
Hence, by \cref{assump:metric-discr-conditions}(c), for some $\varepsilon > 0$, $\discrest{\widehat{F}_{\ell}^{(\numcomps,\numobs)} }{F^{(\numcomps,\numobs)}_{\ell}} = \rho + \varepsilon +\op(1)$.
Hence, we have
\begin{align}
	\mcR^{\rho}_{\numobs}(\allparam^{(\numcomps)}_{\star})
	 & \ge \widehat{n}_{\ell}^{(\numcomps,\numobs)}\max\{0, \discrest{\widehat{F}_{\ell}^{(\numcomps,\numobs)} }{F^{(\numcomps,\numobs)}_{\ell}} - \rho\}                             \\
	 & = \widehat{n}_{\ell}^{(\numcomps,\numobs)}\max\{0, \varepsilon + \op(1)\}                                                                                                               \\
	 & \rightarrow \infty,\label{ineq:N-infty}
\end{align}
where \cref{ineq:N-infty} follows
since $\widehat{n}_{\ell}^{(\numcomps,\numobs)}  \rightarrow \infty$ in probability for $\numobs\rightarrow\infty$.

\subsection{Proof of \cref{coro:KL}}

 \cref{assump:metric-discr-conditions}(a) follows by \citet[Theorem 1]{Wellner:1981}.
 \cref{assump:metric-discr-conditions}(b) holds by \cref{lem:ipm-joint-convex}.
 \cref{assump:metric-discr-conditions}(e) holds with $\kappa(\rho) =(\rho/2)^{1/2}$ by the fact that, letting $d_{\mathrm{TV}}$
denote total variation distance, $d_{\mathrm{BL}} \le d_{\mathrm{TV}}$ and
$(2d_{\mathrm{TV}})^{2}  \le \mathrm{KL}$ \citep[][Section 3]{Gibbs:2002}.
 \cref{assump:metric-discr-conditions}(f) holds since $d_{\mathrm{BL}} \le 1 < \infty$.

\subsection{Proof of \cref{coro:MMD}}
 \cref{assump:metric-discr-conditions}(a) follows by the assumption that $d_{\mathrm{MMD}}$ metrizes weak convergence.
 \cref{assump:metric-discr-conditions}(b) holds by \cref{lem:ipm-joint-convex}.
 \cref{assump:metric-discr-conditions}(e) holds by choosing $\kappa(\rho) =\rho$ for maximum mean discrepancy. \cref{assump:metric-discr-conditions}(f) holds for maximum mean discrepancy with bounded kernels since $d_{\mathrm{MMD}} \le \sup_{x} \mcK(x,x)^{1/2} < \infty$.

\section{Robust Consistency for Probabilistic Matrix Factorization} \label{sec:PMF-consistency-proof}

\subsection{Proof of \cref{thm:main-pmf}}

\paragraph{Notation.}
Let
\[
  Q^{(K)}\lrp{{y}_{1:K}\mid{\phi}_{1:K},H,G_{1:K}}=
  \mathbb{E}_{z_{1:K}\sim H}\lrb{Q^{(K)}\lrp{{y}_{1:K}\mid{\phi}_{1:K},z_{1:K},G_{1:K}}}
\]
denote the joint distribution of ${y}_{n,1:K}$ when $z^{(K)}_{n,1:K}$ is marginalized out, and let
\[
	P^{(K)}({x}\mid{\phi}_{1:K},H,G_{1:K})
  =\int_{\sigma({x})}Q^{(K)}\lrp{{y}_{1:K}\mid{\phi}_{1:K},H,G_{1:K}}\d{y}_{1:K}
\]
be the marginal probability of drawing ${x}_{n}$, where $\sigma({x})=\lrc{{y}_{1:K}\ \mid \sum_{k}{y}_{k}={x}}$.
Similarly, we define the model distributions
\[
	Q^{(K)}_{\star}({y}_{1:K})=
  \mathbb{E}_{z_{1:K}\sim H^{(K)}_{\star}}\lrb{Q^{(K)}_{\star}\lrp{{y}_{1:K}\mid z_{1:K}}} \quad\text{and}\quad
	P^{(K)}_{\star}({x})=
	\int_{\sigma({x})}Q^{(K)}_{\star}({y}_{1:K})\d{y}_{1:K},
\]
the empirical distributions
\[
	\widehat{Q}^{(K,N)}({y}_{1:K}) = \mathbb{E}_{z_{1:K}\sim\widehat{H}^{(K,N)}}\lrb{\widehat{Q}^{(K,N)}\lrp{{y}_{1:K}\mid z_{1:K}}}
    \quad\text{and}\quad
  \widehat{P}^{(K,N)}({x}) =
	\int_{\sigma({x})}\widehat{Q}^{(K,N)}({y}_{1:K})\d{y}_{1:K},
\]
and the bridging distributions
\[
	\check{Q}^{(K,N)}({y}_{1:K})=
\mathbb{E}_{z_{1:K}\sim\widehat{H}^{(K,N)}}\lrb{Q^{(K)}_{\star}\lrp{{y}_{1:K}\mid z_{1:K}}}
 \quad\text{and}\quad
  \check{P}^{(K,N)}({x})=
	\int_{\sigma({x})}\check{Q}^{(K,N)}({y}_{1:K})\d{y}_{1:K},
\]
Let
\[
	\mathcal{R}^{\rho}\lrp{\widehat{G}^{(K,N)}_{1:K}}
	=N\sum^{K}_{k=1}\mop{max}\lrp{0, \widehat{\mathcal{D}}\lrp{\widehat{G}^{(K,N)}_{k},G}-\rho}.
\]

\paragraph{Approach.}
Similar to the proof of \cref{thm:main}, we show that 
(1) if $K=K_{o}$, then $\mathcal{R}^{\rho}(\widehat{G}^{(K,N)}_{1:K}) \to 0$ in probability, and 
(2) if $K<K_{o}$, then $\mathcal{R}^{\rho}(\widehat{G}^{(K,N)}_{1:K}) \to\infty$ in probability. 
The conclusion follows immediately from these two results, as 
in the proof of \cref{thm:main}. 

\paragraph{Proof of part (1).} 
If $K=K_{o}$, then it follows from \cref{assump:metric-discr-conditions}(c,d) and \cref{assump:consistency} that $\widehat{\mathcal{D}}\lrp{\widehat{G}^{(K_{o},N)}_{k},G}\to\mathcal{D}\lrp{G^{(K_{o})}_{\star}, G}$ in probability.
Hence, it follows that there exists $\varepsilon>0$ such that
\[
	\widehat{\mathcal{D}}\lrp{\widehat{G}^{(K_{o},N)}_{k},G}<\rho-\varepsilon+o_{P}(1).
\]
Using this inequality, we have
\[
	\mathcal{R}^{\rho}\lrp{\widehat{G}^{(K,N)}_{1:K}}
	\leq N\sum^{K}_{k=1}\mop{max}\lrp{0, -\varepsilon+o_{P}(1)}.
\]
Hence, we can conclude that
$\lim_{N\to\infty}\Pr\lrc{\mathcal{R}^{\rho}\lrp{\widehat{G}^{(K,N)}_{1:K}}=0}=1$.

\paragraph{Proof of part (2).} Consider the case of $K<K_{o}$.
We have
\[
	\distDisc(\datadist,P^{(K)}_{\star})&\leq \distDisc\lrp{\datadist,\widehat{P}^{(K,N)}}+\distDisc\lrp{\widehat{P}^{(K,N)},P^{(K)}_{\star}}\\
	&\leq \distDisc\lrp{\datadist,\widehat{P}^{(K,N)}}+\distDisc\lrp{\widehat{P}^{(K,N)},\check{P}^{(K,N)}}+\distDisc\lrp{\check{P}^{(K,N)},P^{(K)}_{\star}}\\
	&=o_{P}(1)+\distDisc\lrp{\widehat{P}^{(K,N)},\check{P}^{(K,N)}}
  +\distDisc\lrp{\check{P}^{(K,N)},P^{(K)}_{\star}}\label{eq:follows_assump1a}\\
	&\leq\widetilde{\kappa}\lrp{ \widehat{\mathcal{D}}\lrp{\widehat{P}^{(K,N)},\check{P}^{(K,N)}}+o_{P}(1)}
	+\widetilde{\kappa}\lrp{\mathcal{D}\lrp{\check{P}^{(K,N)},P^{(K)}_{\star}}}+o_{P}(1)\\
	&\leq\widetilde{\kappa}\lrp{ \widehat{\mathcal{D}}\lrp{\widehat{Q}^{(K,N)},\check{Q}^{(K,N)}}}
	+\widetilde{\kappa}\lrp{\mathcal{D}\lrp{\check{Q}^{(K,N)},Q^{(K)}_{\star}}}+o_{P}(1).\label{eq:main_ineq}
\]
where \cref{eq:follows_assump1a} follows from \cref{assump:metric-discr-conditions}(a).
A bound on $\widehat{\mathcal{D}}\lrp{\widehat{Q}^{(K,N)},\check{Q}^{(K,N)}}$ is given by
\[
  &\widehat{\mathcal{D}}\lrp{\widehat{Q}^{(K,N)},\check{Q}^{(K,N)}}\\
	&=\widehat{\mathcal{D}}\lrp{\mathbb{E}_{z_{1:K}\sim\widehat{H}^{(K,N)}}\lrb{\prod_{k}\widehat{Q}^{(K,N)}_{k}(\blank\mid z_{k})},
	\mathbb{E}_{z_{1:K}\sim\widehat{H}^{(K,N)}}\lrb{\prod_{k}Q^{(K)}_{\star k}(\blank\mid z_{k})}}\\
	&\leq\mathbb{E}_{z_{1:K}\sim\widehat{H}^{(K,N)}}\lrb{ \widehat{\mathcal{D}}
  \lrp{\prod_{k}\widehat{Q}^{(K,N)}_{k}(\blank\mid z_{k}),\prod_{k}Q^{(K)}_{\star k}(\blank\mid z_{k})}}+o_{P}(1)
  \label{eq:follows-convex-KL}\\
  &=\sum_{k}\mathbb{E}_{z_{1:K}\sim\widehat{H}^{(K,N)}}\lrb{ \widehat{\mathcal{D}}\lrp{\widehat{Q}^{(K,N)}_{k}(\blank\mid z_{k}),Q^{(K)}_{\star k}(\blank\mid z_{k})}}+o_{P}(1)\\
  &=\sum_{k}\mathbb{E}_{z_{1:K}\sim\widehat{H}^{(K,N)}}\lrb{
    \widehat{\mathcal{D}}\lrp{\widehat{Q}^{(K,N)}(\blank\mid z_{k}),Q\lrp{\blank\mid {\phi}^{(K,N)}_{k},z_{k},G}}
    +o_{P}(1)
  }+o_{P}(1)\label{eq:follows-assump-6b}\\
  &\leq\sum_{k}\widehat{\mathcal{D}}\lrp{\widehat{G}^{(K,N)}_{k},G} + o_{P}(1). \label{eq:follows-KL-invariance}
\]
where \cref{eq:follows-convex-KL} follows from the fact that the KL divergence is convex with respect to both of its arguments,
\cref{eq:follows-assump-6b} follows from \cref{assump:sample_approx_able}(b),
and \cref{eq:follows-KL-invariance} follows from the fact that KL divergence is invariant under diffeomorphism. 
With this we bound
\[
  \widetilde{\kappa}\lrp{\mathcal{D}\lrp{\widehat{Q}^{(K,N)},\check{Q}^{(K,N)}}}
  &\leq\widetilde{\kappa}\lrp{\sum_{k}\mathcal{D}\lrp{\widehat{G}^{(K,N)}_{k},G}+o_{P}(1)}\\ 
  &\leq\widetilde{\kappa}\lrp{\sum_{k}\mathcal{D}\lrp{\widehat{G}^{(K,N)}_{k},G}}+o_{P}(1).
  \label{eq:subbound_1}
\]

Next we will bound $\mathcal{D}\lrp{\check{Q}^{(K,N)},Q^{(K)}_{\star}}$. 
For the chi-squared distance $D_{\chi^2}(P,Q) = \int \frac{(P(x) - Q(x))^2}{Q(x)} \, dx$ \citep{cover2006elements}, we have that
\begin{align}
\lefteqn{\mathbb{E}_{h_{1:N}} \left[ D_{\chi^2}\left( \check{Q}^{(K,N)}, Q^{(K)}_{\star} \right) \right]} \\
&= \mathbb{E}_{h_{1:N}} \left[ \int_{\mathcal{X}^{\otimes K}_{D}} \frac{(\check{Q}^{(K,N)} - Q^{(K)}_{\star})^2}{Q^{(K)}_{\star}} \, d y_{1:K} \right] \\
&= \int_{\mathcal{X}^{\otimes K}_{D}} \mathbb{E}_{h_{1:N}} \left[ \frac{(\check{Q}^{(K,N)} - Q^{(K)}_{\star})^2}{Q^{(K)}_{\star}} \right] d y_{1:K} 
\quad \text{[by bounded convergence]}\\
&= \int_{\mathcal{X}^{\otimes K}_{D}} \mathbb{E}_{h_{1:N}} \left[ \frac{(\check{Q}^{(K,N)})^2}{Q^{(K)}_{\star}} - 2\check{Q}^{(K,N)} + Q^{(K)}_{\star} \right] d y_{1:K} \\
&= \int_{\mathcal{X}^{\otimes K}_{D}} \left[ \mathbb{E}_{h_{1:N}} \left( \frac{(\check{Q}^{(K,N)})^2}{Q^{(K)}_{\star}} \right)
- 2Q^{(K)}_{\star} + Q^{(K)}_{\star} \right] d y_{1:K}
\quad \text{[ since } \mathbb{E}[\check{Q}^{(K,N)}] = Q^{(K)}_{\star} \text{]}\\
&= \int_{\mathcal{X}^{\otimes K}_{D}} \frac{ \mathbb{E}_{h_{1:N}} \left( (\check{Q}^{(K,N)})^2 \right) }{Q^{(K)}_{\star}} - Q^{(K)}_{\star}\, d y_{1:K} \\
&= \int_{\mathcal{X}^{\otimes K}_{D}} \frac{ \operatorname{Var}_{h_{1:N}}(\check{Q}^{(K,N)}) + \left( \mathbb{E}_{h_{1:N}} \check{Q}^{(K,N)} \right)^2 }{Q^{(K)}_{\star}}- Q^{(K)}_{\star} \, d y_{1:K} \\
&= \int_{\mathcal{X}^{\otimes K}_{D}} \frac{ \operatorname{Var}_{h_{1:N}}(\check{Q}^{(K,N)}) }{ Q^{(K)}_{\star} } d y_{1:K}\\
&= \frac{1}{N} \int_{\mathcal{X}^{\otimes K}_{D}} \frac{ \Var_{z_{1:K}\sim H^{(K)}_{\star}}\lrp{Q^{(K)}_{\star}({y}_{1:K}\mid z_{1:K})}}
{ Q^{(K)}_{\star} } d y_{1:K}
\quad \text{[by variance of sample mean: } \frac{1}{N} \operatorname{Var}(X) \text{]}\\
&\leq\frac{1}{N}\int_{\mathcal{X}^{\otimes K}_{D}}\frac
{C({y}_{1:K})\cdot\mathbb{E}_{z_{1:K}\sim H^{(K)}_{\star}}\lrb{Q^{(K)}_{\star}({y}_{1:K}\mid z_{1:K})}}
{Q^{(K)}_{\star}({y}_{1:K})}d{y}_{1:K} \\
&\leq\frac{1}{N}\int_{\mathcal{X}^{\otimes K}_{D}} C({y}_{1:K}) d{y}_{1:K}.
\end{align}
Since $\mathrm{KL} \leq \log(1 + D_{\chi^2}) \le D_{\chi^2}$ \citep{Gibbs:2002},
it follows that
$\chi^2 \to 0  \Rightarrow \mathrm{KL}\to 0.$
Together with \cref{assump:sample_approx_able}, it follows that 
\[
	\lim_{N\to\infty}\mathbb{E}_{{{h}_{1:N}}}\lrb{\mathcal{D}_{KL}\lrp{\check{Q}^{(K,N)},Q^{(K)}_{\star}}}
	=0.
\]
and therefore
\[
	\widetilde{\kappa}\lrp{\mathcal{D}\lrp{\check{Q}^{(K,N)},Q^{(K)}_{\star}}}\leq\widetilde{\kappa}\lrp{o_{P}(1)}
	\leq o_{P}(1).\label{eq:subbound_2}
\]
Using \cref{eq:subbound_1,eq:subbound_2}, we rewrite \cref{eq:main_ineq} as
\[
	\distDisc(\datadist,P^{(K)}_{\star})&\leq\widetilde{\kappa}\lrp{\sum_{k}\widehat{\mathcal{D}}\lrp{\widehat{G}^{(K,N)}_{k},G}}+o_{P}(1).
\]
Therefore, under the conditions for $\kappa$-robust consistency, we are guaranteed that 
\[
  \widetilde{\kappa}(K\rho)\leq\widetilde{\kappa}\lrp{\sum_{k}\widehat{\mathcal{D}}\lrp{\widehat{G}^{(K,N)}_{k},G}}+o_{P}(1).
\]
Because $\widetilde{\kappa}$ is monotonic,
\[
  K\rho\leq\sum_{k}\widehat{\mathcal{D}}\lrp{\widehat{G}^{(K,N)}_{k},G}+o_{P}(1).
\]
This holds if and only if there exists $\ell\in\lrc{1,\dots,K}$ such that $\mathcal{D}\lrp{G^{(K)}_{\star \ell},G}\geq\rho$. Hence, for some $\varepsilon>0$, $\mathcal{D}\lrp{G^{(K)}_{\star \ell},G}=\rho+\varepsilon$. As a result
\[
  \mathcal{R}^{\rho}\lrp{\widehat{G}^{(K,N)}_{1:K}}
  &\geq N\mop{max}\lrp{0, \widehat{\mathcal{D}}\lrp{\widehat{G}^{(K,N)}_{\ell},G}-\rho}\\
  &= N\mop{max}\lrp{0, \mathcal{D}\lrp{G^{(K)}_{\star \ell},G}-\rho+o_{P}(1)} \label{eq:follow1c}\\
  &= N\mop{max}\lrp{0, \varepsilon+o_{P}(1)}\\
  &\to \infty
\]
where \cref{eq:follow1c} follows from \cref{assump:metric-discr-conditions}(c).

\subsection{Verifying \cref{assump:sample_approx_able} for Applications} \label{sec:sample_approx_able_example}

We show that \cref{assump:sample_approx_able} holds for both the PMF models used for the experiments in \cref{sec:pmf-applications}. 
It is sufficient to verify the assumption for a single element $y_{nk}$ since the variance and integrability conditions can often be checked component-wise.
Hence, we drop the dependence on $n$ and $k$ in our notation. 

\subsubsection{Poisson PMF}

Consider the Poisson model with  $y \sim \distPoiss(\lambda)$ and $\lambda = W h$.
For convenience, we assume $h \sim \distGamma(\alpha, \beta)$.
To compute the first and second moments of 
\[
  P(y\mid h) = \frac{(W h)^{y} e^{-W h}}{y!},
\]
we will use the identity
$\Gamma(z) = \int_0^\infty t^{z-1} e^{-t} \, \d t$ and integration by substitution.
For the first moment we have 
\begin{align}
\mathbb{E}_{h}\bigl[P(y\mid h)\bigr]
  &= \int_{0}^{\infty} \frac{(W h)^{y} e^{-W h}}{y!}
     \;\frac{\beta^{\alpha}}{\Gamma(\alpha)} h^{\alpha-1} e^{-\beta h} \, \mathrm{d}h\\[6pt]
  &= \frac{W^{y} \beta^{\alpha}}{y!\,\Gamma(\alpha)}
     \int_{0}^{\infty} h^{y+\alpha-1} e^{-(W+\beta)h} \, \mathrm{d}h\\[6pt]
  &= \frac{W^{y} \beta^{\alpha}\,\Gamma(y+\alpha)}{y!\,\Gamma(\alpha)(W+\beta)^{y+\alpha}},
\end{align}
while the second moment is 
\begin{align}
\mathbb{E}_{h}\bigl[P^{2}(y\mid h)\bigr]
  &= \int_{0}^{\infty} \frac{(W h)^{2y} e^{-2 W h}}{(y!)^{2}}
     \;\frac{\beta^{\alpha}}{\Gamma(\alpha)} h^{\alpha-1} e^{-\beta h} \, \mathrm{d}h\\[6pt]
  &= \frac{W^{2y} \beta^{\alpha}}{\Gamma(\alpha)(y!)^{2}}
     \int_{0}^{\infty} h^{2y+\alpha-1} e^{-(2W+\beta)h} \, \mathrm{d}h\\[6pt]
  &= \frac{W^{2y} \beta^{\alpha}\, \Gamma(2y+\alpha)}{\Gamma(\alpha)(y!)^{2}(2W+\beta)^{2y+\alpha}}.
\end{align}
Taking the ratio of the second to the first moment, define 
\[
  C(y) = \frac{\mathbb{E}_{h}[P^{2}(y\mid h)]}{\mathbb{E}_{h}[P(y\mid h)]}
         = \frac{W^{y} \, \Gamma(2y+\alpha)}{y!\,\Gamma(y+\alpha)}
           \left(\frac{W+\beta}{2W+\beta}\right)^{y+\alpha}
           \cdot\left(\frac{1}{2W+\beta}\right)^{y},
\]
which is continuous and finite for all $y \in \nats$. 
Now, using Stirling's approximation
$ \Gamma(z) \sim \sqrt{2\pi}\, z^{z-1/2} e^{-z},$ we have
\begin{align}
    \frac{\Gamma(2y+\alpha)}{y!\,\Gamma(y+\alpha)} &\sim
    \frac{(2y)^{2y+\alpha-1/2} e^{-2y}}{\sqrt{2\pi}\,y^{y+\alpha-1/2} e^{-y}\cdot \sqrt{2\pi}\,y^{y+1/2} e^{-y}}
    \sim \frac{2^{2y+\alpha}}{\sqrt{\pi y}}.
\end{align}
Substituting into $C(y)$ yields
\begin{align}
C(y)
 &\sim \frac{W^{y}}{\sqrt{\pi y}}
        \left(\frac{W+\beta}{2W+\beta}\right)^{\alpha}
        \left(\frac{W+\beta}{2W+\beta}\right)^{y}
        \left(\frac{1}{2W+\beta}\right)^{y}
        2^{2y}\\
 &\sim \frac{1}{\sqrt{\pi y}}
        \left(\frac{W+\beta}{2W+\beta}\right)^{\alpha}
        \left(\frac{2^2W(W+\beta)}{(2W+\beta)^2}\right)^{y}\\
 &\sim \frac{1}{\sqrt{\pi y}}
        \left(\frac{4W^2+4W\beta}{4W^2+\beta^2+4W\beta}\right)^{y}.      
\end{align}
Because $0<\frac{4W^2+4W\beta}{4W^2+\beta^2+4W\beta} < 1$, the ratio $C(y)$ decays exponentially as $y \to \infty$, 
hence $\sum_{y=0}^{\infty} C(y)  <\infty$.

\subsubsection{Gaussian PMF}

Consider the Gaussian setting $y \sim \mathcal{N}(\phi h, \sigma^2)$
and, following common practice, we let $h \sim \mathcal{N}(\mu, \tau^2)$. 
To compute the moments of $p(y \mid h)$, we integrate over the latent variable $h$ by combining the terms in the exponential, completing the square, and using the Gaussian integral identity
$\int_{-\infty}^{\infty} e^{-(a x^2 + b x + c)} \, \dee x = \sqrt{\frac{\pi}{a}} e^{\frac{b^2}{4a} - c}.$ 
The first moment is 
\begin{align}
\mathbb{E}_h [p(y \mid h)] 
&= \int_{-\infty}^{\infty} \frac{1}{\sqrt{2\pi\sigma^2}} \exp\left(-\frac{(y - \phi h)^2}{2\sigma^2} \right)
\cdot \frac{1}{\sqrt{2\pi\tau^2}} \exp\left(-\frac{(h - \mu)^2}{2\tau^2} \right) \, dh \\
&\propto \exp\left\{ -\frac{(y - \phi \mu)^2}{2(\phi^2 \tau^2 + \sigma^2)} \right\},
\end{align}
while the second moment is 
\begin{align}
\mathbb{E}_h [p(y \mid h)^2] 
&= \int_{-\infty}^{\infty} \left( \frac{1}{\sqrt{2\pi\sigma^2}} \exp\left(-\frac{(y - \phi h)^2}{2\sigma^2} \right) \right)^2
\cdot \frac{1}{\sqrt{2\pi\tau^2}} \exp\left(-\frac{(h - \mu)^2}{2\tau^2} \right) \, dh \\
&\propto \exp\left\{ -\frac{(y - \phi \mu)^2}{2(\phi^2 \tau^2 + \frac{1}{2}\sigma^2)} \right\}.
\end{align}
Therefore, up to a multiplicative constant, the ratio of moments is 
\begin{align}
    C(y) 
    &= \frac{\mathbb{E}_{h}[P^{2}(y\mid h)]}{\mathbb{E}_{h}[P(y\mid h)]} \\
    &\propto \frac{ \exp\left\{ -\frac{(y - \phi \mu)^2}{2(\phi^2 \tau^2 + \frac{1}{2}\sigma^2)} \right\} }{ \exp\left\{ -\frac{(y - \phi \mu)^2}{2(\phi^2 \tau^2 + \sigma^2)} \right\} }\\
    &= \exp\left\{ -(y - \phi \mu)^2 \left( \frac{1}{2(\phi^2 \tau^2 + \frac{1}{2}\sigma^2)} - \frac{1}{2(\phi^2 \tau^2 + \sigma^2)} \right) \right\}.
\end{align}
Since $\frac{1}{2(\phi^2 \tau^2 + \frac{1}{2}\sigma^2)} > \frac{1}{2(\phi^2 \tau^2 + \sigma^2)}$, 
we get $C(y) \propto \exp\left\{ -a (y - \phi \mu)^2 \right\}$ for some $a > 0$. Therefore,
$C(y)$ decays exponentially as $y \to \infty$ and hence $\int C(y)\d y <\infty$.

\section{Technical Lemma}

The following lemma states that integral probability metrics (as defined in \cref{eq:IPM}) are jointly convex -- that is, they satisfy \cref{assump:metric-discr-conditions}(b).
\begin{lemma}\label{lem:ipm-joint-convex}
	Suppose $P_i$ and $Q_i$, $i = 1,\ldots, n$ are probability measures defined on $\mcX$. Then for $0\leq w_i \leq 1$ and $\sum_{i = 1}^n w_i = 1$,
	\[
		d_{\mcH}\left(\sum_{i=1}^{n}w_i P_i, \sum_{i=1}^{n}w_i Q_i\right)
		\le \sum_{i=1}^{n}w_i  	d_{\mcH} (P_i, Q_i).
	\]

\end{lemma}

\begin{proof}
	By definition of the integral probability metric,
	\begin{align}
		{d_{\mcH}\left(\sum_{i=1}^{n}w_i P_i, \sum_{i=1}^{n}w_iQ_i\right)}
		 & =  \sup_{h \in \mcH}\left\vert \int_{\mcX} h(x) \left(\sum_{i=1}^{n}w_i P_i(\dee x)\right) - \int_{\mcX} h(x)\left(\sum_{i=1}^{n}w_i Q_i(\dee x)\right)\right\vert  \nonumber \\
		 & = \sup_{h \in \mcH}\left\vert \sum_{i=1}^{n}w_i \left(\int_{\mcX} h(x)  P_i(\dee x) - \int_{\mcX} h(x) Q_i(\dee x) \right)\right\vert\nonumber                                \\
		 & \le  \sup_{h \in \mcH} \sum_{i=1}^{n}w_i \left\vert  \int_{\mcX} h(x) P_i(\dee x) - \int_{\mcX} h(x) Q_i(\dee x) \right\vert \nonumber                                        \\
		 & \le \sum_{i=1}^{n}w_i  \sup_{h \in \mcH} \left\vert  \int_{\mcX} h(x) P_i(\dee x) - \int_{\mcX} h(x) Q_i(\dee x) \right\vert  \nonumber                                       \\
		 & = \sum_{i=1}^{n}w_i  d_{\mcH} (P_i, Q_i).
	\end{align}
\end{proof}

\section{Connection between \methodname and Likelihood-based Inference} \label{sec:connection-to-likelihood}

We can relate our loss function to the conditional negative log-likelihood in the case where the discrepancy measure is the KL divergence.
Let $\widetilde{X}_{k}$ be a random observation selected uniformly from the $k$th cluster.
Then the conditional negative log-likelihood given $z_{1:\numobs}$ is
\begin{align}
	-\log p(x_{1:\numobs} \mid \theta, z_{1:\numobs})
	 & = \sum_{k=1}^{\numcomps}\sum_{n=1}^N - u_{nk}\log f_{\param_{k}}(x_n)      \\
	 & = N \sum_{k=1}^{\numcomps} \eta_{k} \E\{-\log f_{\param_{k}}(\widetilde X_{k})\}
\end{align}
For the sake of argument, if $\widetilde{X}_{k}$ were distributed according to a density $f_{ok}$, then
the Kullback--Leibler divergence between $f_{ok}$ and $f_{\param_{k}}$ would be
\[
	\begin{aligned}
		\kl{f_{ok}}{ f_{\param_{k}}}
		 & = \E\{\log f_{ok}(\widetilde X_{k})\} -  \E\{\log f_{\param_{k}}(\widetilde X_{k})\} \\
		 & = -\mathcal{H}(f_{ok}) -  \E\{\log f_{\param_{k}}(\widetilde X_{k})\},
	\end{aligned}
\]
where $\mathcal{H}(f) = \int f(x) \log f(x) \dee x$ denotes the entropy of a density $f$.
Then the negative conditional log-likelihood for each is equal to the Kullback--Leibler divergence, up to an entropy term that depends on the data (i.e., $f_{ok}$) but not the parameter $\param_{k}$:
\[
	-\log p(x_{1:\numobs} \mid \theta, z_{1:\numobs})
	\approx N \sum_{k=1}^{\numcomps}\eta_k\left\{ \kl{f_{ok}}{ f_{\param_{k}}} + \mathcal{H}(f_{ok})\right\}.
\]
Thus, we can view the loss from \cref{eq:robust-loss} as targeting the negative conditional log-likelihood but
(a) using a consistent estimator $\klest{f_{ok}}{f_{\param_{k}}}$ in place of $\kl{f_{ok}}{f_{\param_{k}}}$ and
(b) ``coarsening'' the Kullback--Leiber divergence using the map $t \mapsto \max(0, t - \rho)$ to avoid overfitting.

\section{Kullback--Leibler Divergence Estimation} \label{sec:kl-estimation}

We first summarize how to best estimate the KL divergence. 
The remaining subsections contain further details, including supporting theory
and experiments. 

\subsection{Summary}
We consider a general setup with observations $y_{1:\numobs} = (y_{1},\dots, y_{\numobs}) \in \mcX^{\otimes \numobs}$ independent, identically distributed
from a distribution $P$.
First consider the case where $\mcX$ is countable, and let $\numobs(x) = \#\{ n \in \{1,\dots, \numobs\} \mid y_{n} = x\}$ denote the number of observations
taking the value $x \in \mcX$.
Letting $Q$ be a distribution with probability mass function $q$, we can use the plug-in estimator for $\kl{P}{Q}$,
\[
	\begin{aligned}
		\klest{y_{1:\numobs}}{F}
		= \sum_{x \in \mcX} \frac{\numobs(x)}{\numobs} \log \left\{\frac{\numobs(x)}{\numobs q(x)}\right\},
		\label{eq:plug-in-KL}
	\end{aligned}
\]
which is consistent under modest regularity conditions \citep{Paninski:2003}.

Next we consider the case of general distributions on $\mcX \subseteq \mathbb{R}^{D}$,
when estimation is less straightforward.
One common approach is to utilize $k$-nearest-neighbor density estimation.
For $r > 0$, let $V_{D}(r) = \frac{\eta^{D/2}}{\Gamma(D/2+1)}r^{D}$ denote the volume of an
$D$-dimensional ball of radius $r$ and let $r_{k,n}$ denote the distance to the $k$th nearest neighbor of $y_n$.
Following the same approach as \citet{Zhao:2020} and assuming the distribution $Q$ has Lebesgue density $q$,
we obtain a one-sample estimator for $\kl{P}{Q}$:
\begin{equation}
	\begin{aligned}
		\klestsub{b}{k}{y_{1:\numobs}}{Q}
		 & = \frac{1}{N}\sum_{n=1}^{\numobs} \log\left\{\frac{k/(\numobs-1)}{V_{D}(r_{k,n}) q(y_n)}\right\}.
		\label{eq:stare-knn-kl-est}
	\end{aligned}
\end{equation}
As we discuss in detail below, for fixed $k$, the estimator in \cref{eq:stare-knn-kl-est} is asymptotically biased.
However, it is easy to correct this bias, leading to the unbiased, consistent estimator
\begin{equation}
	\klestsub{u}{k}{y_{1:\numobs}}{Q} = \klestsub{b}{k}{y_{1:\numobs}}{Q} - \log k + \psi(k),
\end{equation}
where $\psi(k)$ denotes the digamma function.
Another way to construct a consistent estimator is to let $k = k_{\numobs}$ depend on the data size $\numobs$,
with $k_{\numobs}  \to \infty$ as $\numobs \to \infty$.
A canonical choice is $\klestsub{b}{k_{\numobs}}{y_{1:\numobs}}{Q}$ with $k_{\numobs}=\numobs^{1/2}$.

We compare the three estimators for various dimensions below.
Our results show that the bias-corrected estimator slightly improves the biased version when $\numobs \ge 5000$,
while the adaptive estimator with $k_\numobs = \numobs^{1/2}$ has the most reliable performance  when $D=4$.
Since the data dimensions are relatively low in all our experiments, we use the adaptive estimator with $k_N=N^{1/2}$. %

It is important to highlight that KL divergence estimators require density estimation, which in general requires the sample size to grow exponentially with the dimension \citep{Donoho:2000}.
This limits the use of such estimators with generic high dimensional data.
However, a general strategy to address this would be to take advantage of some known or inferable
structure in the distribution to reduce the effective dimension of the problem.
We provide a more detailed illustration of this strategy in \cref{sec:high-dim-simulation} with simulated data that exhibits weak correlations across coordinates.

\subsection{Detailed Theory and Methods}

Following \citet{Wang:2009}, we derive and study the theory of various one-sample Kullback--Leibler estimators on continuous distributions. Estimating Kullback--Leibler between continuous distributions is a delicate matter.
One common way is to start with density estimations.

Consider a general case on $\mcX = \reals^{D}$.
Suppose $y_{1:\numobs} = (y_1, \ldots, y_{\numobs}) \in \mcX^{\otimes \numobs}$ are independent,
identically distributed from a continuous distribution $P$ with density $p$.
For $r>0$, one can estimate the density $p(y_n)$ by
\begin{equation}
	P(V_{D}(r)) \approx p(y_{n}) V_{D}(r),
	\label{eq:knn-density-intuition}
\end{equation}
where $V_{D}(r) = \eta^{D/2}r^D/\Gamma(D/2+1)$ is the volume of a $D$-dimensional ball centered at $y_{n}$ of radius $r$.
Fix the query point $y_{n}$.
The radius $r$ can be determined by finding the $k$-th nearest neighbor $y_{n(k)}$ of $y_{n}$, i.e., $r_{k,n} = \Vert y_{n(k)} - y_n\Vert$, where $\Vert \cdot \Vert$ denotes the Euclidean distance.
Therefore, the ball centered at $y_{n}$ with radius $r_{k,n}$ contains $k$ points and thus $P(V_{D}(r_{k,n}))$ can be estimated by $k/(\numobs-1)$.
Plugging this estimate back to \cref{eq:knn-density-intuition} yields the $k$-nearest-neighbor density estimator for $p(y_n)$,
\begin{equation}
	\widehat{p}_{\numobs}(y_n) = \frac{k/(\numobs-1)}{V_{D}(r_{k,n})}.
	\label{eq:knn-density-est-p}
\end{equation}

To estimate Kullback-Leibler divergence, \citet{Wang:2009} studied various two-sample estimators given two sets of samples $y_{1},\dots, y_{N} \sim P$ and $z_{1},\dots,z_{M} \sim Q$ where the distributions $P$ and $Q$ are unknown.
However, in the context of our method, we want to estimate the Kullback-Leibler divergence with one set of samples $y_{1:\numobs}$ and one known distribution from our assumed model $Q$.
Hence, we modify the two-sample $k$-nearest-neighbor estimators from \citep{Wang:2009} to
create one-sample Kullback--Leibler estimators.

Given samples $y_{1},\dots, y_{N} \sim P$, where $P$ is unknown, and a known distribution $Q$ with density $q$,
we can use \cref{eq:knn-density-est-p} to obtain the one-sample $k$-nearest-neighbor estimator
\begin{equation}
	\begin{aligned}
		\klestsub{b}{k}{y_{1:\numobs}}{Q}
		 & = \frac{1}{\numobs}\sum_{n=1}^{\numobs}\log\left\{\frac{\widehat{p}_{\numobs}(y_n)}{q(y_n)}\right\}   \\
		 & = \frac{1}{N}\sum_{n=1}^{\numobs} \log\left\{\frac{k/(\numobs-1)}{V_{D}(r_{k,n}) q(y_n)}\right\}.
		\label{eq:canonical-knn-kl-est}
	\end{aligned}
\end{equation}
Following the proof of \citet[Theorem 1]{Wang:2009}, we can show that for fixed $k$
\begin{equation}
	\lim\limits_{n\rightarrow\infty}E[\klestsub{b}{k}{y_{1:\numobs}}{Q}] = \kl{P}{Q} +\log k- \psi(k),
	\label{eq:consistency-stare-knn-kl-est}
\end{equation}
where $\psi(k) = \Gamma'(k)/\Gamma(k)$ is the digamma function. \cref{eq:consistency-stare-knn-kl-est} suggests that this canonical Kullback--Leibler estimator is asymptotically biased.
However, using \cref{eq:consistency-stare-knn-kl-est}, we can define the consistent (asymptotically unbiased) estimator
\begin{equation}
	\klestsub{u}{k}{y_{1:\numobs}}{Q} = \klestsub{b}{k}{y_{1:\numobs}}{Q} - \log k + \psi(k).
	\label{eq:biased-correct-knn-kl-est}
\end{equation}
Another way to eliminate the bias is to make $k$ data-dependent, which we call \emph{adaptive} $k$-nearest-neighbor estimators.
Following the proof of \citet[Theorem 5]{Wang:2009}, we can show that $\klestsub{b}{k_{\numobs}}{y_{1:\numobs}}{Q}$ is asymptotically consistent by choosing $k_{N}$ to satisfy mild growth conditions.
\begin{proposition} \label{prop:one-sample-klest}
	Suppose $P$ and $Q$ are distributions uniformly continuous on $\mathbb{R}^D$ with densities $p$ and $q$,
	and $\kl{P}{Q} < \infty$. Let $k_{\numobs}$ be a positive integer satisfying
	\[
		\frac{k_{\numobs}}{\numobs} \rightarrow 0, \qquad \frac{k_{\numobs}}{\log \numobs} \rightarrow \infty.
	\]
	If $\inf_{p(y)>0} p(y)>0$ and $\inf_{q(y)>0} q(y)>0$, then
	\begin{equation}
		\lim\limits_{n\rightarrow\infty}\klestsub{b}{k_{\numobs}}{y_{1:\numobs}}{Q} = \kl{P}{Q}
	\end{equation}
	almost surely.
\end{proposition}
\begin{proof}
	Let $p$ and $q$ are densities of $P$ and $Q$ respectively.
	Consider the following decomposition of the error
	\begin{equation}
		\begin{aligned}
			 & \left|\klestsub{b}{k_{\numobs}}{y_{1:\numobs}}{Q} -\kl{P}{Q} \right|                                                                                                                                                                                                                                                                                                          \\
			 & \leq\left|\frac{1}{\numobs} \sum_{n=1}^{\numobs} \log\left\{\frac{\widehat{p}_{\numobs}\left(y_n\right)}{q\left(y_n\right)}\right\} -\frac{1}{\numobs} \sum_{n=1}^{\numobs} \log\left\{\frac{p\left(y_n\right)}{q\left(y_n\right)}\right\} \right| +\left|\frac{1}{\numobs} \sum_{n=1}^{\numobs} \log\left\{\frac{p\left(y_n\right)}{q\left(y_n\right)}\right\} -\kl{P}{Q}\right| \\
			 & \leq \frac{1}{\numobs} \sum_{n=1}^{\numobs}\left|\log \widehat{p}_{\numobs}\left(y_n\right)-\log p\left(y_n\right)\right|  +\left|\frac{1}{\numobs} \sum_{n=1}^{\numobs} \log\left\{\frac{p\left(y_n\right)}{q\left(y_n\right)}\right\} -\kl{P}{Q}\right|                                                                                                                         \\ &= e_1+e_2 .\end{aligned}
	\end{equation}
	It follows by the conditions that $k_{\numobs}/\numobs \rightarrow 0$ and $k_{\numobs}/\log \numobs \rightarrow \infty$ 	and the theorem given in \cite{Devroye:1977} that $\widehat{p}_{\numobs}$ is uniformly strongly consistent: almost surely
	\begin{equation}
		\lim _{\numobs \rightarrow \infty} \sup _y\left|\widehat{p}_{\numobs}(y)-p(y)\right| \rightarrow 0.
	\end{equation}
	Therefore, following the proof of \cite{Wang:2009}, for any $\varepsilon > 0$, there exists $N_1$ such that for any $n > N_1$, $e_1 < \varepsilon/2$. For $e_2$, it simply follows by the Law of Large Numbers that for any $\varepsilon > 0$, there exists $N_2 $ such that for any $n > N_2$, $e_2 < \varepsilon/2$. By choosing $N = \max(N_1, N_2)$, for any $n>N$, we have $|\klestsub{b}{k_{\numobs}}{y_{1:\numobs}}{Q} -\kl{P}{Q} | < \varepsilon$.
\end{proof}

\subsection{Detailed Empirical Comparison}

We now empirically compare the behavior of these $k$-nearest-neighbor Kullback--Leibler estimators.
Consider two multivariate Gaussian distributions $P = \distNorm(\mu_1, \Sigma_1)$ and $Q = \distNorm(\mu_2, \Sigma_2)$. The theoretical value for the Kullback--Leibler divergence between $P$ and $Q$ is
\begin{equation}
	\kl{P}{Q} = \frac{1}{2}\left[\log\frac{|\Sigma_2|}{|\Sigma_1|} - d + tr(\Sigma_2^{-1}\Sigma_1) + (\mu_2-\mu_1)^T\Sigma_2^{-1}(\mu_2-\mu_1)\right],
	\label{eq:kl-gaussian}
\end{equation}
where $|\cdot|$ is the determinant of a matrix and $tr(\cdot)$ denotes the trace.
We generate samples $y_{1:\numobs}$ from $P$ and estimate $\klest{y_{1:\numobs}}{Q}$ with the three estimators above:
the canonical fixed $k$ estimator in \cref{eq:canonical-knn-kl-est} with $k\in\{1,10\}$,  the bias-corrected estimator in \cref{eq:biased-correct-knn-kl-est} with $k\in\{1,10\}$ and the adaptive estimator with $k_{\numobs} =  \numobs^{1/2}$.

We generate samples from a weakly correlated multivariate Gaussian distribution. Set $P = \distNorm(\mu, \Sigma)$ with $\mu=(1,\ldots,1)\in R^D$ and $\Sigma_{ij}=\exp\{-(i-j)^2/\sigma^2\}$, where large $\sigma$ results in high correlations and vice versa.
Let $\sigma=0.6$ and set $Q = \distNorm(0, I_{D})$.
We test the performance of each estimator with varying $\numobs \in \{100, 1000, 5000, 10000, 20000, 50000\}$ and varying dimensions $D \in \{4, 10, 25, 50\}$.

As shown in \cref{fig:knn-kl-est-comparison}, when $D=4$, the adaptive estimator with $k_{\numobs}=\numobs^{1/2}$ outperforms and shows reliable estimation when sample size is large ($\numobs \ge 5000$).
This scenario resembles the setup in our simulation and real-data experiments.
We therefore use the adaptive estimator with $k_{\numobs}=\numobs^{1/2}$ for our experiments.

When the dimension increases, the stability of all $k$-nearest-neighbor estimators drops due to the sparsity of data in high dimensions.
This reveals a limitation of all $k$-nearest-neighbor estimators.
Although proposing estimators for divergence is beyond the scope of the paper,
we test one possible adaption in \cref{sec:high-dim-simulation} to use the $k$-nearest-neighbor estimators in high dimensions by assuming independence across coordinates.

\begin{figure}[hp]
	\centering
	\subfloat[$d = 4$]{\includegraphics[width=.48\textwidth]{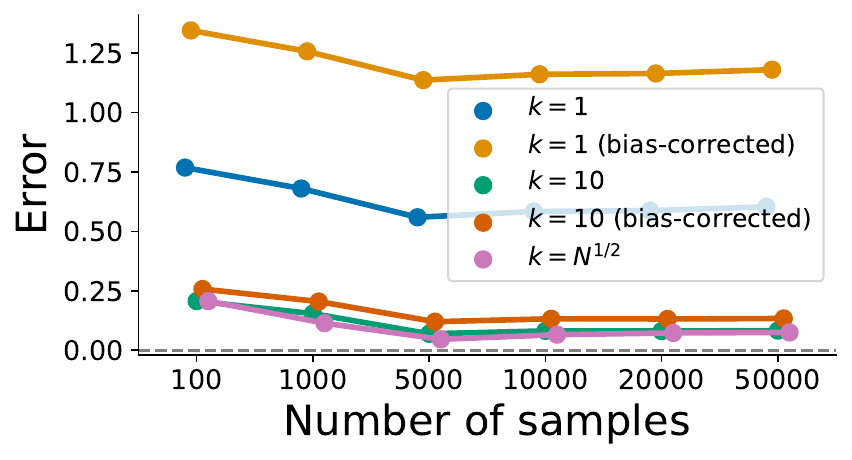}}
	\subfloat[$d = 10$]{\includegraphics[width=.48\textwidth]{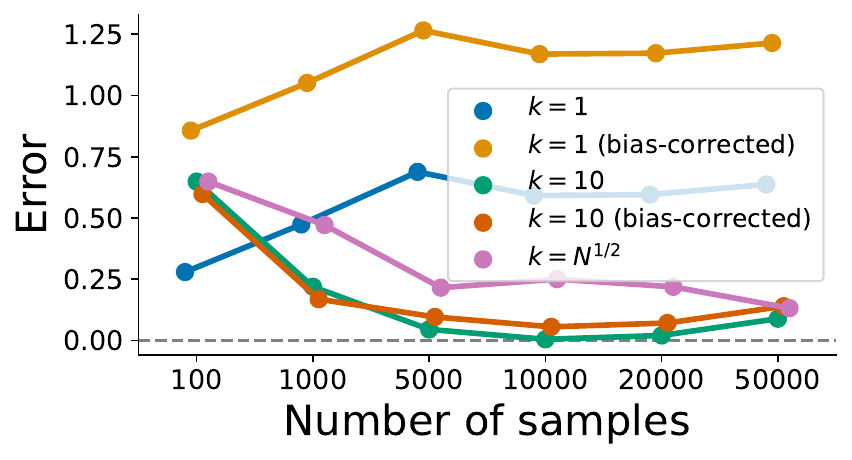}}
	\\
	\subfloat[$d = 25$]{\includegraphics[width=.48\textwidth]{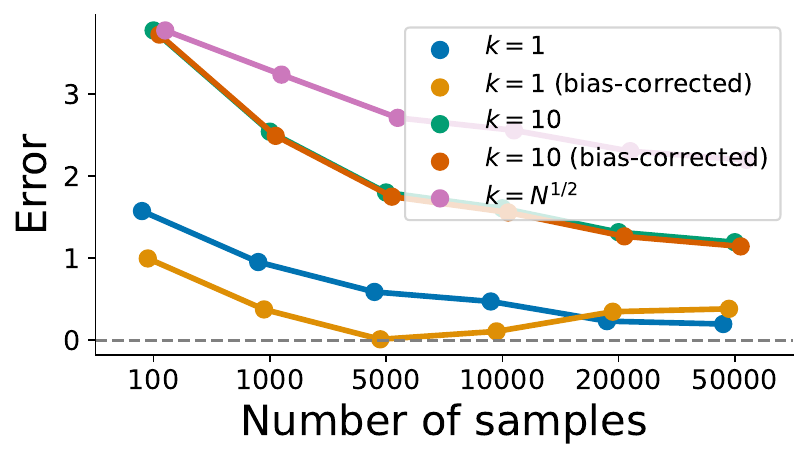}}
	\subfloat[$d = 50$]{\includegraphics[width=.48\textwidth]{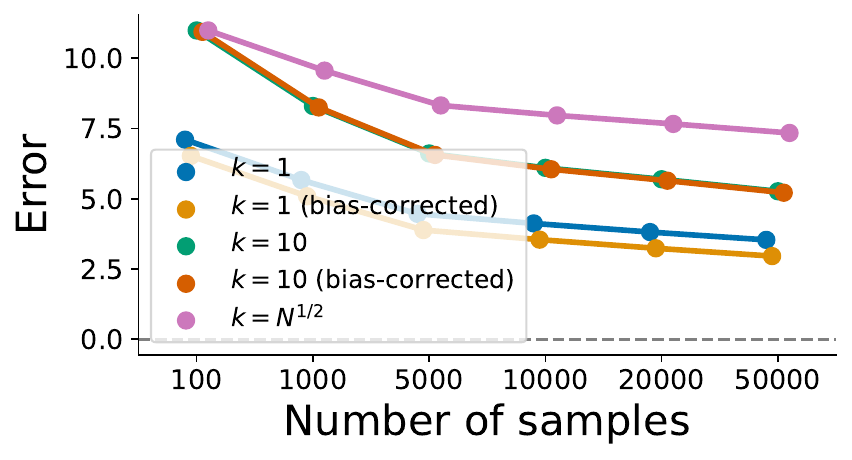}}
	\caption{Absolute error against sample size for canonical $1$-nearest-neighbor estimator, canonical $10$-nearest-neighbor estimator, bias-corrected $1$-nearest-neighbor estimator, bias-corrected $10$-nearest-neighbor estimator and adaptive $k_{\numobs}$-nearest-neighbor estimator with $k_{\numobs}=\numobs^{1/2}$. Each panel corresponds with a different dimension $D\in\{4,10,25,50\}$. Gray dotted lines indicate no error.}
	\label{fig:knn-kl-est-comparison}
\end{figure}

\section{Additional Simulation and Method Implementation Details}

\subsection{Sinkhorn Distance}
\label{sec:sinkhorn}
While the Wasserstein distance has appealing properties, it can be challenging to obtain an accurate estimate from finite samples
because it is sensitive to small changes in empirical distributions and suffers from slow convergence rates in high dimensions. 
The Sinkhorn distance, however, provides a regularized alternative that approximates the Wasserstein distance with faster sample convergence \citep{fast_sh}.
So, in practice, we can use the Sinkhorn distance to approximate the Wasserstein distance, 
which is the approach we take in \cref{sec:scRNA}. 
Specifically, we use the unbalanced Sinkhorn distance \citep{unbsinkhorn}, which solves an unbalanced optimal transport (OT) problem
 in the discrete setting.
 Given samples $x_{1},\dots,x_{N}$ and $y_{1},\dots,y_{L}$, we can construction the 
cost matrix $M \in \reals^{N \times L}$ for a metric $m$ defined by $M_{n\ell} = m(x_{n}, y_{\ell})$. 
Let $U$ denote the set of transport plans 
\[
U = \left\{ A \in \mathbb{R}_+^{D \times D} \mid {1}_D^T A {1}_D = 1 \right\},
\]
where $1_{D}$ denotes the $D$-dimension vector with all components equal to 1. 
Given nonnegative regularization constants $\veps$ and $\varrho$, for $r \in \Delta_{N}$ and $c \in \Delta_{L}$,
the unbalanced Sinkhorn distance is defined as
\[
	d_{M,\varepsilon, \varrho}(r, c) = \min_{A\in U
    } \mathrm{tr}(A^{\top}M) + \varepsilon\,\kl{A}{rc^T}
	+ \varrho\,\kl{A 1}{r}
	+ \varrho\,\kl{A^T 1}{c}. 
\]
A larger value of $\varepsilon$ encourages a smoother, more numerically stable transport plan by penalizing the divergence between 
the plan $A$ and the independence (maximum entropy) transport plan $rc^T$. 
Smaller marginal penalty $\varrho$ introduces robustness to the marginal constraints of the transport plan.
The so-called balanced OT is retrieved in the limit of $\rho \to \infty$.
Additionally taking $\varepsilon\to0$ recovers the unregularized OT, which is equal to the
empirical 1-Wasserstein distance. 

\subsection{Limitations of Coarsening}
\label{sec:coarsening-limitations}

The following toy example illustrates how coarsening and the Bayesian information criterion can overfit even in scenarios with only a modest degree of misspecification.
	We generate data from a mixture of $\numcomps_{o} = 2$ skew normal distributions but fit the data using a Gaussian mixture.
	The level of misspecification is controlled by the skewness parameter of each skew normal component in the true generative distribution $P_{o}$.
	We consider the following scenarios: two equal-sized clusters with the same level of misspecification (denoted \texttt{same})
	and two equal-sized clusters with different levels of misspecification (denoted \texttt{different}).
	See \cref{sec:high-dim-simulation} for further details about the experimental set-up.
	As shown in the first row of \cref{fig:motivate-comparison}, in both scenarios using expectation--maximization (EM)
	with the Bayesian information criterion (BIC) results in estimating $\numcomps \gg \numcomps_{o}$
	to capture the skewness of each component.
	As shown in the second row of \cref{fig:motivate-comparison},
	the coarsened posterior performs well in the \texttt{same} case but overfits the cluster with a larger degree of misspecification in the \texttt{different} scenario.

\begin{figure}[tp]
	\centering
	\includegraphics[width=.48\textwidth]{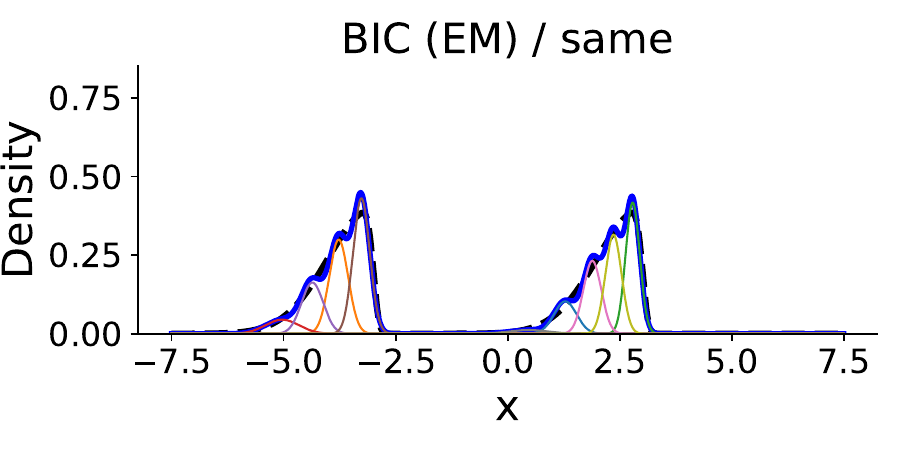}
	\includegraphics[width=.48\textwidth]{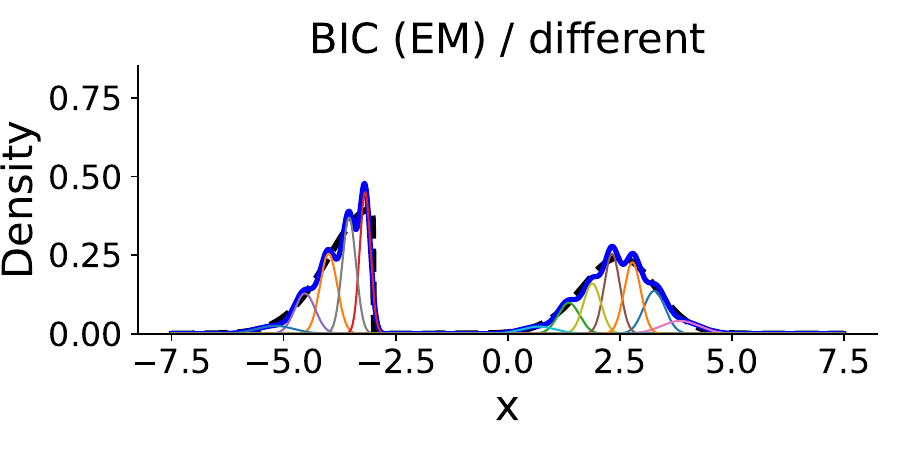}\\
	\includegraphics[width=.48\textwidth]{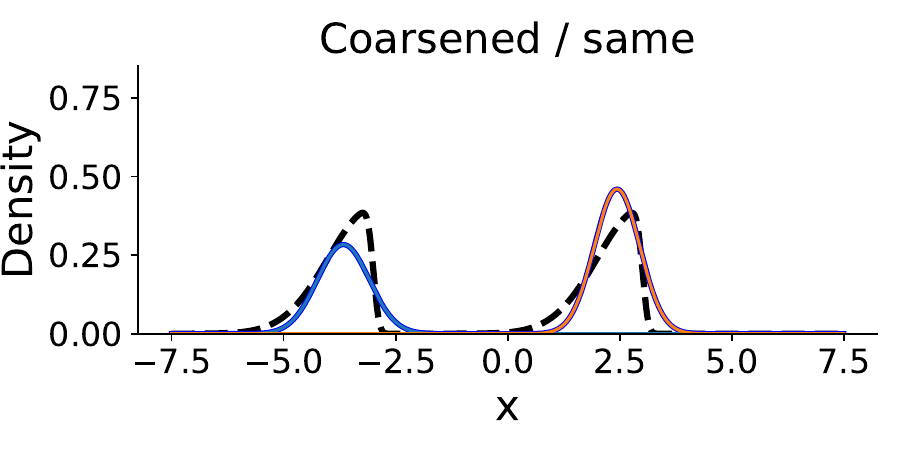}
	\includegraphics[width=.48\textwidth]{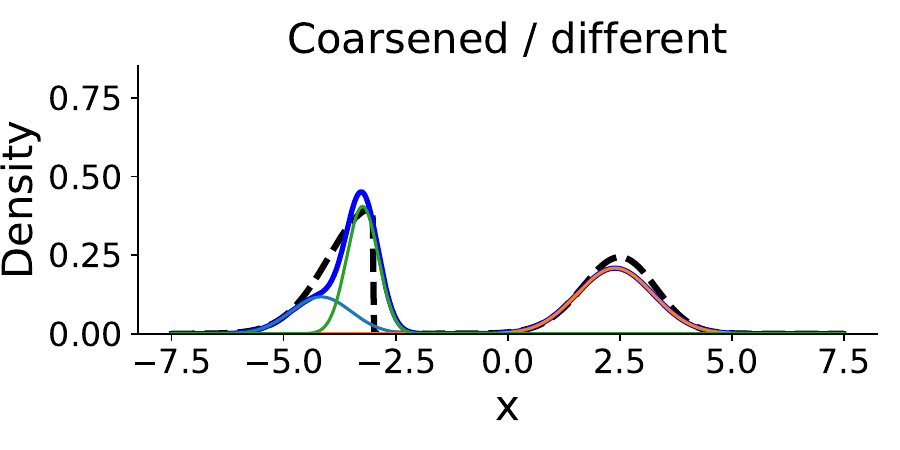}\\
	\includegraphics[width=.48\textwidth]{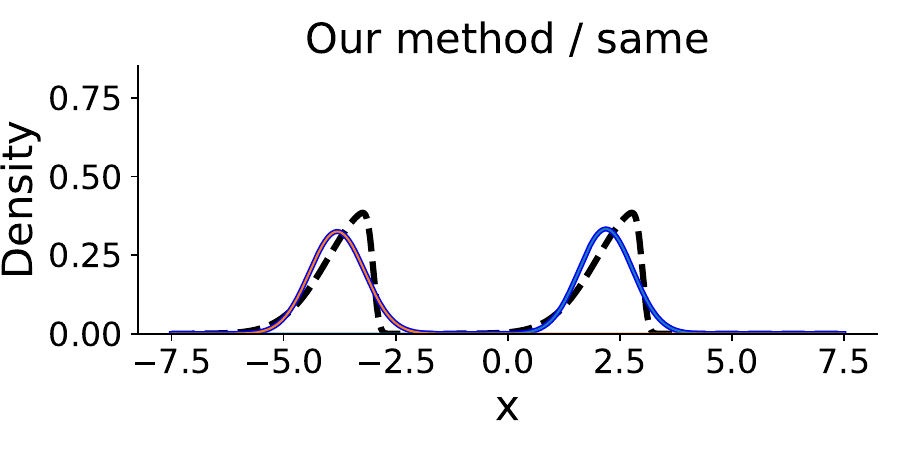}
	\includegraphics[width=.48\textwidth]{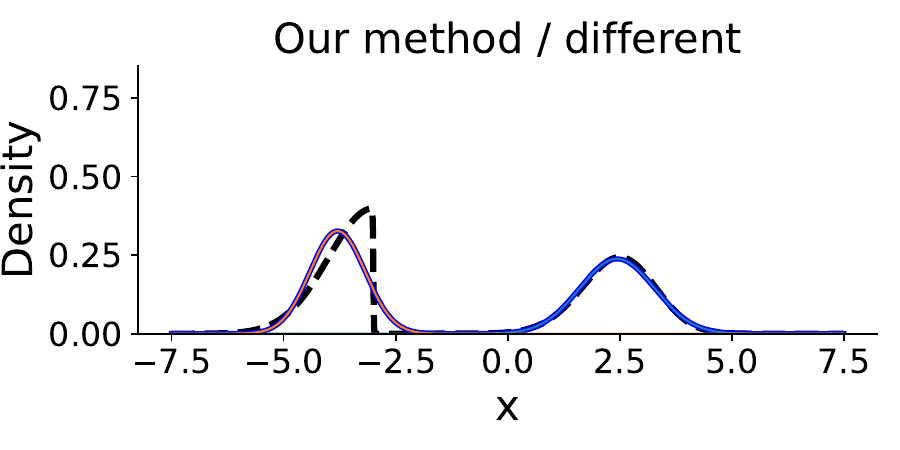}
	\caption{
		For the mixture of skew-normals example from \cref{sec:intro},
		each panel shows the density of $P_{o}$ (dashed lines) and the densities of the fitted Gaussian mixture model and each
		component distribution (solid lines) using $\numobs = 10\,000$ observations.
		The ``same'' and ``different'' scenarios describes the relative degree of misspecification of the two components.
		Results are given for three approaches:
		expectation--maximization with the Bayesian information criterion (first row),
		the coarsened posterior (second row),
		and our robust model selection method (third row).}
	\label{fig:motivate-comparison}
\end{figure}

\subsection{Coarsened Posterior Calibration Details}
\label{appx:simulation-gauss}

The coarsened posterior requires calibration of the hyperparameter $\alpha$, which determines the degree of misspecification.
We select $\alpha$ using the \emph{elbow method} proposed by \citet{Miller:2019}.
In this section, we include all calibration figures for the coarsened posterior following the code provided by \citet{Miller:2019}.

As shown in \cref{fig:coarsen-calibration}, we set $\alpha$ based on the turning point where we see no significant increase
in the log-likelihood if $\alpha$ increases further.
Using these values for $\alpha$, we can see for all cases except the \texttt{small-large} case, the coarsened posterior consistently estimates the number of clusters (after removing mini clusters with size $<2\%$) as $\widehat{\numcomps} = 3 > \numcomps_{o} = 2$.

\begin{figure}[tp]
	\centering
	\includegraphics[width=.48\textwidth]{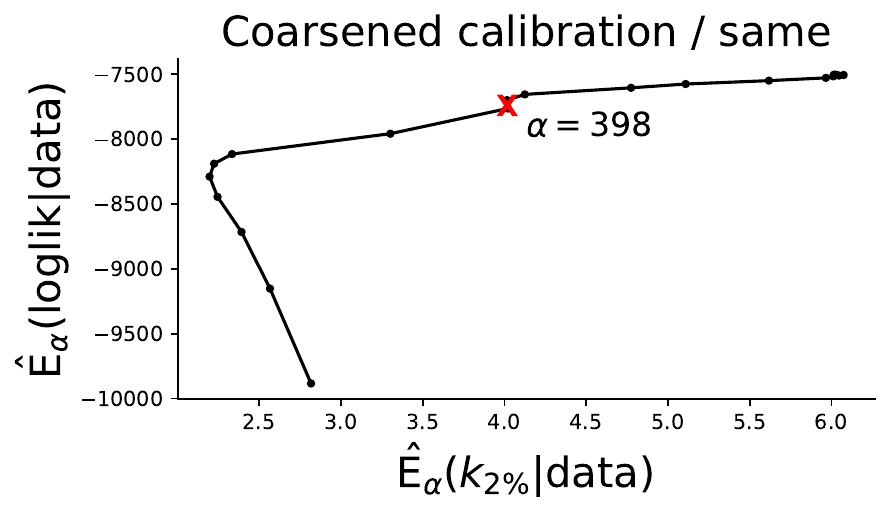}
	\includegraphics[width=.48\textwidth]{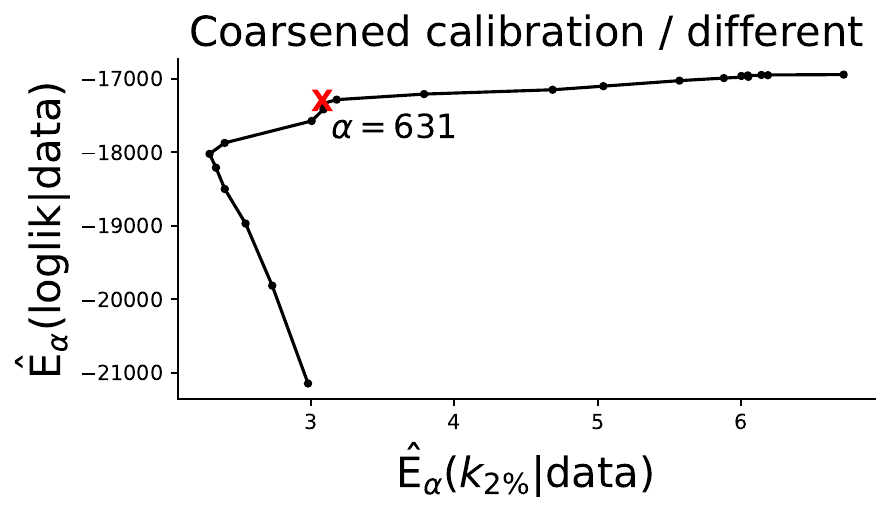}\\
	\includegraphics[width=.48\textwidth]{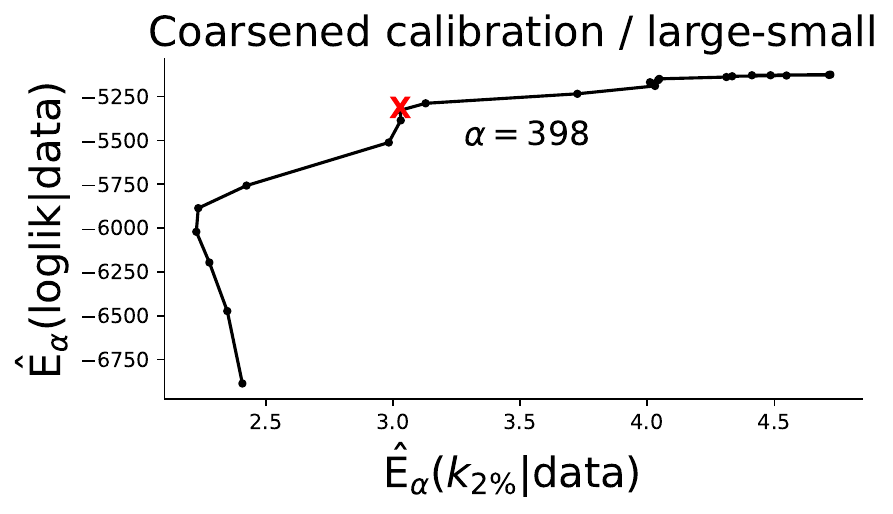}
	\includegraphics[width=.48\textwidth]{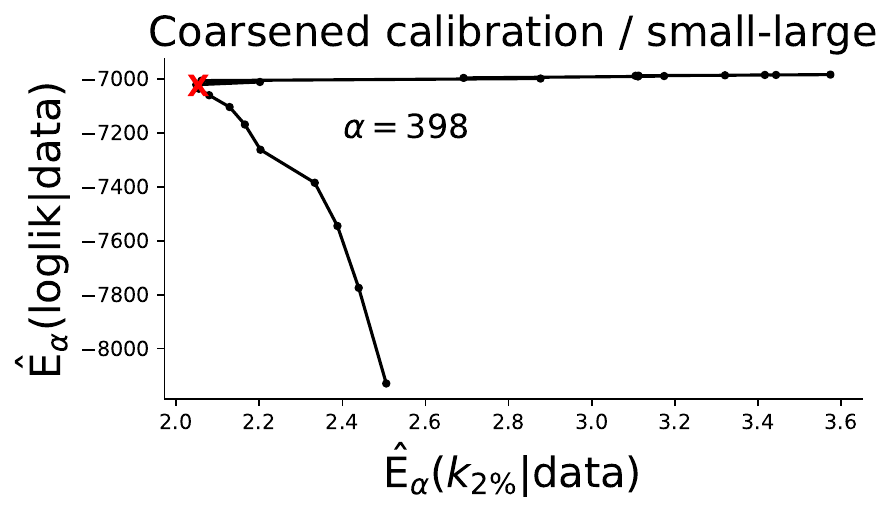}\\
	\includegraphics[width=.48\textwidth]{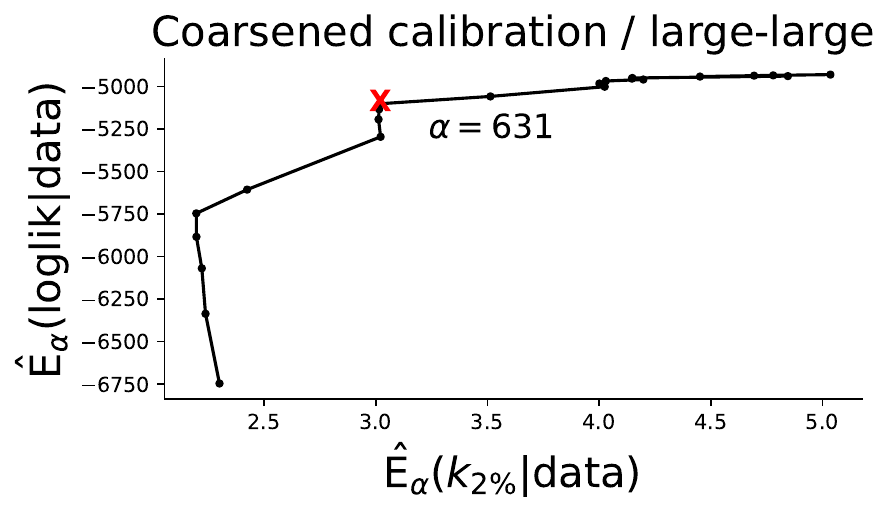}
	\caption{For the mixture of skew normals experiments from \cref{sec:intro,sec:high-dim-simulation}, each panel shows the expected log-likelihood $\widehat{E}_{\alpha}(\mathrm{loglik}\mid \mathrm{data})$ against the expected number of clusters which excludes tiny clusters of size less that $2\%$ of whole dataset denoted as $\widehat{E}_{\alpha}(k_{2\%}\mid \mathrm{data})$.
		We select $\alpha$ as the elbow point in the plots.}
	\label{fig:coarsen-calibration}
\end{figure}

\section{Additional Experimental Details and Results for Mixture Model}
\label{sec:case-study-details}

\subsection{Skew-normal Mixture Simulation Study}
\label{sec:simulation-gauss}

We first consider the case of two clusters of equal size. 
We set $\eta_{o} = (0.5, 0.5)$, $\mu_{o} = (-3,3)$, and $\sigma_{o} = (1,1)$ for the two scenarios in \cref{fig:motivate-comparison}: $\gamma_o = (-10,-1)$ (denoted \texttt{different}) and $\gamma_o = (-10,-10)$ (denoted \texttt{same}).
We also compare \methodname to the coarsened posterior with data from
two-component mixtures of different cluster sizes. %
We set $\eta_{o} = (0.95, 0.05)$, $\mu_{o} = (-3,3)$, and $\sigma_{o} = (1,1)$ for the following three scenarios:
$\gamma_o = (-10,-1)$ (denoted \texttt{large-small}),  $\gamma_o = (-1,-10)$ (denoted \texttt{small-large}),
and  $\gamma_o = (-10,-10)$ (denoted \texttt{large-large}).

\begin{figure}[tp]
	\centering
	\subfloat{\label{fig:gauss-cpos-pdfs-1}\includegraphics[width=.32\textwidth]{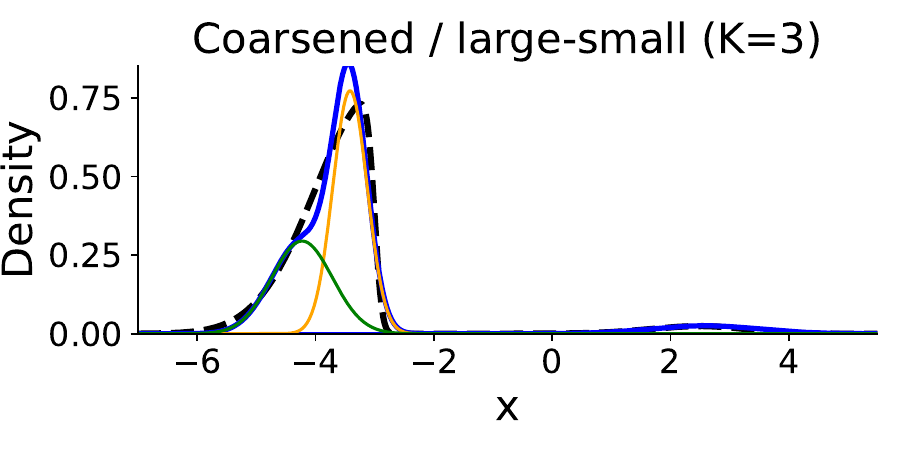}}
	\subfloat{\label{fig:gauss-cpos-pdfs-2}\includegraphics[width=.32\textwidth]{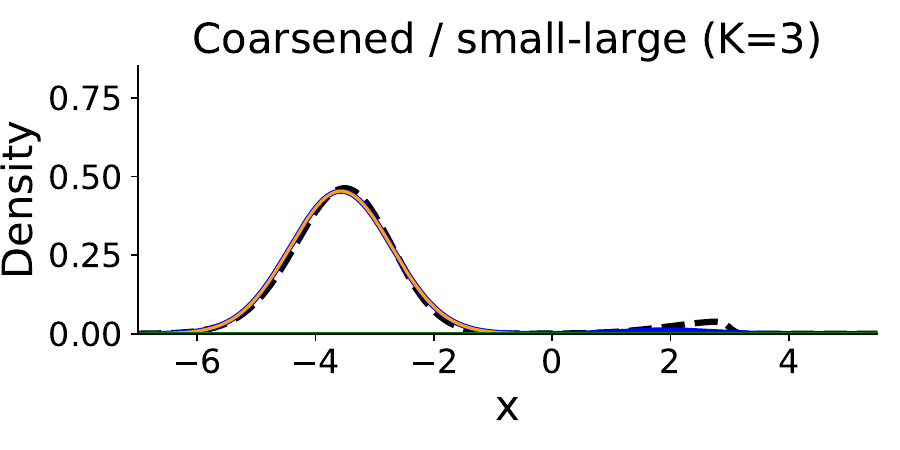}}
	\subfloat{\label{fig:gauss-cpos-pdfs-3}\includegraphics[width=.32\textwidth]{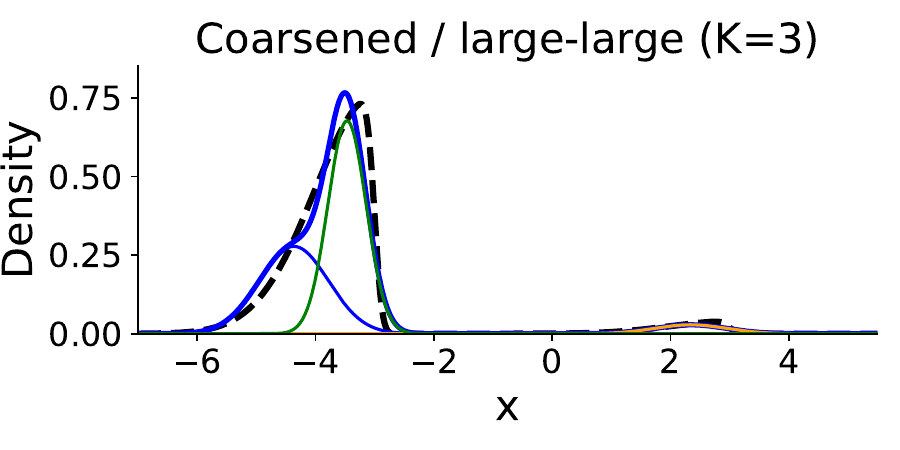}}
	\\
	\subfloat{\label{fig:gauss-stare-pdfs-1}\includegraphics[width=.32\textwidth]{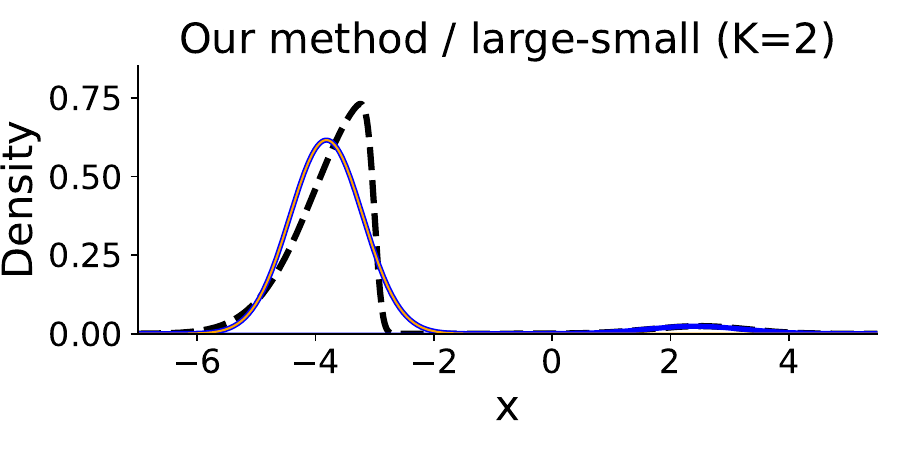}}
	\subfloat{\label{fig:gauss-stare-pdfs-2}\includegraphics[width=.32\textwidth]{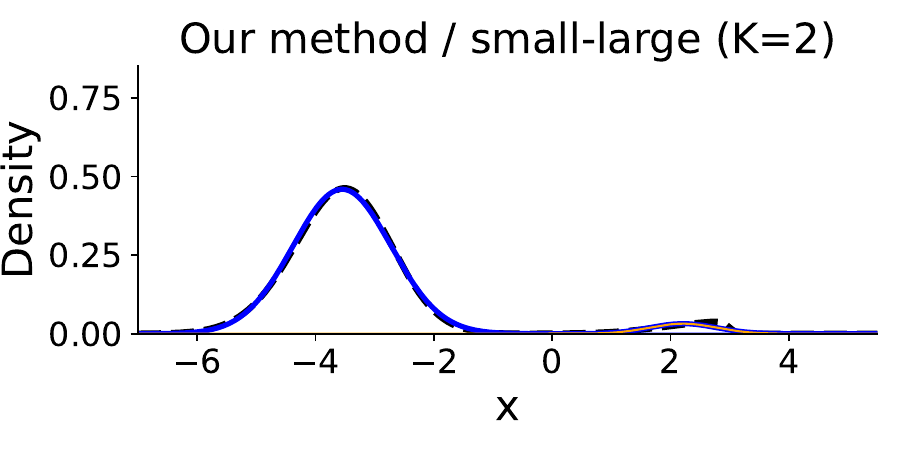}}
	\subfloat{\label{fig:gauss-stare-pdfs-3}\includegraphics[width=.32\textwidth]{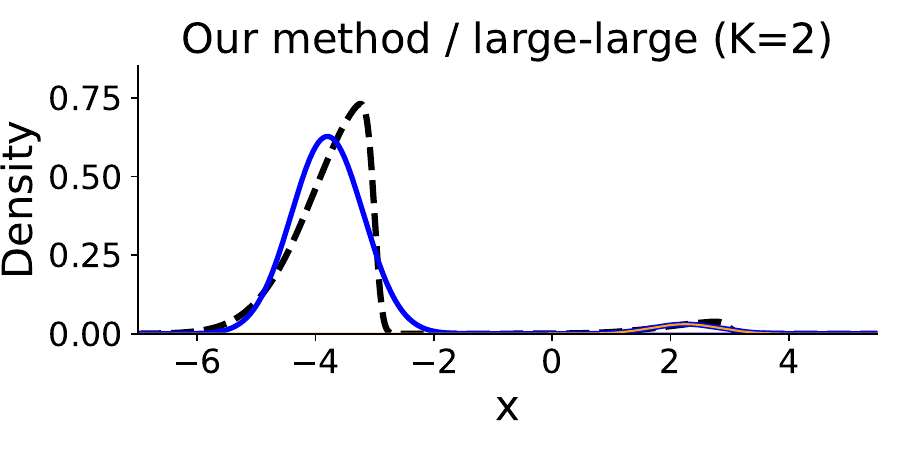}}
	\\
	\subfloat{\label{fig:gauss-stare-penloss-1}\includegraphics[width=.32\textwidth]{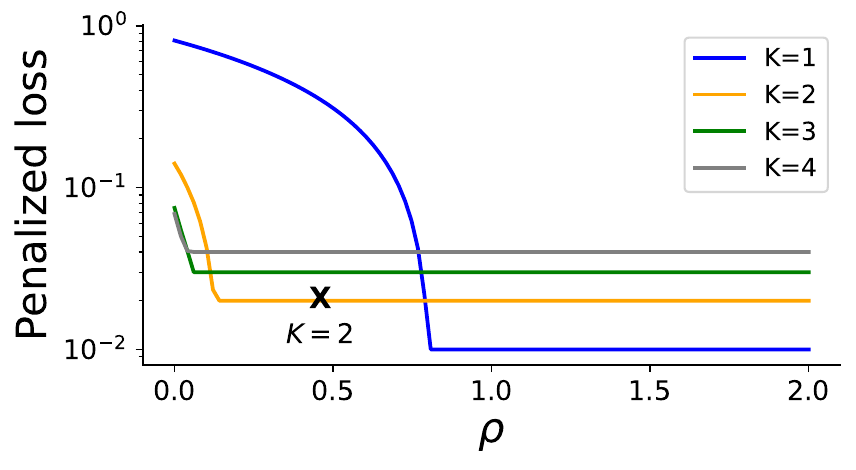}}
	\subfloat{\label{fig:gauss-stare-penloss-2}\includegraphics[width=.32\textwidth]{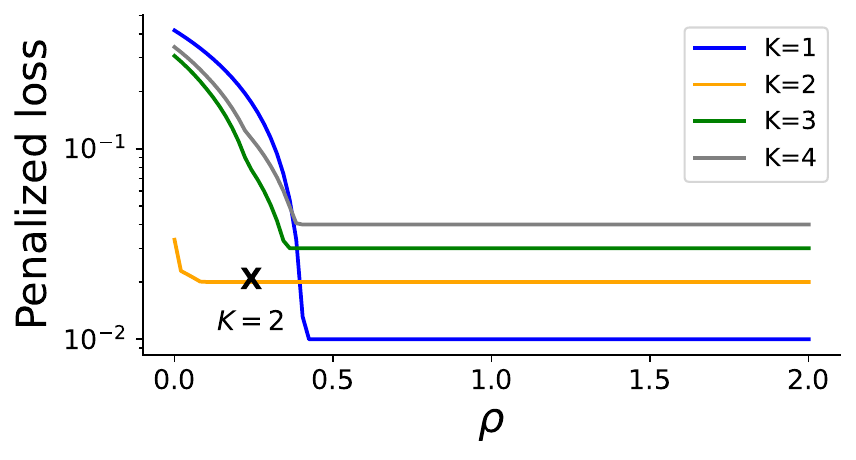}}
	\subfloat{\label{fig:gauss-stare-penloss-3}\includegraphics[width=.32\textwidth]{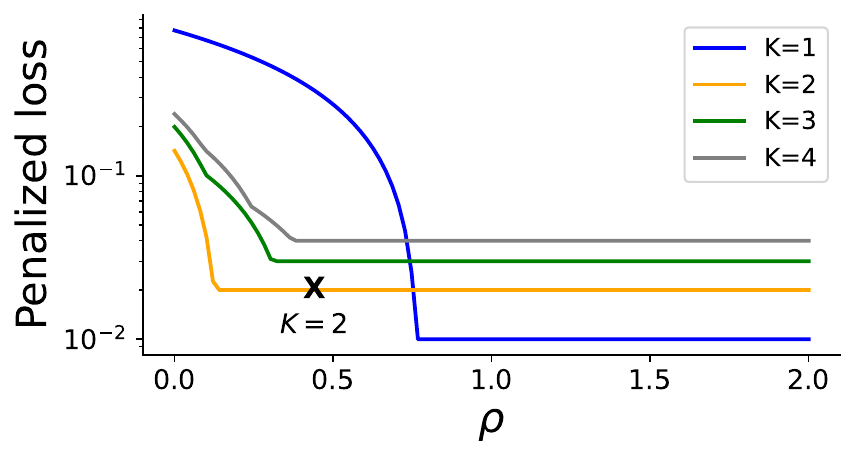}}

	\caption{
		Comparison between the coarsened posterior and \methodname when using a Gaussian mixture model to fit data generated from a mixture of
		skew-normal distributions.
		\textbf{First row:} Densities of the model and components selected using the coarsened posterior (solid lines) and the density of the data distribution (dashed line). The title specifies the data-generating distribution and the number of components selected.
		In the middle plot of the first row, the minor cluster contains two components.
		\textbf{Second row:} Densities of the model and components selected using our structurally aware robust method.
		\textbf{Third row:} Penalized loss plots, where the cross mark indicates the first wide stable region and is labeled with the number of clusters \methodname selects.
		Line colors correspond to different $\numcomps$ values.}
	\label{fig:gauss-simulation}
\end{figure}

For the coarsened posterior, we calibrate the hyperparameter $\alpha$ following the procedure from \citet[Section 4]{Miller:2019}.
First, we use Markov chain Monte Carlo to approximate the coarsened posterior for $\alpha$ values ranging from $10$ to $10^5$.
Then, we select the coarsened posterior with the $\alpha$ value at the clear cusp that indicates a good fit and low complexity.
See the Supplementary Materials for further details and calibration plots.
As shown for \texttt{large-small} and \texttt{large-large} in \cref{fig:gauss-simulation}, when the
larger cluster has large misspecification, the coarsened posterior introduces one additional cluster to explain the larger cluster.
For the \texttt{small-large} case, when the larger cluster exhibits a small degree of misspecification,
the coarsened posterior introduces one additional cluster to explain the smaller cluster.

\methodname correctly calibrates the model mismatch cutoff $\rho$ using the penalized loss plots shown in \cref{fig:gauss-simulation},
as in all cases $\numcomps = 2$ corresponds to the first wide, stable region.
By the density plots in the middle column, we can see that \methodname 
is able to properly trade off a worse density estimate for better model selection.

\paragraph{Comparison to other methods}
To compare \methodname to other mixture model selection criteria, we generate synthetic datasets from mixtures of multivariate skew-normal distributions under both low- and high-dimensional cases.  
The density of multivariate skew normal distribution is $f(x; m, \Sigma, \gamma) = 2\phi(x; m, \Sigma)\Phi(\gamma \odot x; m, \Sigma)$, where $\phi(x; m, \Sigma)$ and $\Phi(x; m, \Sigma)$ are, respectively, the probability density function
and cumulative distribution function of the multivariate normal $\distNorm(m, \Sigma)$.
Cluster means are first placed on an evenly spaced grid and then, for each component, a mean vector is sampled from a multivariate normal distribution centered at the corresponding grid location, and then randomly shuffled within each dimension.
Mixture weights are sampled from a Dirichlet distribution and adjusted to ensure a minimum component proportion of $0.05$, and cluster-specific skewness parameters are drawn from normal distributions with means evenly spaced between $2.0$ and $8.0$.  
To introduce controlled variation in cluster covariance structure, we compute, for each cluster mean, its minimum separation distance to every other cluster and multiply this distance by a scalar $\alpha \in (0,1)$. This ensures that more isolated clusters receive larger variance scales, while clusters that are closer together are assigned smaller ones.
For the low-dimensional experiments, we generate 60 datasets with $K_{o} \in \{3,4\}$, $D \in \{2,3\}$, and fifteen replicates per configuration, each containing $N = 1000$ samples.  
For the high-dimensional experiments, we generate 40 datasets with $K_{o} \in \{8,9\}$, dimension $D = 20$, and twenty replicates per configuration, each containing $N = 5000$ samples.

\begin{table}[tp]
\centering
\caption{Wilcoxon signed-rank test $p$-values comparing absolute errors of
\methodname\ against common model selection criteria.}
\label{tab:wilcoxon-pvals}
\begin{tabular}{lcc}
\toprule
\textbf{Comparison} & \textbf{Low-dim} & \textbf{High-dim} \\
\midrule
\methodname\ vs. Gap        & $4.07\times 10^{-4}$ & 0.028890 \\
\methodname\ vs. Elbow      & $8\times 10^{-6}$ & 0.002188 \\
\methodname\ vs. Silhouette & $2\times 10^{-6}$ & 0.158655 \\
\bottomrule
\end{tabular}
\end{table}

\subsection{Flow Cytometry Calibration Results}
\label{appx:flow-cytometry}

In this section, we include loss and F-measure plots of our model selection method on all test datasets 7--12.
See \citet[Section 5.2]{Miller:2019} for a discussion of the exact calibration procedure for the coarsened posterior.

Recall that to calibrate $\rho$, we select $\rho$ that optimizes the F-measure across first $6$ datasets.
To incorporate this prior knowledge on test datasets, we suggests selecting the value of $\numcomps$
that has has stable penalized loss and is closest to the optimal $\rho$.
We compare our selection $\widehat{\numcomps}$ with the ground truth $\numcomps_o$ labeled by experts.
For each dataset, there is always one cluster labeled as unknown due to some unclear information for cells.
With automatic clustering algorithm, it is natural for the algorithm to identify those unlabeled points and assign them to other clusters, which results in $\numcomps_o-1$ clusters.
So we treat both $\numcomps_o$ and $\numcomps_o-1$ as ground truth in our analysis.
As shown in \cref{fig:GvHD,fig:GvHD2}, our selection method results in highest F-measure for datasets 8--12.
Dataset 7 is challenging and even the ground truth does not produce a large F-measure.

\begin{figure}[tp]
	\centering
	\subfloat[Data 7]{\label{fig:data7}
		\includegraphics[width=.48\textwidth]{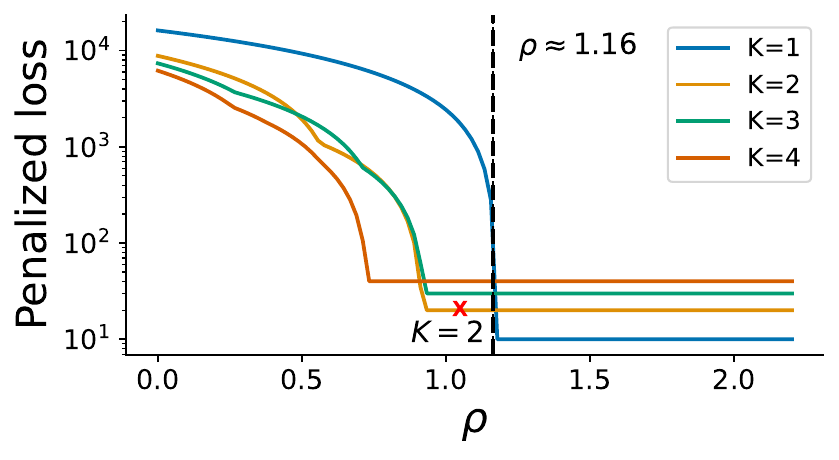}
		\includegraphics[width=.48\textwidth]{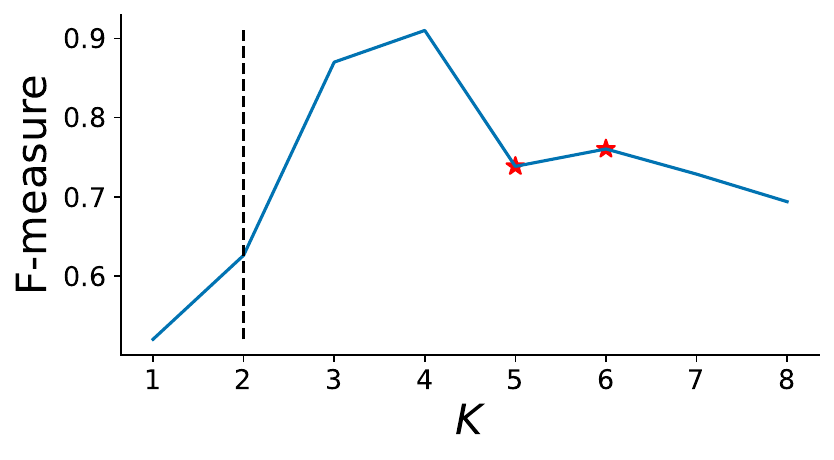}}	\\
	\subfloat[Data 8]{\label{fig:data8}
		\includegraphics[width=.48\textwidth]{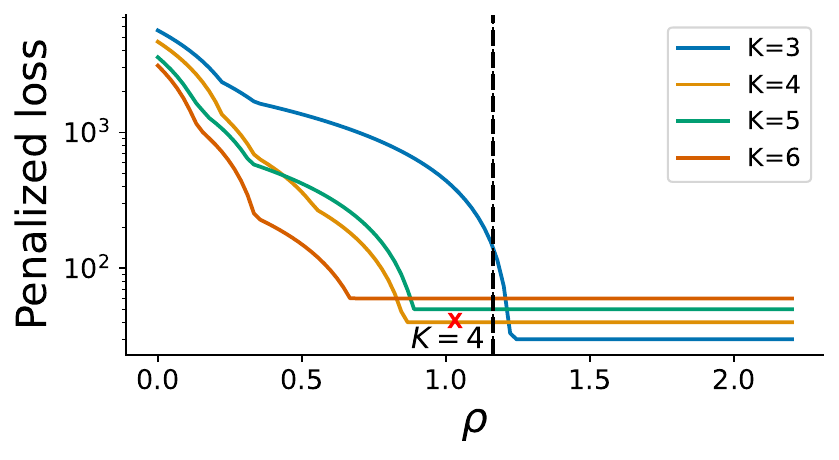}
		\includegraphics[width=.48\textwidth]{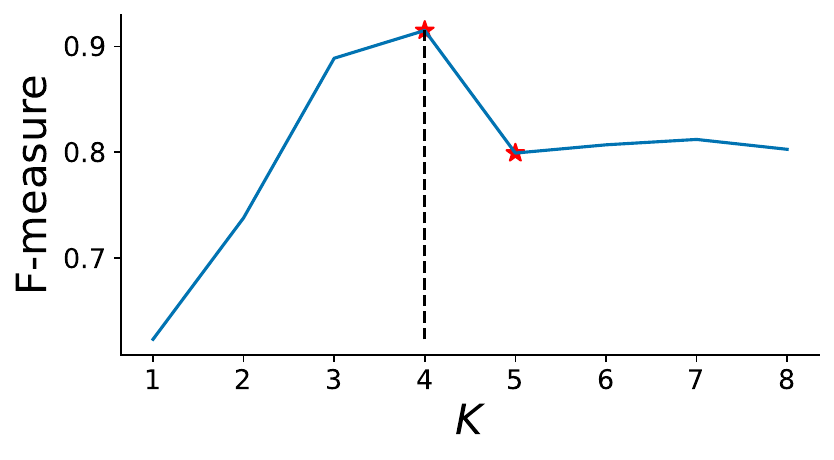}}	\\
	\subfloat[Data 9]{\label{fig:data9}
		\includegraphics[width=.48\textwidth]{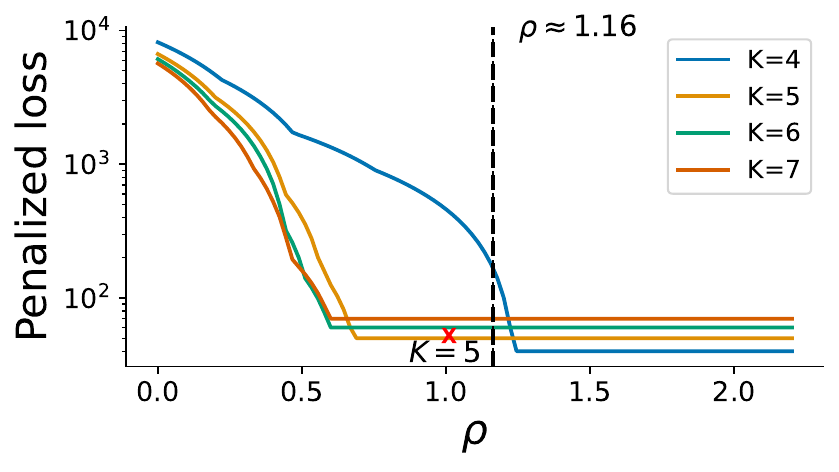}
		\includegraphics[width=.48\textwidth]{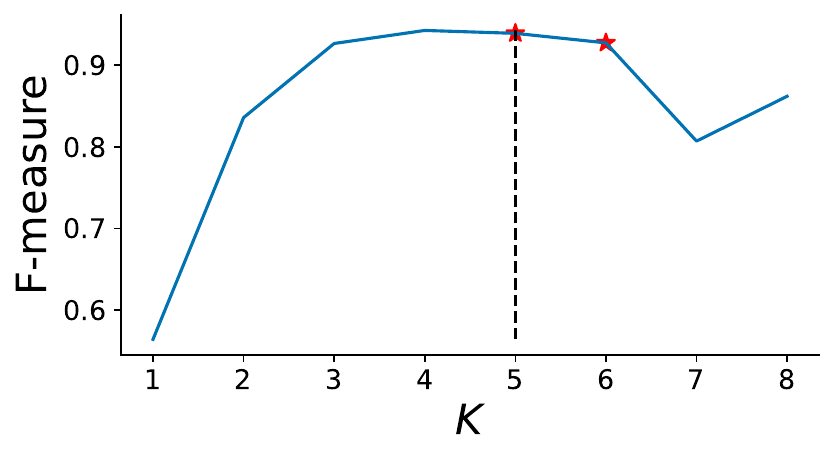}}
	\caption{Calibration and F-measure plots for test datasets 7--9 in flow cytometry experiments. \textbf{Left}: The black dashed lines indicate the optimal $\rho$ calibrated on training datasets 1--6. The cross mark indicates the selection for number of clusters. \textbf{Right}: F-measure against the number of clusters. The dashed line shows the number of clusters selected by \methodname and the red star indicates the ground truth $\numcomps_{o}$.}
	\label{fig:GvHD}
\end{figure}

\begin{figure}[tp]
	\centering
	\subfloat[Data 10]{\label{fig:data10}
		\includegraphics[width=.48\textwidth]{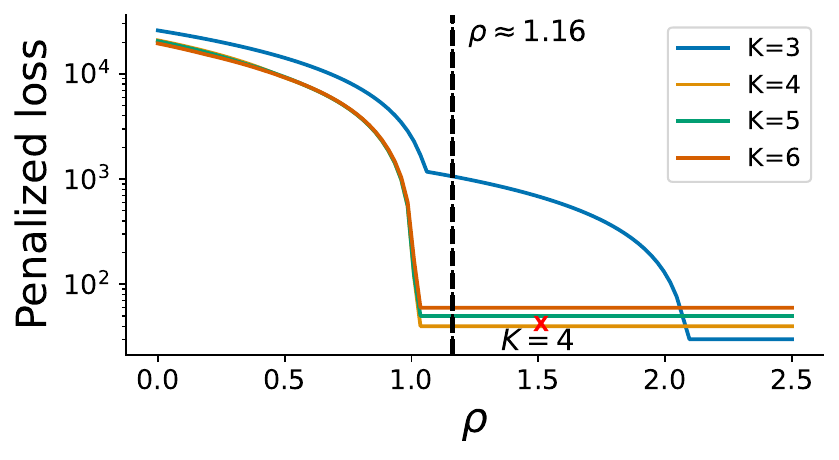}
		\includegraphics[width=.48\textwidth]{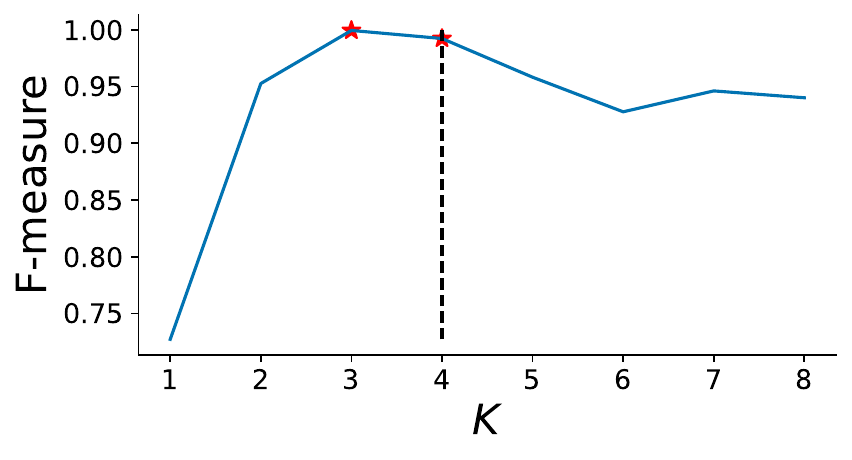}}	\\
	\subfloat[Data 11]{\label{fig:data11}
		\includegraphics[width=.48\textwidth]{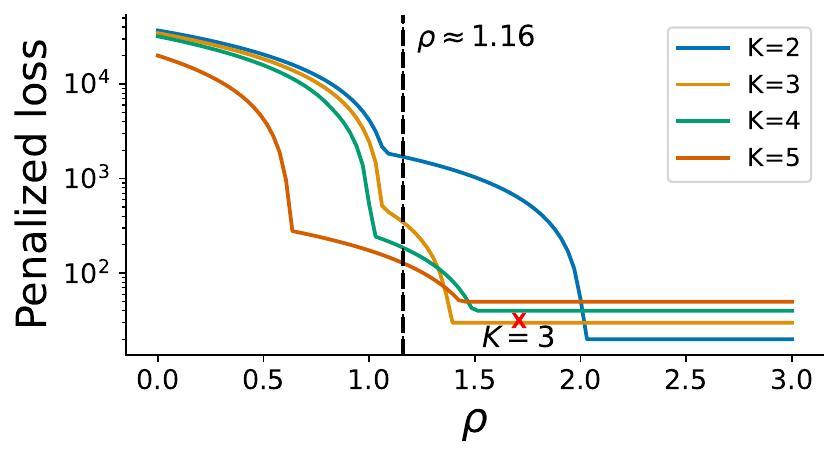}
		\includegraphics[width=.48\textwidth]{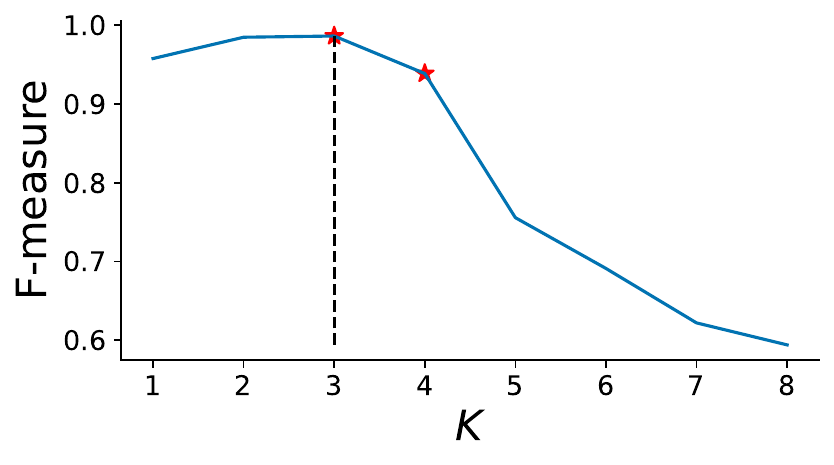}}	\\
	\subfloat[Data 12]{\label{fig:data12}
		\includegraphics[width=.48\textwidth]{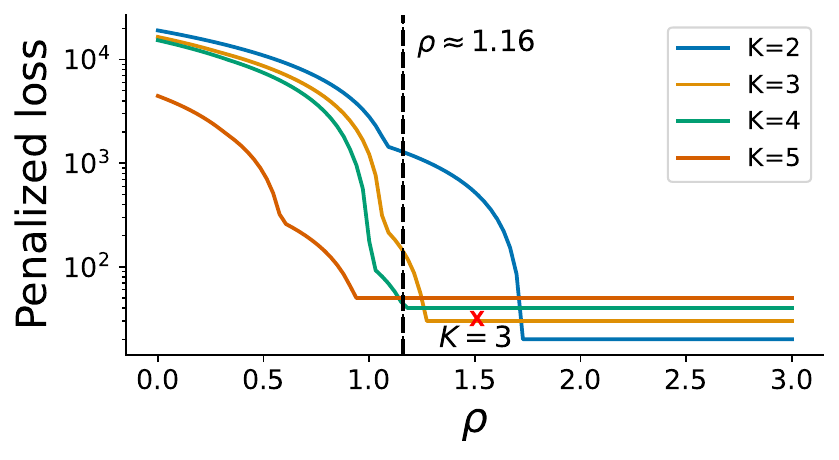}
		\includegraphics[width=.48\textwidth]{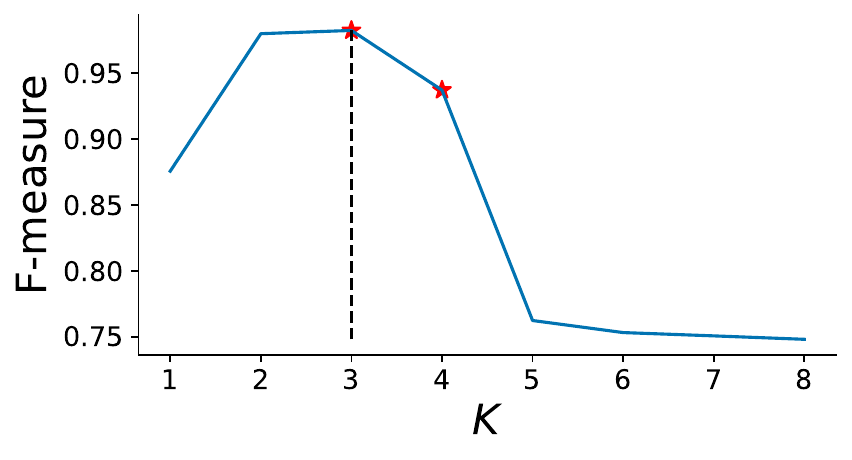}}
	\caption{Calibration and F-measure plots for test datasets $10-12$ in flow cytometry experiments. See caption for \cref{fig:GvHD} for details}
	\label{fig:GvHD2}
\end{figure}

\begin{figure}[tp]
	\centering
	\subfloat[Data 7, $\numcomps=8$]{\label{fig:data7-loss}
		\includegraphics[width=.48\textwidth]{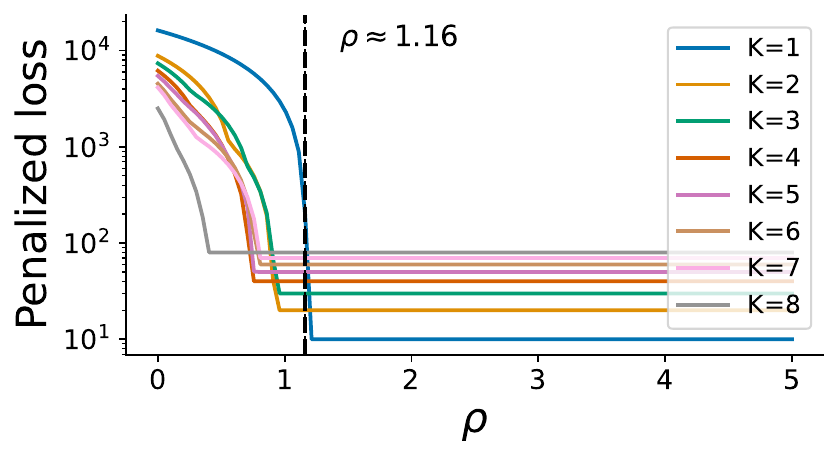}}
	\subfloat[Data 8, $\numcomps=7$]{\label{fig:data8-loss}
		\includegraphics[width=.48\textwidth]{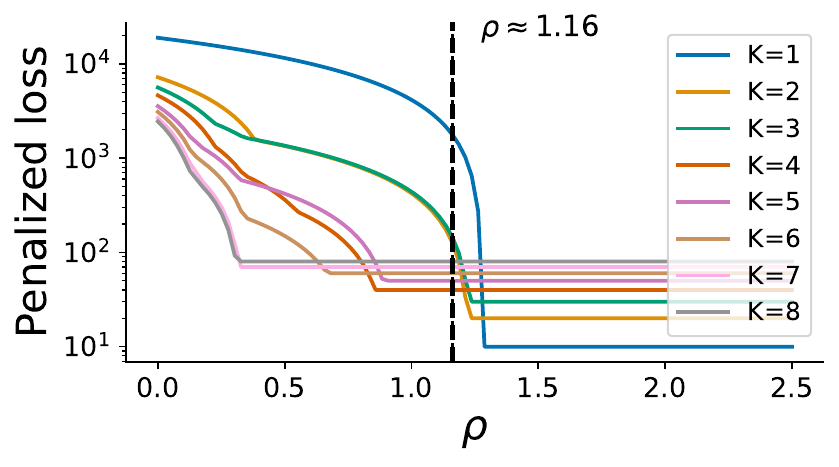}}\\
	\subfloat[Data 9, $\numcomps=5$]{\label{fig:data9-loss}
		\includegraphics[width=.48\textwidth]{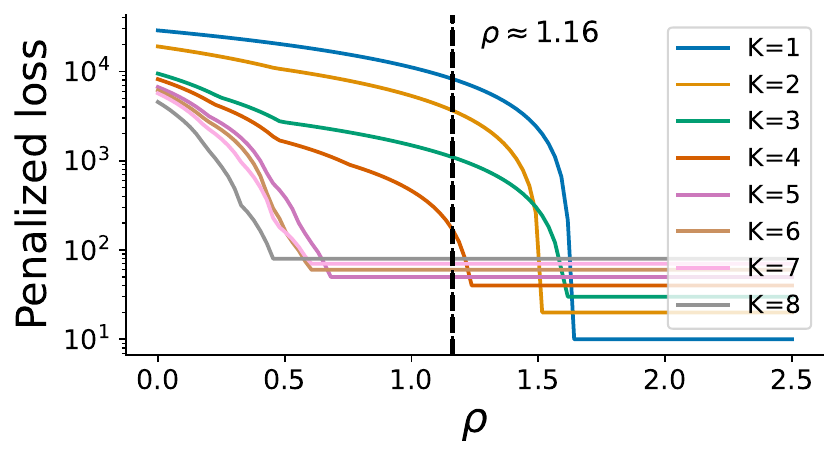}}
	\subfloat[Data 10, $\numcomps=4$]{\label{fig:data10-loss}
		\includegraphics[width=.48\textwidth]{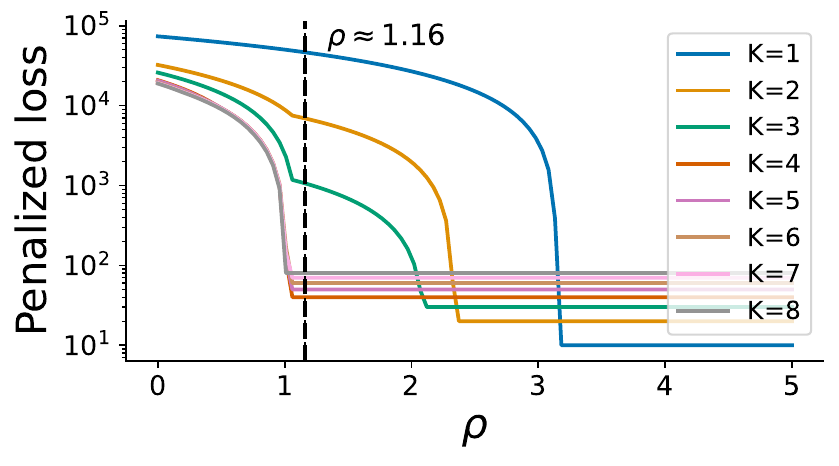}}\\
	\subfloat[Data 11, $\numcomps=3$]{\label{fig:data11-loss}
		\includegraphics[width=.48\textwidth]{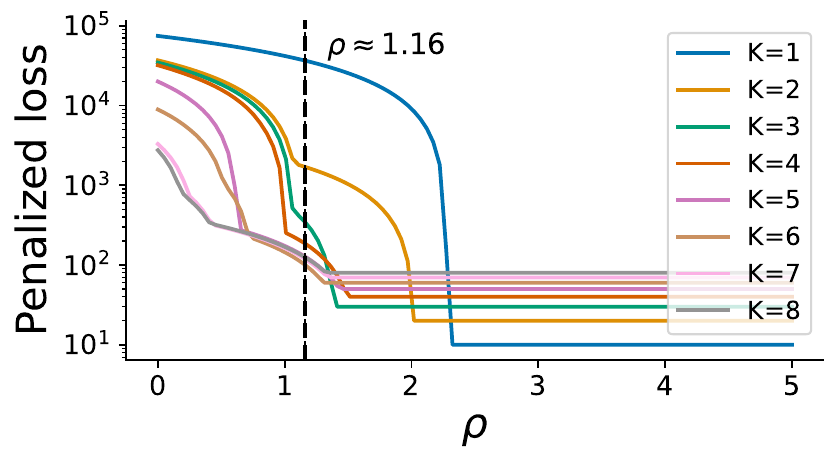}}
	\subfloat[Data 12, $\numcomps=3$]{\label{fig:data12-loss}
		\includegraphics[width=.48\textwidth]{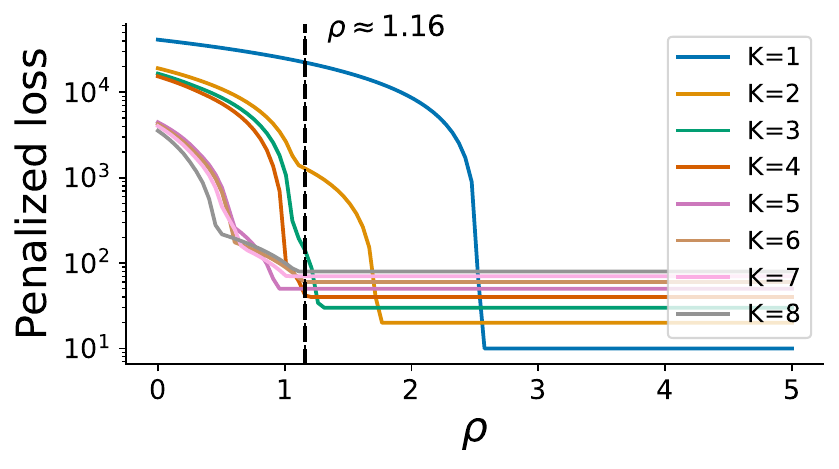}}
	\caption{Calibration including $\numcomps=1,\ldots,8$ for test datasets 7--12 in flow cytometry experiments. }
	\label{fig:GvHD3}
\end{figure}

\begin{figure}[tp]
	\centering
    \subfloat{\includegraphics[width=0.54\textwidth]{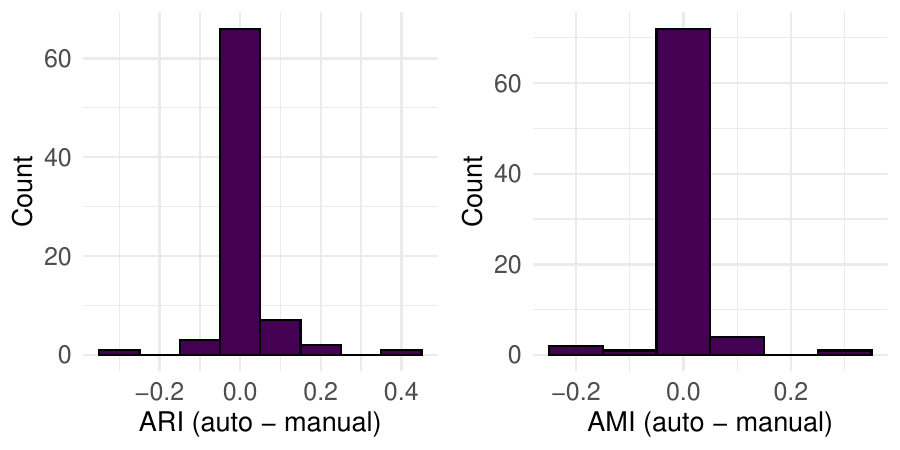}}
    \subfloat{\includegraphics[width=0.36\textwidth]{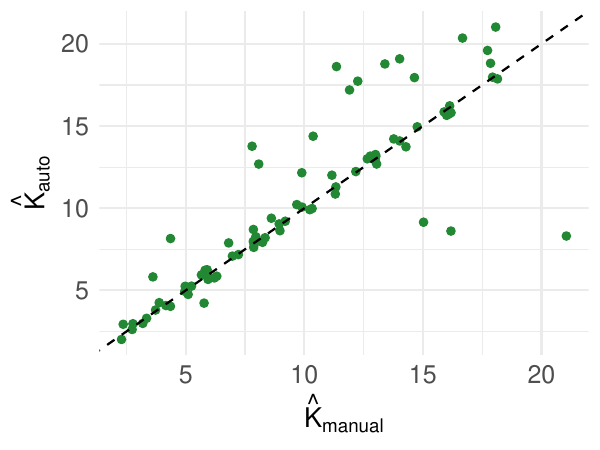}}

	\caption{Comparison between manual and automated selection of $\rho$ for the \textsf{uniform} scRNA-seq datasets. 
        \textbf{Left:} Pairwise difference in ARI. 
		\textbf{Middle:} Pairwise difference in AMI.
		\textbf{Right:}
        Manually vs. automatically estimated number of cell types in uniform datasets (with x-axis jitter).} 
	\label{fig:auto_comp_unif}
\end{figure}

\subsection{Tabula Muris: Data Overview and Processing}
\label{sec:rna-data}

Tabula Muris \citep{mice}
is a comprehensive collection of single-cell transcriptome data.
The gene count tables are derived from SMART-Seq2 RNA-seq libraries and consist
of 53,760 cells from 20 different organs and tissues of 8 mice. 
We subsampled datasets from the Tabula Muris data while controlling for the number of cell types and the number of cell observations in each cell type.
We constructed 80 \textsf{uniform} datasets using 8 experimental settings and 10 replications each, as
described in \cref{tab:subsampled-datasets}.
\begin{table}[tp]
	\centering
	\caption{Summary of uniform cluster data:
		A total number of $E=8$ experiment settings, each with $10$ replications, resulting in experiments with 8 different numbers of cell types: $T = 2\times[E]  =  \{2,4,…,16\}$, where each cell type has $ N_T = 500$ cell observations. Each experiment has total cell observations $t\times N_t=500\times t,\ \forall \, t\in T$.}{
		\begin{tabular}{cccc}
			\hline
			\textbf{Experiment Settings (E)} & \textbf{Cell Types (T)}   & \textbf{Cell Observations}           & \textbf{Replications} \\ \hline
			8                                & $\{2, 4, 6, \ldots, 16\}$ & $500 \times T$ (from 1,000 to 8,000) & 10 per setting        \\ \hline
		\end{tabular}}
	\label{tab:subsampled-datasets}
\end{table}

All datasets were processed according to the following procedure using the Seurat R package before being used for clustering.
Cells with low gene counts ($<200$) and genes expressed in very few cells ($<2$) are excluded.
The gene expression counts are normalized and log-transformed by each cell.
After log transformation, the counts are scaled so that each gene has a mean expression of 0 and a variance of 1 across all cells.
Finally, PCA is performed on a subset of highly variable genes that exhibit significant variation across cells, and the projected data dimension is determined by the Jackstraw method \citep{Chung_jackstraw}.

We evaluate cell clustering performance by examining both the accuracy of cluster assignments and the precision in estimating the number of cell types.
To test whether \methodname effectively prevents overestimation, we examine the deviation of the estimate from the true number of cell types $K_o$. The agreement between the ground truth labels and the estimated labels is quantified using the Adjusted Rand Index (ARI) and Adjusted Mutual Information (AMI).
For both ARI and AMI, a value of 1 indicates a perfect agreement between the compared clusters, and 0 indicates random clustering.
\begin{table}[tp]
	\centering
	\caption{Sinkhorn parameter settings}{
		\begin{tabular}{cccc}
			\hline
			\textbf{Data Dimension} & \textbf{$\gamma, c$} & \textbf{$\varepsilon$} & \textbf{$\rho_1, \rho_2$} \\ \hline
			$\leq 20$               & Uniform              & 1                   & 20                        \\
			21--30                  & Uniform              & 2                   & 10                        \\
			31--60                  & Uniform              & 2                   & 5                         \\ \hline
		\end{tabular}}
	\label{tab:parameter-settings}
\end{table}

\subsection{Additional Figures for Mixture Model Selection Criteria} \label{sec:sup-rnaseq-mmcrit}

\begin{figure}[tp]
	\centering
    \includegraphics[width=0.8\textwidth]{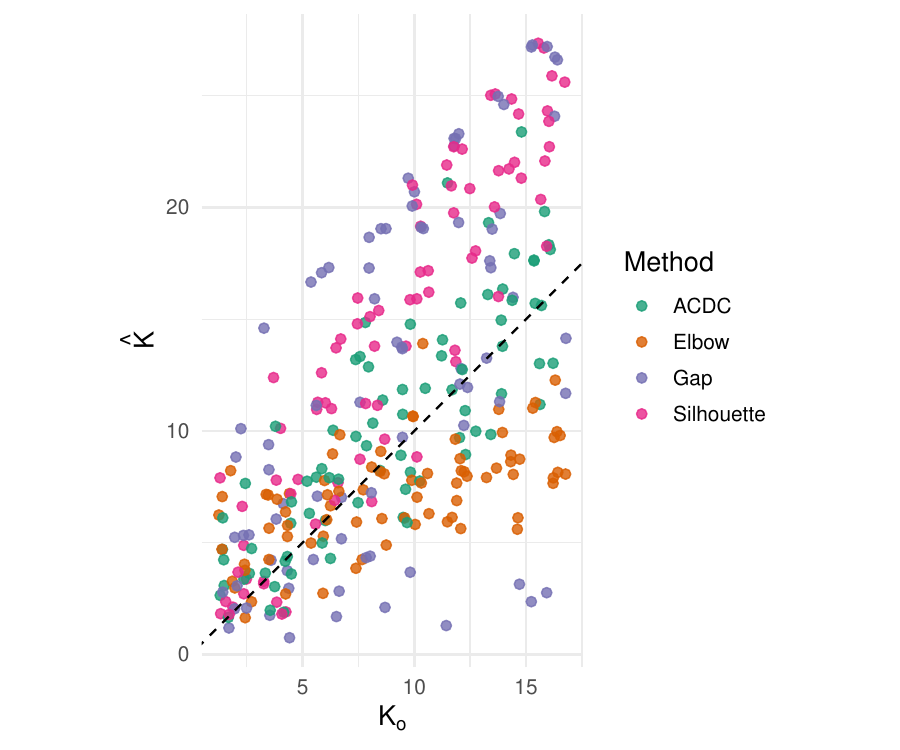}
	\caption{True number of cell types vs.\ estimated number of cell types for ACDC and other common model selection criteria (Elbow, Gap, and Silhouette).} 
	\label{fig:rnaesq_eq_suppK}
\end{figure}
In \cref{fig:rnaesq_eq_suppK}, \methodname produces estimates that lie tightly around the diagonal reference line, with relatively few extreme errors. For small $K_o$, Elbow and Silhouette behave similarly to \methodname.
However, as $K_o$ grows larger, Elbow has a tendency to underestimate, while Silhouette increasingly overestimate the number of cell types. The Gap statistic exhibits the highest variability, with substantial overestimation and underestimation.
Overall, \methodname provides the most stable and accurate recovery of the true number of cell types across the full range of $K_o$.

\begin{figure}[tp]
	\centering
    \includegraphics[width=0.9\textwidth]{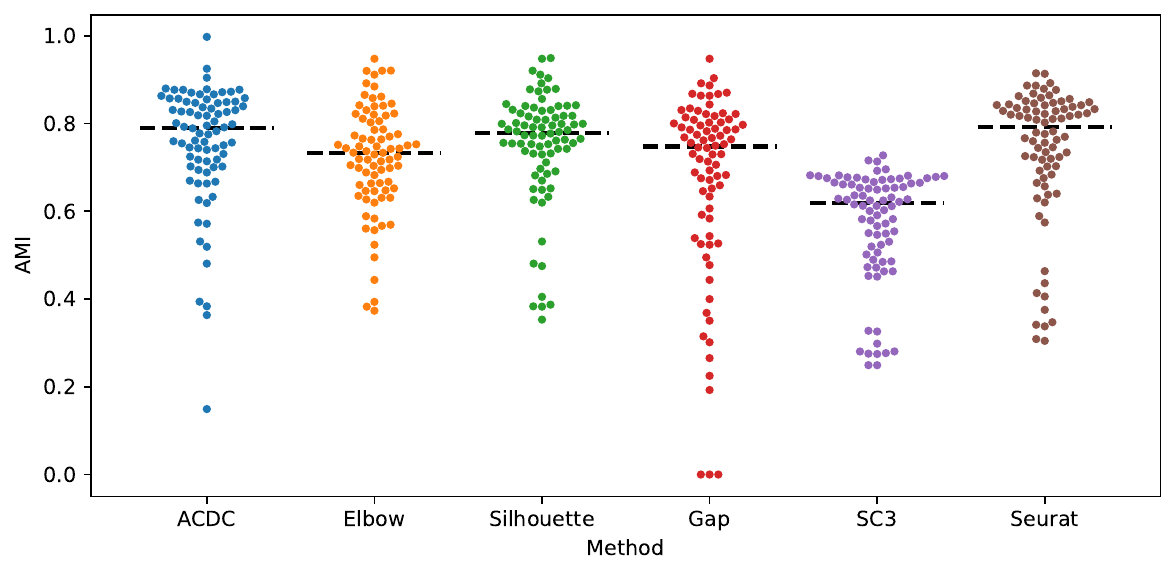}
    \includegraphics[width=0.9\textwidth]{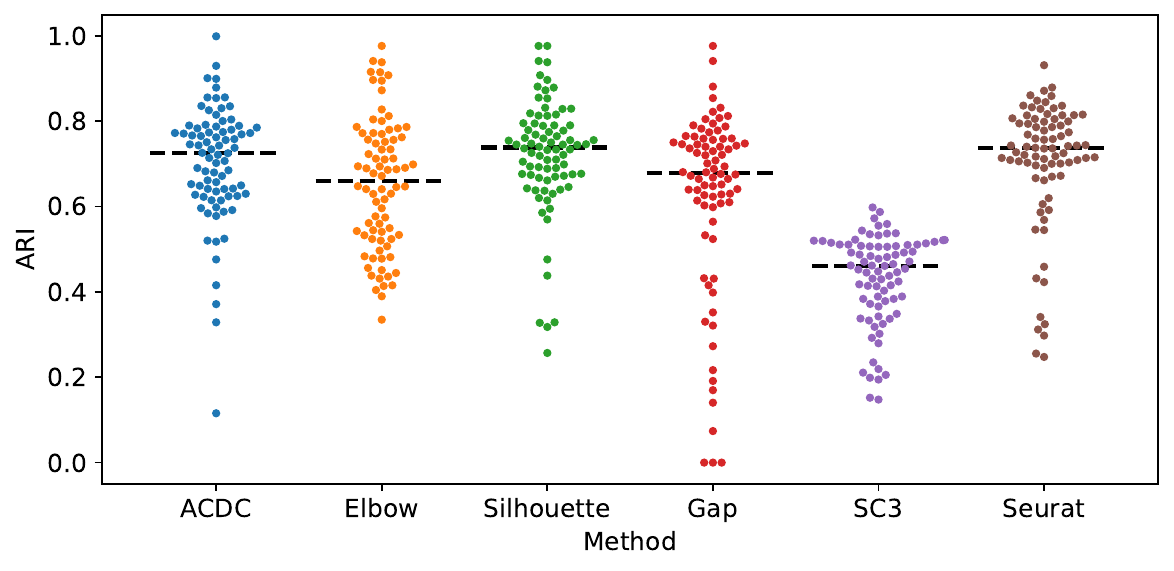}
	\caption{Comparison of ACDC with common model selection criteria (Elbow, Gap, and Silhouette) and specialized tools for scRNA-seq clustering (Seurat and SC3) on AMI and ARI. Dotted black lines indicate the medians.}
    \label{fig:rna-seq-mm-ami-ari}
\end{figure}

\subsection{Description of Existing Tools} \label{sec:existing-tools}

The Seurat R package \citep{seurat} uses a clustering algorithm based on shared nearest neighbor (SNN) modularity optimization.
Seurat first constructs a k-nearest neighbor (KNN) graph based on the Euclidean distance in the PCA space. A SNN graph is then constructed where edges are determined by the shared nearest neighbors among cells in the KNN graph.
The weights are assigned to the edges so that the edges connecting cells sharing close nearest neighbors are weighted higher compared to those joining cells sharing distant nearest neighbors.
Finally, the SNN graph is partitioned into clusters using the Louvain algorithm, which optimizes the modularity of the clustering solution.

The SC3 \citep[Single-Cell Consensus Clustering;][]{sc3} R package employs a robust consensus clustering approach that integrates PCA, K-means, and hierarchical clustering.
Distance matrices are computed using Euclidean, Pearson, and Spearman metrics and transformed using PCA. The transformed distance matrices are used for K-means clustering, and multiple clustering solutions are generated based on different numbers of eigenvectors of the matrices. A cell-to-cell binary similarity matrix is constructed for each clustering result with each entry indicating whether two cells belong to the same cluster. These similarity matrices are averaged to form a consensus matrix that is then clustered using agglomerative hierarchical clustering where the clusters are identified at a user-specified level of hierarchy. For our experiments, we use the cluster number estimation function provided by the package to determine $K$.

\section{Probabilistic Matrix Factorization Experimental Details and Further Results}

\subsection{Conditional Sampling for PMF Models}

\paragraph{Poisson NMF.}
Recall the standard Poisson NMF model:
\[
	{y}^{(K)}_{nk}& \distas \distPoiss\lrp{{\phi}^{(K)}_{k}z^{(K)}_{nk}} \quad&&\text{for}\ n=1,\dots,N,k=1,\dots,K\\
	{x}_{n}&=\sum^{K}_{k=1}{y}^{(K)}_{nk}\quad&&\text{for}\ n=1,\dots,N,
\]
Applying Bayes' rule, we can sample ${\veps}_{nk}\mid{x}_{n}$ for any given $n,k$ using the following procedure, with
each dimension $d$ sampled independently:
\[
	{y}^{(K)}_{n,1:K,[d]}\mid{x}_{n,[d]}&\sim \distMulti \lrp{{x}_{n,[d]};\widehat{p}_{n,1:K,d}}, \\
	{\veps}_{n,k,[d]}\mid{y}^{(K)}_{n,k,[d]}&\sim\Unif\left(\mathcal{F}_{\distPoiss}\lrp{y^{(K)}_{n,k,[d]}-1;p_{n,k,d}},\right.
	\left.\mathcal{F}_{\distPoiss}\lrp{{y}^{(K)}_{n,k,[d]};p_{n,k,d}}\right),
	\label{eq:Pois_eps_sampling}
\]
where
\[
	p_{n,k,d} &= {\phi}^{(K)}_{k,[d]}z^{(K)}_{nk}, &
	\widehat{p}_{n,k,d} &= \frac{p_{n,k,d}}{\sum^{K}_{k'=1}p_{n,k',d}},
\]
and $\mathcal{F}_{\distPoiss}(\blank;\lambda)$ is the cdf of $\distPoiss(\lambda)$.

\paragraph{Gaussian Factor Analysis.}
Recall the usual Gaussian factor analysis model: 
\[
	\begin{aligned}
		{\Sigma}^{(K)}_{k} & =\lrb{\sigma^{2}_{k, 1},\dots,\sigma^{2}_{k, D}}\transpose&& \text{for}\ k=1,\dots,K, \\
		{y}^{(K)}_{nk}          & \mop{\sim}^{\mathrm{e.w}} \distNorm\lrp{{\phi}^{(K)}_{k}z^{(K)}_{nk},{\Sigma}^{(K)}_{k}}
		\qquad\quad&& \text{for}\ n=1,\dots,N,k=1,\dots,K,                        \\
		{x}_{n}           &=\sum^{K}_{k=1}{y}_{nk}&& \text{for}\ n=1,\dots,N.
	\end{aligned}
\]
Again, applying Bayes' rule, we can sample ${\veps}_{nk}\mid{x}_{n}$ for any given $n,k$ using the following formulation, 
with each dimension $d$ sampled independently:
\[
	&{y}^{(K)}_{n,k,[d]}\mid{x}_{n,[d]},y^{(K)}_{n,1:k-1,[d]},\sigma_{1:K,d}\sim
	\left\{\begin{array}{lr}
		\distNorm(\widetilde{\mu}_{n,d,k},\widetilde{\sigma}^{2}_{k,d}) & \text{if}\ k\neq K, \\
		\delta\lrp{\overline{x}_{n,k,d}}                              & \text{if}\ k=K,
	\end{array}\right. \\
  &{\veps}_{n,k,[d]} = \mathcal{F}_{\distNorm}
	\lrp{{y}^{(K)}_{n,k,[d]};\mu_{n,k,d},\sigma^{2}_{k,d}},
\]
where
\[
	&\overline{x}_{n,k,d}={x}_{n,[d]} - \sum^{k-1}_{k'=1} {y}^{(K)}_{n,k',[d]},
	&\mu_{n,k,d}={\phi}^{(K)}_{k,[d]}z^{(K)}_{nk},\\
	&\overline{\mu}_{n,k,d}=\sum^{K}_{k'=k+1}\mu_{n,k',d}, %
	&\overline{\sigma}^{2}_{k,d}=\sum^{K}_{k'=k+1}\sigma^{2}_{k',d},\\
	&\widetilde{\mu}_{n,k,d}=\frac{\sigma^{-2}_{k,d}\mu_{n,k,d}-
		\overline{\sigma}_{k,d}^{-2}(\overline{\mu}_{n,k,d}-\overline{x}_{n,k,d})}
	{\sigma^{-2}_{k,d}+\overline{\sigma}^{-2}_{k,d}},
	&\widetilde{\sigma}^{2}_{k,d}=\frac{\sigma^{2}_{k,d}\overline{\sigma}^{2}_{k,d}}
	{\sigma^{2}_{k,d}+\overline{\sigma}^{2}_{k,d}},
\]
and $\mathcal{F}_{\distNorm}(\blank;\mu,\sigma^{2})$ is the cdf of $\distNorm(\mu,\sigma^{2})$. 

\subsection{Mutational Signature Discovery} \label{sec:mutsigs-details}

We use simulated breast cancer data
based on the COSMIC v2 catalog and the pan-cancer analysis of whole genomes (PCAWG),
following the procedure of \citet{Xue:2024}.
First, we applied nonnegative least squares regression to the count matrix and COSMIC signatures, resulting in best fit exposure vectors.
We then selected the signatures with significant loadings contributions and used them as the ground-truth signatures ${\phi}_{o}$, with
the inferred per-sample loadings serving as the ground truth exposures,
$z_{o 1}, \dots, z_{o N}$.
Finally, we generated four synthetic datasets: one well specified, and three others each with a different form of model misspecification.
The forms of misspecification we use are from \citet{Xue:2024}:
\begin{itemize}
	\item \textbf{Perturbation:} for each ${x}_{n}$, the signatures ${\phi}_{o}$ are stochastically perturbed
	      before being used to simulate the observed counts.
	\item \textbf{Contamination:} for each ${x}_{n}$, in addition to the ground truth signatures ${\phi}_{o}$,
	      a randomly generated signature with small exposure is included in the sampling process.
	\item \textbf{Overdispersion:} the data is sampled from a negative binomial distribution instead of a Poisson distribution.
\end{itemize}
For each value of $K$, we compute the MLE of the signature and loadings parameters using the multiplicative update algorithm %
\citep{Lee-Seung_multdiv_2000}.

\subsection{Additional Figures}
\cref{fig:mutsig_result_appendix_1} depicts the comparison of model selection quality of ACDC against BIC and PA. 
As mentioned in \cref{sec:mutsigs}, PA consistently underestimates, BIC overestimates, while ACDC is giving reasonable suggestions. \cref{fig:hyprunmix_gt} visualizes the urban dataset, along with its ground truth labeling, and inference labeling for $K=3,\dots,6$. The discussion of this figure can be found in \cref{sec:hyperspectral}.

\begin{figure}[]
  \centering
  \subfloat[Well specified]{\includegraphics[width=\textwidth]{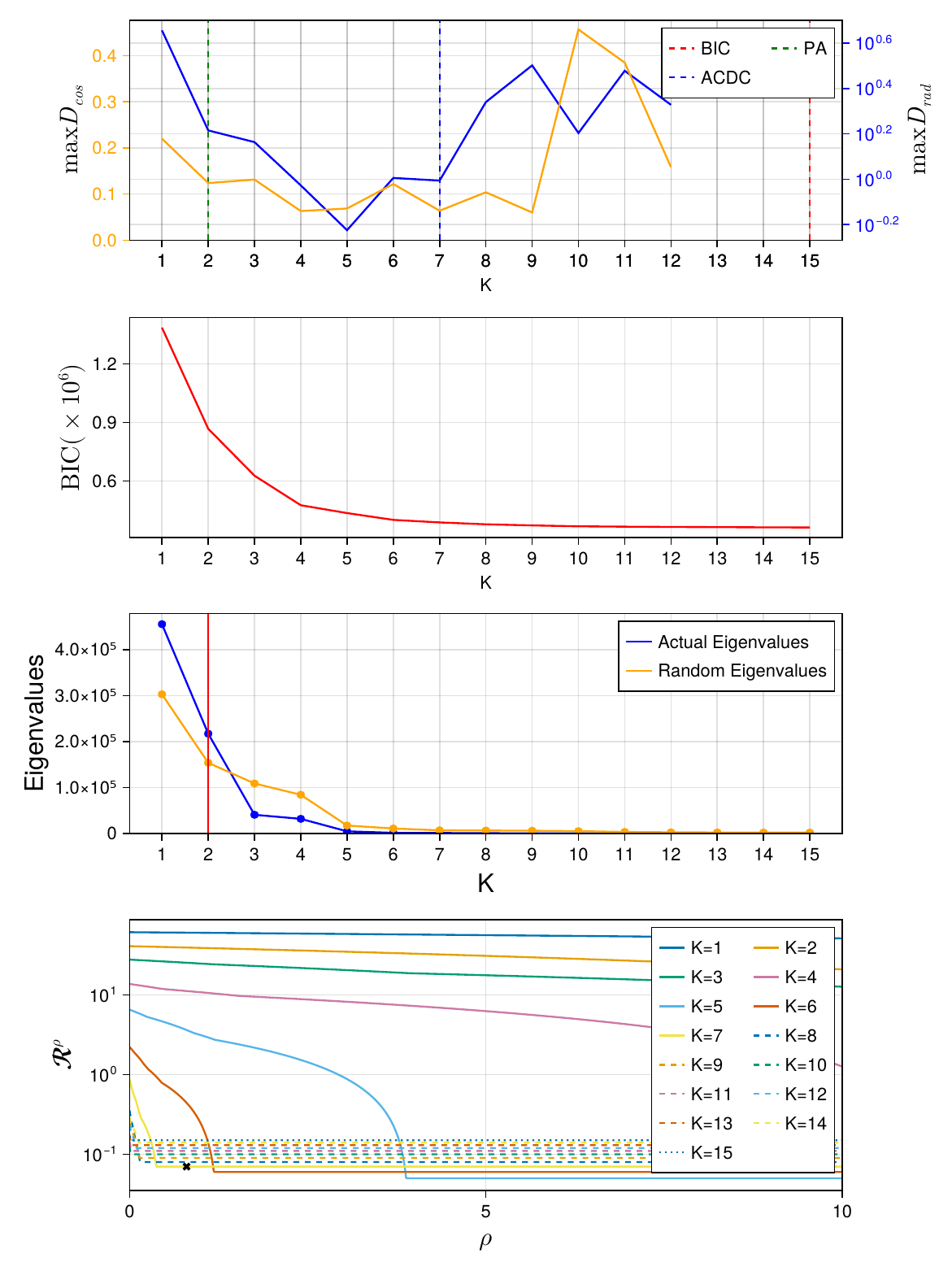}}\\
  \caption{Estimation quality for various scheme of data generation. See \cref{fig:mutsig_result} caption for explanation.}
  \label{fig:mutsig_result_appendix_1}
\end{figure}

\begin{figure}[]\ContinuedFloat
  \centering
  \subfloat[Contaminated]{\includegraphics[width=\textwidth]{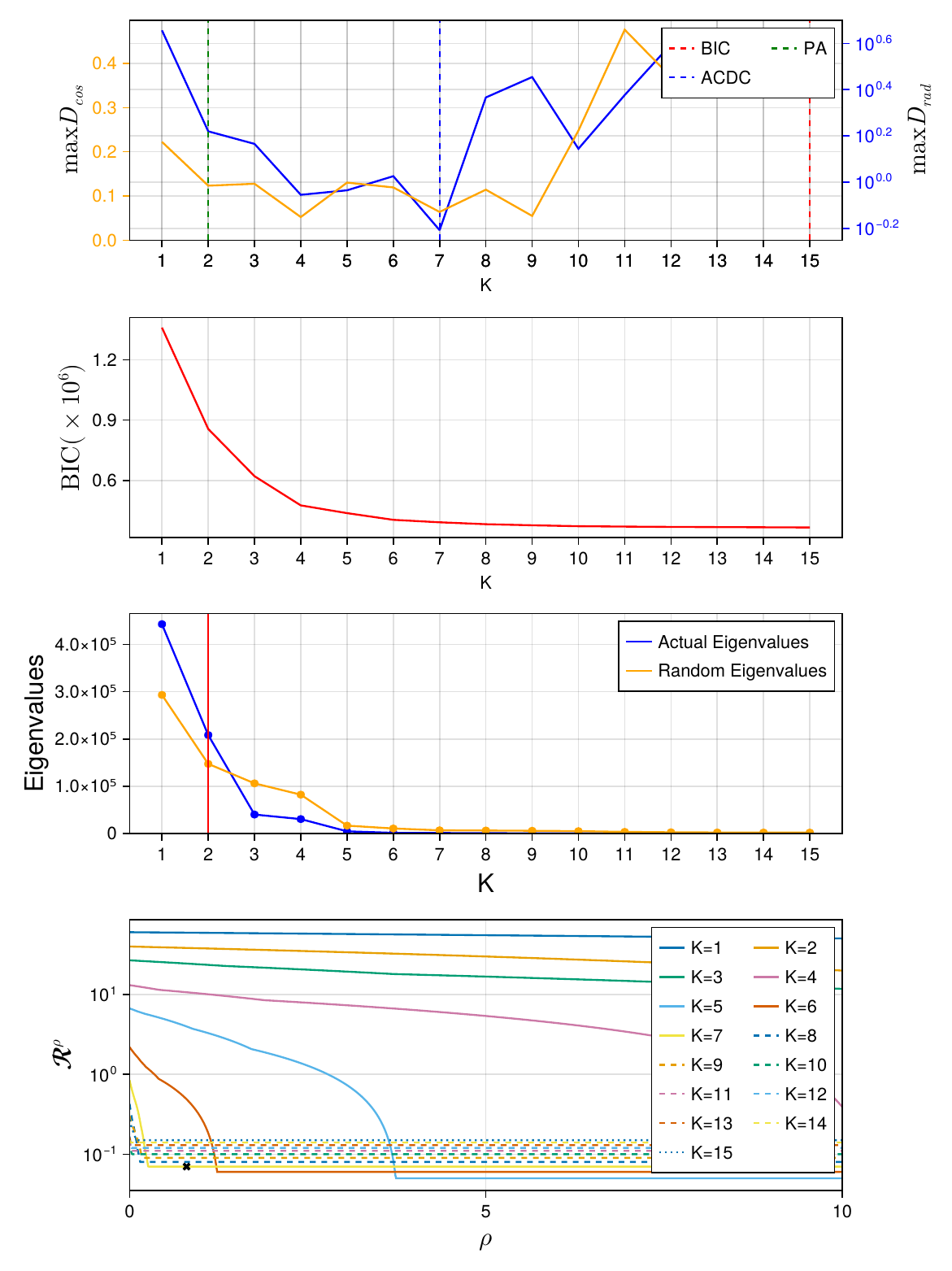}}\\
  \caption{Estimation quality for various scheme of data generation. See \cref{fig:mutsig_result} caption for explanation.}
  \label{fig:mutsig_result_appendix_2}
\end{figure}

\begin{figure}[]\ContinuedFloat
  \centering
  \subfloat[Overdispersed]{\includegraphics[width=\textwidth]{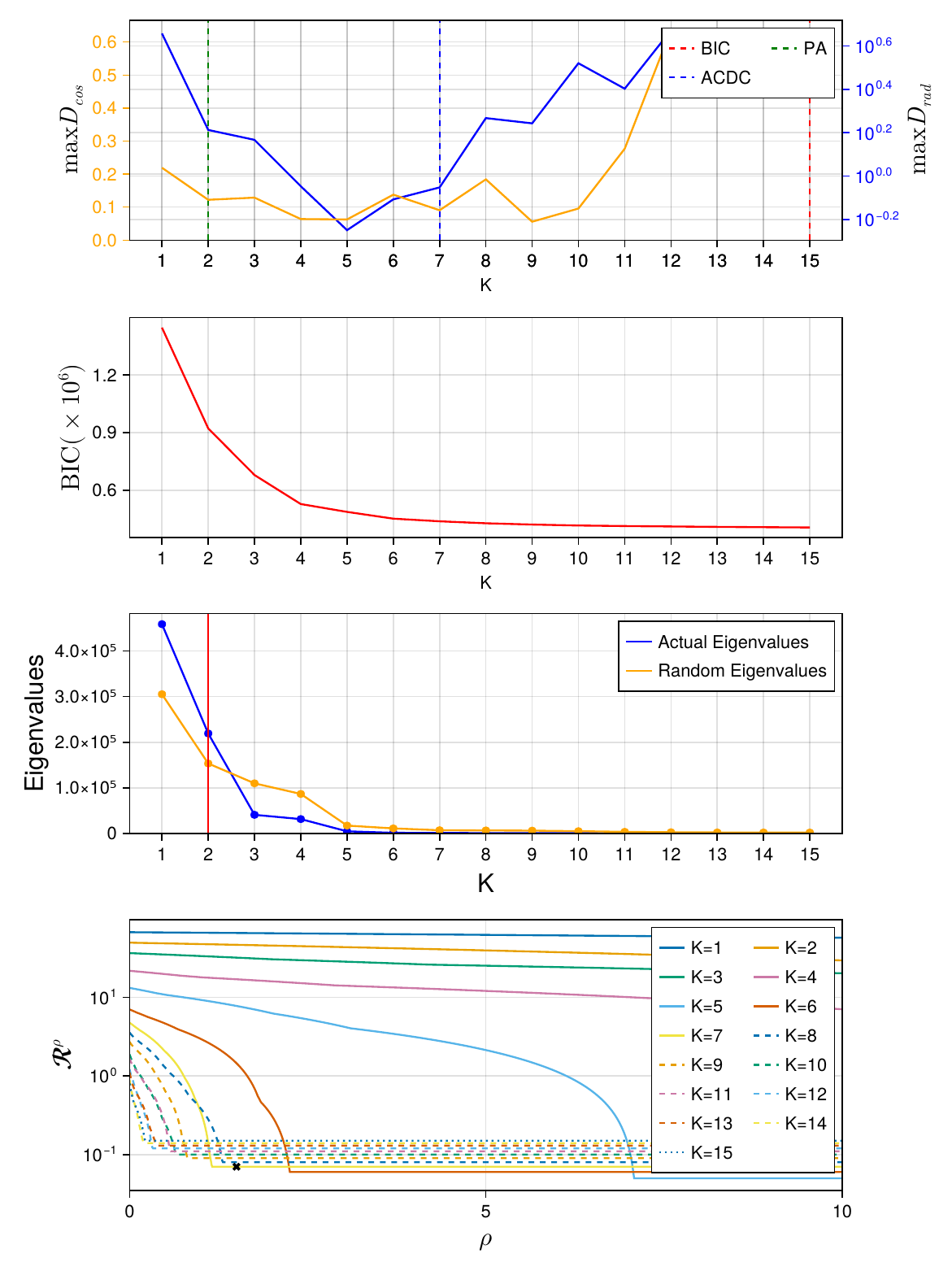}}
  \caption{Estimation quality for various scheme of data generation. See \cref{fig:mutsig_result} caption for explanation.}
  \label{fig:mutsig_result_appendix_3}
\end{figure}

\begin{figure}[]\ContinuedFloat
  \centering
  \subfloat[Perturbed]{\includegraphics[width=\textwidth]{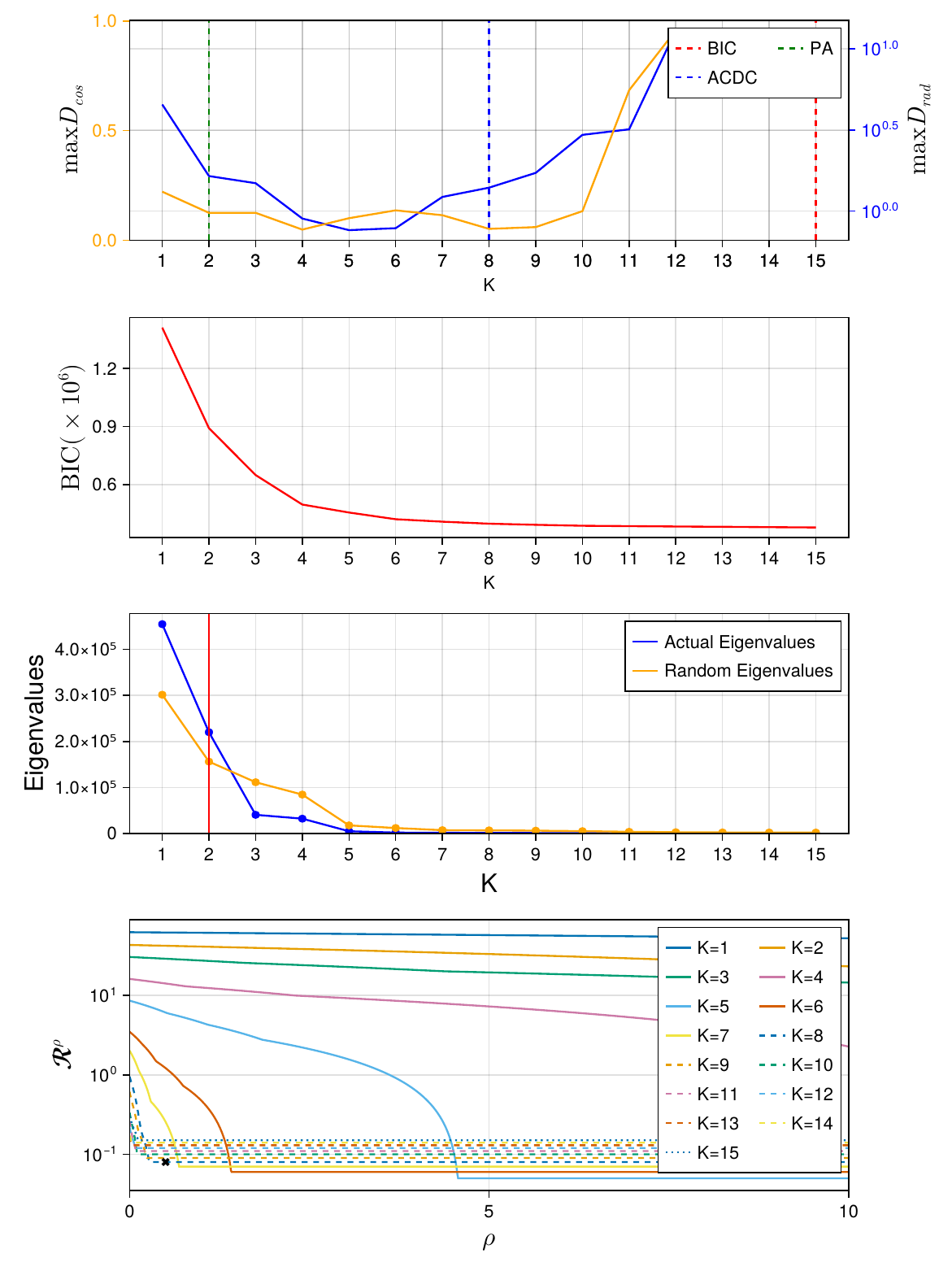}}
  \caption{Estimation quality for various scheme of data generation. See \cref{fig:mutsig_result} caption for explanation.}
  \label{fig:mutsig_result_appendix_4}
\end{figure}

\begin{figure}[h!]
	\centering
  \subfloat[Unlabeled (left) and labeled (right) versions of the urban dataset]{\includegraphics[width=0.9\textwidth]{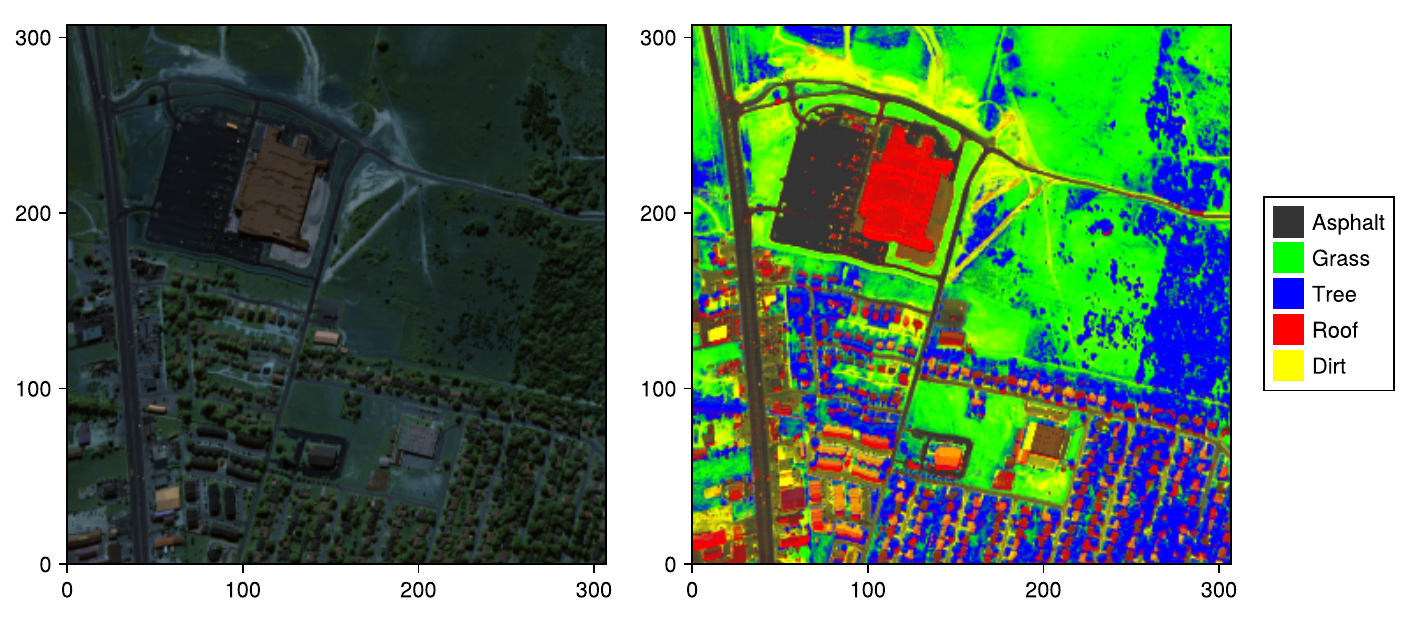}}\\
	\subfloat[$K=3$]{\includegraphics[width=0.45\textwidth]{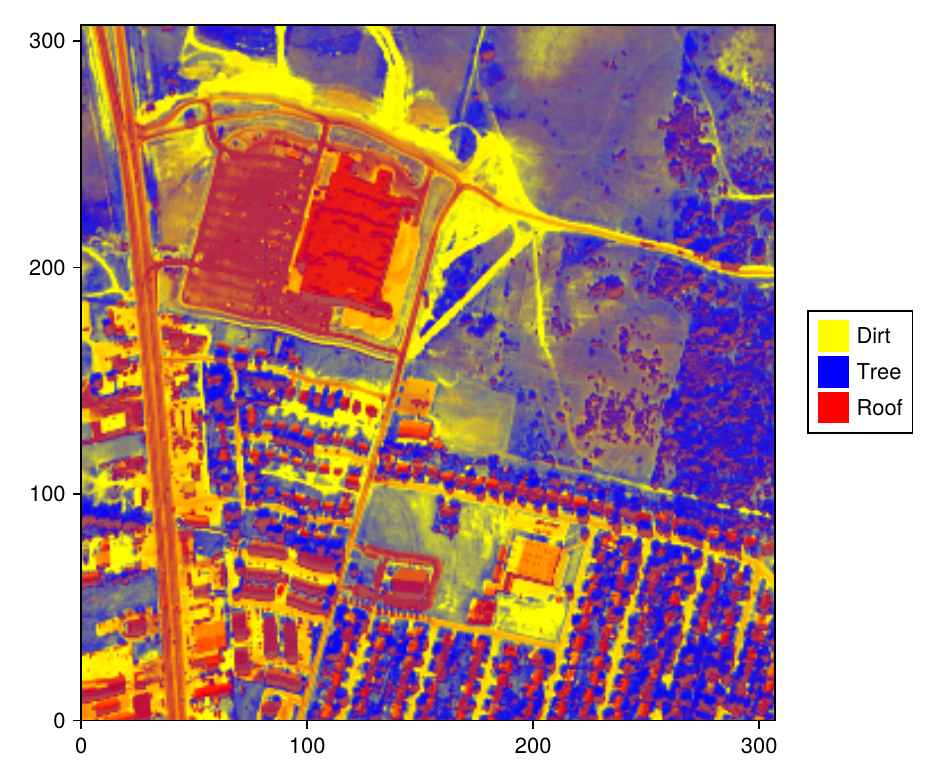}}
	\subfloat[$K=4$]{\includegraphics[width=0.45\textwidth]{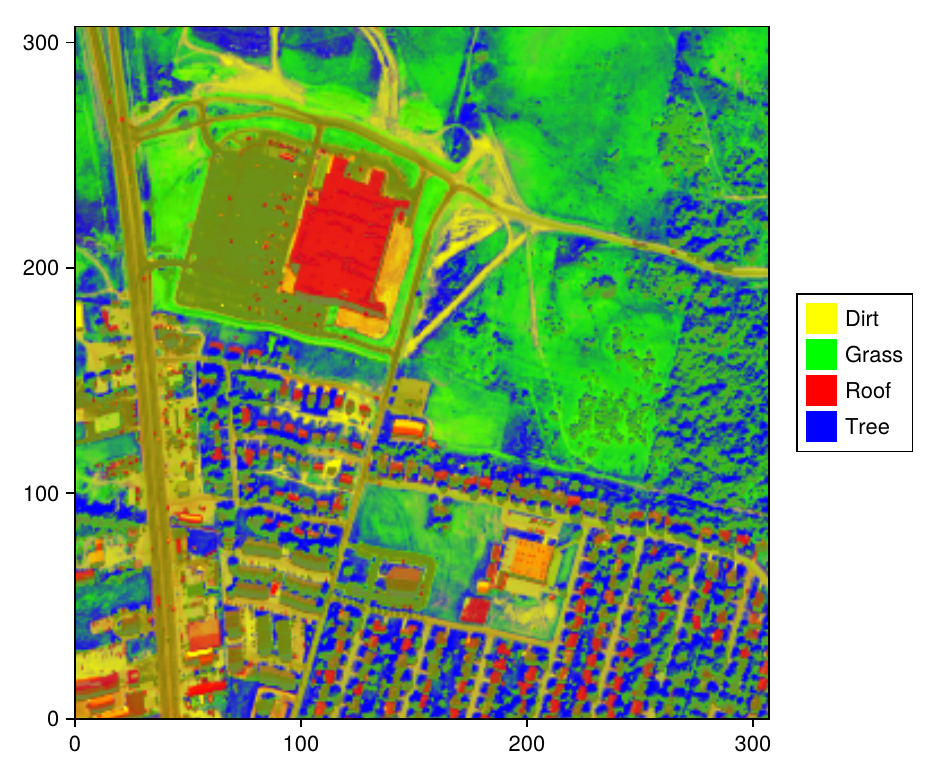}}\\
	\subfloat[$K=5$]{\includegraphics[width=0.45\textwidth]{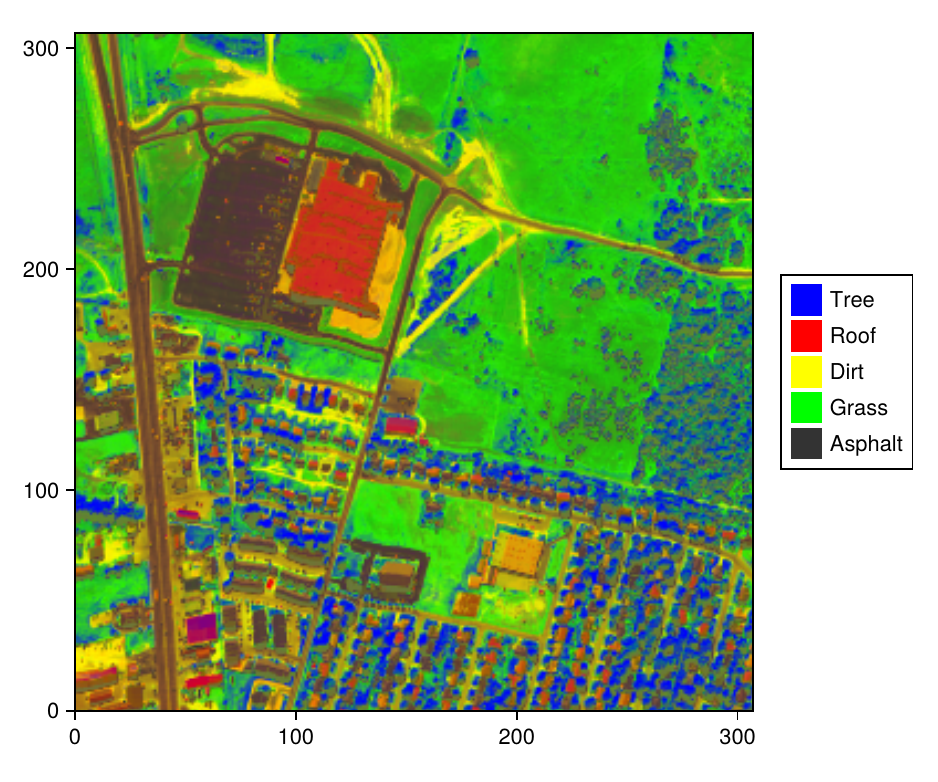}}
	\subfloat[$K=6$]{\includegraphics[width=0.45\textwidth]{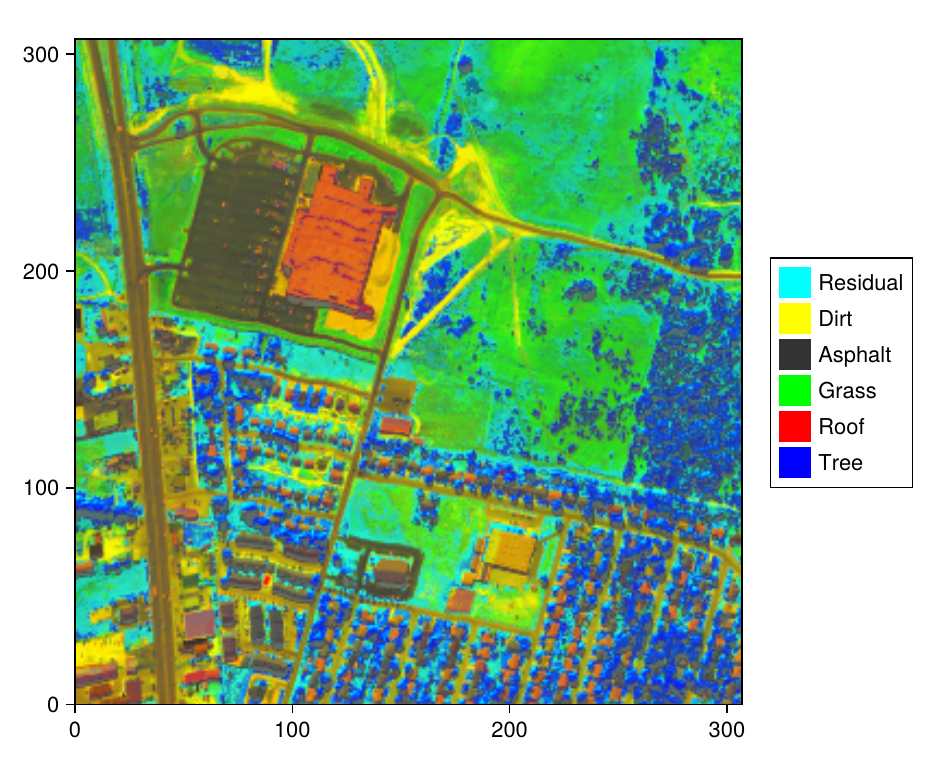}}
  \caption{Visualization of the hyperspectral urban datasets.
  \textbf{(a)} Data and ground truth labels.
  \textbf{(b, c, d, e)} Inferred end-member abundance maps for $K=3,4,5,6$.}
	\label{fig:hyprunmix_gt}
\end{figure}

%% file: _main.bbl
\begin{thebibliography}{93}

\bibitem[\protect\citeauthoryear{Akaike}{1974}]{Akaike:1974}
\begin{barticle}[author]
\bauthor{\bsnm{Akaike},~\bfnm{H.}\binits{H.}}
(\byear{1974}).
\btitle{A new look at the statistical model identification}.
\bjournal{IEEE Transactions on Automatic Control}
\bvolume{19}
\bpages{716-723}.
\bdoi{10.1109/TAC.1974.1100705}
\end{barticle}
\endbibitem

\bibitem[\protect\citeauthoryear{Alexandrov
  et~al.}{2013}]{Alexandrov_mut-sig-NMF_2013}
\begin{barticle}[author]
\bauthor{\bsnm{Alexandrov},~\bfnm{Ludmil~B.}\binits{L.~B.}},
  \bauthor{\bsnm{{Nik-Zainal}},~\bfnm{Serena}\binits{S.}},
  \bauthor{\bsnm{Wedge},~\bfnm{David~C.}\binits{D.~C.}},
  \bauthor{\bsnm{Campbell},~\bfnm{Peter~J.}\binits{P.~J.}} \AND
  \bauthor{\bsnm{Stratton},~\bfnm{Michael~R.}\binits{M.~R.}}
(\byear{2013}).
\btitle{Deciphering {{Signatures}} of {{Mutational Processes Operative}} in
  {{Human Cancer}}}.
\bjournal{Cell Reports}
\bvolume{3}
\bpages{246}.
\bdoi{10.1016/j.celrep.2012.12.008}
\end{barticle}
\endbibitem

\bibitem[\protect\citeauthoryear{Alexandrov et~al.}{2020}]{pcawg2020}
\begin{barticle}[author]
\bauthor{\bsnm{Alexandrov},~\bfnm{Ludmil~B.}\binits{L.~B.}},
  \bauthor{\bsnm{Kim},~\bfnm{Jaegil}\binits{J.}},
  \bauthor{\bsnm{Haradhvala},~\bfnm{Nicholas~J.}\binits{N.~J.}},
  \bauthor{\bsnm{Huang},~\bfnm{Mi~Ni}\binits{M.~N.}},
  \bauthor{\bsnm{Tian~Ng},~\bfnm{Alvin~Wei}\binits{A.~W.}},
  \bauthor{\bsnm{Wu},~\bfnm{Yang}\binits{Y.}},
  \bauthor{\bsnm{Boot},~\bfnm{Arnoud}\binits{A.}},
  \bauthor{\bsnm{Covington},~\bfnm{Kyle~R.}\binits{K.~R.}},
  \bauthor{\bsnm{Gordenin},~\bfnm{Dmitry~A.}\binits{D.~A.}},
  \bauthor{\bsnm{Bergstrom},~\bfnm{Erik~N.}\binits{E.~N.}},
  \bauthor{\bsnm{Islam},~\bfnm{S.~M.~Ashiqul}\binits{S.~M.~A.}},
  \bauthor{\bsnm{Lopez-Bigas},~\bfnm{Nuria}\binits{N.}},
  \bauthor{\bsnm{Klimczak},~\bfnm{Leszek~J.}\binits{L.~J.}},
  \bauthor{\bsnm{McPherson},~\bfnm{John~R.}\binits{J.~R.}},
  \bauthor{\bsnm{Morganella},~\bfnm{Sandro}\binits{S.}},
  \bauthor{\bsnm{Sabarinathan},~\bfnm{Radhakrishnan}\binits{R.}},
  \bauthor{\bsnm{Wheeler},~\bfnm{David~A.}\binits{D.~A.}},
  \bauthor{\bsnm{Mustonen},~\bfnm{Ville}\binits{V.}}, \bauthor{\bsnm{{PCAWG
  Mutational Signatures Working Group}}},
  \bauthor{\bsnm{Getz},~\bfnm{Gad}\binits{G.}},
  \bauthor{\bsnm{Rozen},~\bfnm{Steven~G.}\binits{S.~G.}},
  \bauthor{\bsnm{Stratton},~\bfnm{Michael~R.}\binits{M.~R.}} \AND
  \bauthor{\bsnm{{PCAWG Consortium}}}
(\byear{2020}).
\btitle{The repertoire of mutational signatures in human cancer.}
\bjournal{Nature}
\bvolume{578}
\bpages{94--101}.
\bdoi{10.1038/s41586-020-1943-3}
\end{barticle}
\endbibitem

\bibitem[\protect\citeauthoryear{Anandkumar et~al.}{2014}]{Anandkumar2014uc}
\begin{barticle}[author]
\bauthor{\bsnm{Anandkumar},~\bfnm{A.}\binits{A.}},
  \bauthor{\bsnm{Ge},~\bfnm{R}\binits{R.}},
  \bauthor{\bsnm{Hsu},~\bfnm{Daniel}\binits{D.}} \AND
  \bauthor{\bsnm{Kakade},~\bfnm{S.~M.}\binits{S.~M.}}
(\byear{2014}).
\btitle{Tensor Decompositions for Learning Latent Variable Models}.
\bvolume{15}
\bpages{2773 -- 2883}.
\end{barticle}
\endbibitem

\bibitem[\protect\citeauthoryear{Anderson}{1963}]{Anderson_AsymptoticPCA_1963}
\begin{barticle}[author]
\bauthor{\bsnm{Anderson},~\bfnm{T.~W.}\binits{T.~W.}}
(\byear{1963}).
\btitle{Asymptotic {{Theory}} for {{Principal Component Analysis}}}.
\bjournal{The Annals of Mathematical Statistics}
\bvolume{34}
\bpages{122--148}.
\end{barticle}
\endbibitem

\bibitem[\protect\citeauthoryear{Anderson and
  Amemiya}{1988}]{Anderson_AsymptoticFA_1988}
\begin{barticle}[author]
\bauthor{\bsnm{Anderson},~\bfnm{T.~W.}\binits{T.~W.}} \AND
  \bauthor{\bsnm{Amemiya},~\bfnm{Yasuo}\binits{Y.}}
(\byear{1988}).
\btitle{The {{Asymptotic Normal Distribution}} of {{Estimators}} in {{Factor
  Analysis}} under {{General Conditions}}}.
\bjournal{The Annals of Statistics}
\bvolume{16}
\bpages{759--771}.
\end{barticle}
\endbibitem

\bibitem[\protect\citeauthoryear{Bai and
  Ng}{2002}]{Bai_DeterminingNumberFactors_2002}
\begin{barticle}[author]
\bauthor{\bsnm{Bai},~\bfnm{Jushan}\binits{J.}} \AND
  \bauthor{\bsnm{Ng},~\bfnm{Serena}\binits{S.}}
(\byear{2002}).
\btitle{Determining the {{Number}} of {{Factors}} in {{Approximate Factor
  Models}}}.
\bjournal{Econometrica}
\bvolume{70}
\bpages{191--221}.
\bdoi{10.1111/1468-0262.00273}
\end{barticle}
\endbibitem

\bibitem[\protect\citeauthoryear{Balakrishnan, Wainwright and
  Yu}{2017}]{Balakrishnan:2017}
\begin{barticle}[author]
\bauthor{\bsnm{Balakrishnan},~\bfnm{Sivaraman}\binits{S.}},
  \bauthor{\bsnm{Wainwright},~\bfnm{Martin~J.}\binits{M.~J.}} \AND
  \bauthor{\bsnm{Yu},~\bfnm{Bin}\binits{B.}}
(\byear{2017}).
\btitle{{Statistical guarantees for the EM algorithm: From population to
  sample-based analysis}}.
\bjournal{The Annals of Statistics}
\bvolume{45}
\bpages{77 -- 120}.
\bdoi{10.1214/16-AOS1435}
\end{barticle}
\endbibitem

\bibitem[\protect\citeauthoryear{Bauer}{2007}]{Bauer:2007}
\begin{barticle}[author]
\bauthor{\bsnm{Bauer},~\bfnm{Daniel~J.}\binits{D.~J.}}
(\byear{2007}).
\btitle{Observations on the Use of Growth Mixture Models in Psychological
  Research}.
\bjournal{Multivariate Behavioral Research}
\bvolume{42}
\bpages{757-786}.
\bdoi{10.1080/00273170701710338}
\end{barticle}
\endbibitem

\bibitem[\protect\citeauthoryear{Bharti et~al.}{2023}]{Krause:2023}
\begin{binproceedings}[author]
\bauthor{\bsnm{Bharti},~\bfnm{Ayush}\binits{A.}},
  \bauthor{\bsnm{Naslidnyk},~\bfnm{Masha}\binits{M.}},
  \bauthor{\bsnm{Key},~\bfnm{Oscar}\binits{O.}},
  \bauthor{\bsnm{Kaski},~\bfnm{Samuel}\binits{S.}} \AND
  \bauthor{\bsnm{Briol},~\bfnm{Francois-Xavier}\binits{F.-X.}}
(\byear{2023}).
\btitle{Optimally-weighted Estimators of the Maximum Mean Discrepancy for
  Likelihood-Free Inference}.
In \bbooktitle{Proceedings of the 40th International Conference on Machine
  Learning}
(\beditor{\bfnm{Andreas}\binits{A.}~\bsnm{Krause}},
  \beditor{\bfnm{Emma}\binits{E.}~\bsnm{Brunskill}},
  \beditor{\bfnm{Kyunghyun}\binits{K.}~\bsnm{Cho}},
  \beditor{\bfnm{Barbara}\binits{B.}~\bsnm{Engelhardt}},
  \beditor{\bfnm{Sivan}\binits{S.}~\bsnm{Sabato}} \AND
  \beditor{\bfnm{Jonathan}\binits{J.}~\bsnm{Scarlett}}, eds.).
\bseries{Proceedings of Machine Learning Research}
\bvolume{202}
\bpages{2289--2312}.
\bpublisher{PMLR}.
\end{binproceedings}
\endbibitem

\bibitem[\protect\citeauthoryear{Blei and Lafferty}{2007}]{Blei:2007}
\begin{barticle}[author]
\bauthor{\bsnm{Blei},~\bfnm{D.~M.}\binits{D.~M.}} \AND
  \bauthor{\bsnm{Lafferty},~\bfnm{J~D}\binits{J.~D.}}
(\byear{2007}).
\btitle{{A correlated topic model of Science}}.
\bjournal{The Annals of Applied Statistics}
\bvolume{1}
\bpages{17 -- 35}.
\bdoi{10.1214/07-aoas114}
\end{barticle}
\endbibitem

\bibitem[\protect\citeauthoryear{Brinkman et~al.}{2007}]{Brinkman:2007}
\begin{barticle}[author]
\bauthor{\bsnm{Brinkman},~\bfnm{Ryan~Remy}\binits{R.~R.}},
  \bauthor{\bsnm{Gasparetto},~\bfnm{Maura}\binits{M.}},
  \bauthor{\bsnm{Lee},~\bfnm{Shang-Jung~Jessica}\binits{S.-J.~J.}},
  \bauthor{\bsnm{Ribickas},~\bfnm{Albert~J.}\binits{A.~J.}},
  \bauthor{\bsnm{Perkins},~\bfnm{Janelle}\binits{J.}},
  \bauthor{\bsnm{Janssen},~\bfnm{William}\binits{W.}},
  \bauthor{\bsnm{Smiley},~\bfnm{Renee}\binits{R.}} \AND
  \bauthor{\bsnm{Smith},~\bfnm{Clay}\binits{C.}}
(\byear{2007}).
\btitle{{High-Content Flow Cytometry and Temporal Data Analysis for Defining a
  Cellular Signature of Graft-Versus-Host Disease}}.
\bjournal{Biology of Blood and Marrow Transplantation}
\bvolume{13}
\bpages{691--700}.
\end{barticle}
\endbibitem

\bibitem[\protect\citeauthoryear{Brook and
  Dor}{2016}]{Brook_Dust_over_Green_Canopy_2016}
\begin{barticle}[author]
\bauthor{\bsnm{Brook},~\bfnm{Anna}\binits{A.}} \AND
  \bauthor{\bsnm{Dor},~\bfnm{Eyal~Ben}\binits{E.~B.}}
(\byear{2016}).
\btitle{Quantitative {{Detection}} of {{Settled Dust Over Green Canopy Using
  Sparse Unmixing}} of {{Airborne Hyperspectral Data}}}.
\bjournal{IEEE Journal of Selected Topics in Applied Earth Observations and
  Remote Sensing}
\bvolume{9}
\bpages{884--897}.
\bdoi{10.1109/JSTARS.2015.2489207}
\end{barticle}
\endbibitem

\bibitem[\protect\citeauthoryear{Brunet et~al.}{2004}]{brunet_CCC_2004a}
\begin{barticle}[author]
\bauthor{\bsnm{Brunet},~\bfnm{Jean-Philippe}\binits{J.-P.}},
  \bauthor{\bsnm{Tamayo},~\bfnm{Pablo}\binits{P.}},
  \bauthor{\bsnm{Golub},~\bfnm{Todd~R.}\binits{T.~R.}} \AND
  \bauthor{\bsnm{Mesirov},~\bfnm{Jill~P.}\binits{J.~P.}}
(\byear{2004}).
\btitle{Metagenes and Molecular Pattern Discovery Using Matrix Factorization}.
\bjournal{Proceedings of the National Academy of Sciences}
\bvolume{101}
\bpages{4164--4169}.
\bdoi{10.1073/pnas.0308531101}
\end{barticle}
\endbibitem

\bibitem[\protect\citeauthoryear{Buettner
  et~al.}{2017}]{Buettner_FscLVM_ScalableVersatile_FA_2017}
\begin{barticle}[author]
\bauthor{\bsnm{Buettner},~\bfnm{Florian}\binits{F.}},
  \bauthor{\bsnm{Pratanwanich},~\bfnm{Naruemon}\binits{N.}},
  \bauthor{\bsnm{McCarthy},~\bfnm{Davis~J.}\binits{D.~J.}},
  \bauthor{\bsnm{Marioni},~\bfnm{John~C.}\binits{J.~C.}} \AND
  \bauthor{\bsnm{Stegle},~\bfnm{Oliver}\binits{O.}}
(\byear{2017}).
\btitle{F-{{scLVM}}: Scalable and Versatile Factor Analysis for Single-Cell
  {{RNA-seq}}}.
\bjournal{Genome Biology}
\bvolume{18}
\bpages{212}.
\bdoi{10.1186/s13059-017-1334-8}
\end{barticle}
\endbibitem

\bibitem[\protect\citeauthoryear{Buja and
  Eyuboglu}{1992}]{Buja_RemarksParallelAnalysis_1992}
\begin{barticle}[author]
\bauthor{\bsnm{Buja},~\bfnm{A.}\binits{A.}} \AND
  \bauthor{\bsnm{Eyuboglu},~\bfnm{N.}\binits{N.}}
(\byear{1992}).
\btitle{Remarks on {{Parallel Analysis}}}.
\bjournal{Multivariate Behavioral Research}
\bvolume{27}
\bpages{509--540}.
\bdoi{10.1207/s15327906mbr2704_2}
\end{barticle}
\endbibitem

\bibitem[\protect\citeauthoryear{Butler et~al.}{2018}]{seurat}
\begin{barticle}[author]
\bauthor{\bsnm{Butler},~\bfnm{Andrew}\binits{A.}},
  \bauthor{\bsnm{Hoffman},~\bfnm{Paul}\binits{P.}},
  \bauthor{\bsnm{Smibert},~\bfnm{Peter}\binits{P.}},
  \bauthor{\bsnm{Papalexi},~\bfnm{Efthymia}\binits{E.}} \AND
  \bauthor{\bsnm{Satija},~\bfnm{Rahul}\binits{R.}}
(\byear{2018}).
\btitle{Integrating single-cell transcriptomic data across different
  conditions, technologies, and species}.
\bjournal{Nature Biotechnology}
\bvolume{36}
\bpages{411-420}.
\bdoi{10.1038/nbt.4096}
\end{barticle}
\endbibitem

\bibitem[\protect\citeauthoryear{Cai, Campbell and Broderick}{2021}]{Cai:2021}
\begin{binproceedings}[author]
\bauthor{\bsnm{Cai},~\bfnm{Diana}\binits{D.}},
  \bauthor{\bsnm{Campbell},~\bfnm{Trevor}\binits{T.}} \AND
  \bauthor{\bsnm{Broderick},~\bfnm{Tamara}\binits{T.}}
(\byear{2021}).
\btitle{Finite mixture models do not reliably learn the number of components}.
In \bbooktitle{Proceedings of the 38th International Conference on Machine
  Learning}
(\beditor{\bfnm{Marina}\binits{M.}~\bsnm{Meila}} \AND
  \beditor{\bfnm{Tong}\binits{T.}~\bsnm{Zhang}}, eds.).
\bseries{Proceedings of Machine Learning Research}
\bvolume{139}
\bpages{1158--1169}.
\bpublisher{PMLR}.
\end{binproceedings}
\endbibitem

\bibitem[\protect\citeauthoryear{Carvalho et~al.}{2008}]{Carvalho:2008}
\begin{barticle}[author]
\bauthor{\bsnm{Carvalho},~\bfnm{Carlos~M.}\binits{C.~M.}},
  \bauthor{\bsnm{Chang},~\bfnm{Jeffrey}\binits{J.}},
  \bauthor{\bsnm{Lucas},~\bfnm{Joseph~E.}\binits{J.~E.}},
  \bauthor{\bsnm{Nevins},~\bfnm{Joseph~R.}\binits{J.~R.}},
  \bauthor{\bsnm{Wang},~\bfnm{Quanli}\binits{Q.}} \AND
  \bauthor{\bsnm{and},~\bfnm{Mike~West}\binits{M.~W.}}
(\byear{2008}).
\btitle{High-Dimensional Sparse Factor Modeling: Applications in Gene
  Expression Genomics}.
\bjournal{Journal of the American Statistical Association}
\bvolume{103}
\bpages{1438--1456}.
\bdoi{10.1198/016214508000000869}
\end{barticle}
\endbibitem

\bibitem[\protect\citeauthoryear{Chakraborty, Bhattacharya and
  Pati}{2023}]{Chakraborty:2023}
\begin{barticle}[author]
\bauthor{\bsnm{Chakraborty},~\bfnm{Abhisek}\binits{A.}},
  \bauthor{\bsnm{Bhattacharya},~\bfnm{Anirban}\binits{A.}} \AND
  \bauthor{\bsnm{Pati},~\bfnm{Debdeep}\binits{D.}}
(\byear{2023}).
\btitle{{Robust probabilistic inference via a constrained transport metric}}.
\bjournal{arXiv}.
\end{barticle}
\endbibitem

\bibitem[\protect\citeauthoryear{Chakravarty and
  Solit}{2021}]{Chakravarty:2021}
\begin{barticle}[author]
\bauthor{\bsnm{Chakravarty},~\bfnm{Debyani}\binits{D.}} \AND
  \bauthor{\bsnm{Solit},~\bfnm{David~B.}\binits{D.~B.}}
(\byear{2021}).
\btitle{Clinical cancer genomic profiling.}
\bjournal{Nature Reviews. Genetics}
\bvolume{22}
\bpages{483--501}.
\bdoi{10.1038/s41576-021-00338-8}
\end{barticle}
\endbibitem

\bibitem[\protect\citeauthoryear{Chiou and Li}{2007}]{Chiou:2007}
\begin{barticle}[author]
\bauthor{\bsnm{Chiou},~\bfnm{Jeng‐Min}\binits{J.}} \AND
  \bauthor{\bsnm{Li},~\bfnm{Pai‐Ling}\binits{P.}}
(\byear{2007}).
\btitle{{Functional clustering and identifying substructures of longitudinal
  data}}.
\bjournal{Journal of the Royal Statistical Society: Series B (Statistical
  Methodology)}
\bvolume{69}
\bpages{679--699}.
\bdoi{10.1111/j.1467-9868.2007.00605.x}
\end{barticle}
\endbibitem

\bibitem[\protect\citeauthoryear{Chizat et~al.}{2020}]{fast_sh}
\begin{binproceedings}[author]
\bauthor{\bsnm{Chizat},~\bfnm{L\'{e}na\"{\i}c}\binits{L.}},
  \bauthor{\bsnm{Roussillon},~\bfnm{Pierre}\binits{P.}},
  \bauthor{\bsnm{L\'{e}ger},~\bfnm{Flavien}\binits{F.}},
  \bauthor{\bsnm{Vialard},~\bfnm{Fran\c{c}ois-Xavier}\binits{F.-X.}} \AND
  \bauthor{\bsnm{Peyr\'{e}},~\bfnm{Gabriel}\binits{G.}}
(\byear{2020}).
\btitle{Faster Wasserstein Distance Estimation with the Sinkhorn Divergence}.
In \bbooktitle{Advances in Neural Information Processing Systems}
(\beditor{\bfnm{H.}\binits{H.}~\bsnm{Larochelle}},
  \beditor{\bfnm{M.}\binits{M.}~\bsnm{Ranzato}},
  \beditor{\bfnm{R.}\binits{R.}~\bsnm{Hadsell}},
  \beditor{\bfnm{M.~F.}\binits{M.~F.}~\bsnm{Balcan}} \AND
  \beditor{\bfnm{H.}\binits{H.}~\bsnm{Lin}}, eds.)
\bvolume{33}
\bpages{2257--2269}.
\bpublisher{Curran Associates, Inc.}
\end{binproceedings}
\endbibitem

\bibitem[\protect\citeauthoryear{Chung and Storey}{2015}]{Chung_jackstraw}
\begin{barticle}[author]
\bauthor{\bsnm{Chung},~\bfnm{Neo~Christopher}\binits{N.~C.}} \AND
  \bauthor{\bsnm{Storey},~\bfnm{John~D.}\binits{J.~D.}}
(\byear{2015}).
\btitle{Statistical significance of variables driving systematic variation in
  high-dimensional data}.
\bjournal{Bioinformatics}
\bvolume{31}
\bpages{545-554}.
\bdoi{10.1093/bioinformatics/btu674}
\end{barticle}
\endbibitem

\bibitem[\protect\citeauthoryear{Cichocki and
  Phan}{2009}]{Cichocki-Phan_coorddesc_2009}
\begin{barticle}[author]
\bauthor{\bsnm{Cichocki},~\bfnm{Andrzej}\binits{A.}} \AND
  \bauthor{\bsnm{Phan},~\bfnm{Anh-Huy}\binits{A.-H.}}
(\byear{2009}).
\btitle{Fast {{Local Algorithms}} for {{Large Scale Nonnegative Matrix}} and
  {{Tensor Factorizations}}}.
\bjournal{IEICE Transactions}
\bvolume{92-A}
\bpages{708--721}.
\bdoi{10.1587/transfun.E92.A.708}
\end{barticle}
\endbibitem

\bibitem[\protect\citeauthoryear{Consortium}{2018}]{mice}
\begin{barticle}[author]
\bauthor{\bsnm{Consortium},~\bfnm{The Tabula~Muris}\binits{T.~T.~M.}}
(\byear{2018}).
\btitle{{Single-cell transcriptomics of 20 mouse organs creates a Tabula
  Muris}}.
\bjournal{Nature}
\bvolume{562}
\bpages{367--372}.
\bdoi{10.1038/s41586-018-0590-4}
\end{barticle}
\endbibitem

\bibitem[\protect\citeauthoryear{Cover and Thomas}{2006}]{cover2006elements}
\begin{bbook}[author]
\bauthor{\bsnm{Cover},~\bfnm{Thomas~M.}\binits{T.~M.}} \AND
  \bauthor{\bsnm{Thomas},~\bfnm{Joy~A.}\binits{J.~A.}}
(\byear{2006}).
\btitle{Elements of Information Theory},
\bedition{2nd} ed.
\bpublisher{Wiley-Interscience}, \baddress{Hoboken, NJ}.
\end{bbook}
\endbibitem

\bibitem[\protect\citeauthoryear{Cunningham and Yu}{2014}]{Cunningham:2014}
\begin{barticle}[author]
\bauthor{\bsnm{Cunningham},~\bfnm{John~P}\binits{J.~P.}} \AND
  \bauthor{\bsnm{Yu},~\bfnm{Byron~M}\binits{B.~M.}}
(\byear{2014}).
\btitle{{Dimensionality reduction for large-scale neural recordings}}.
\bjournal{Nature Neuroscience}
\bvolume{17}
\bpages{1500--1509}.
\bdoi{10.1038/nn.3776}
\end{barticle}
\endbibitem

\bibitem[\protect\citeauthoryear{de~Souto et~al.}{2008}]{Souto:2008}
\begin{barticle}[author]
\bauthor{\bparticle{de} \bsnm{Souto},~\bfnm{Marc{\'i}lio
  Carlos~Pereira}\binits{M.~C.~P.}},
  \bauthor{\bsnm{Costa},~\bfnm{Ivan~G.}\binits{I.~G.}}, \bauthor{\bparticle{de}
  \bsnm{Araujo},~\bfnm{Daniel S.~A.}\binits{D.~S.~A.}},
  \bauthor{\bsnm{Ludermir},~\bfnm{Teresa~Bernarda}\binits{T.~B.}} \AND
  \bauthor{\bsnm{Schliep},~\bfnm{Alexander}\binits{A.}}
(\byear{2008}).
\btitle{Clustering cancer gene expression data: a comparative study}.
\bjournal{BMC Bioinformatics}
\bvolume{9}
\bpages{497 - 497}.
\end{barticle}
\endbibitem

\bibitem[\protect\citeauthoryear{Devarajan}{2019}]{Devarajan_NMFdualdivergence_2019}
\begin{bmisc}[author]
\bauthor{\bsnm{Devarajan},~\bfnm{Karthik}\binits{K.}}
(\byear{2019}).
\btitle{Non-Negative Matrix Factorization Based on Generalized Dual
  Divergence}.
\bdoi{10.48550/arXiv.1905.07034}
\end{bmisc}
\endbibitem

\bibitem[\protect\citeauthoryear{Devroye and Wagner}{1977}]{Devroye:1977}
\begin{barticle}[author]
\bauthor{\bsnm{Devroye},~\bfnm{Luc~P.}\binits{L.~P.}} \AND
  \bauthor{\bsnm{Wagner},~\bfnm{T.~J.}\binits{T.~J.}}
(\byear{1977}).
\btitle{{The Strong Uniform Consistency of Nearest Neighbor Density
  Estimates}}.
\bjournal{The Annals of Statistics}
\bvolume{5}
\bpages{536 -- 540}.
\bdoi{10.1214/aos/1176343851}
\end{barticle}
\endbibitem

\bibitem[\protect\citeauthoryear{Dewaskar et~al.}{2023}]{Dewaskar:2023}
\begin{barticle}[author]
\bauthor{\bsnm{Dewaskar},~\bfnm{Miheer}\binits{M.}},
  \bauthor{\bsnm{Tosh},~\bfnm{Christopher}\binits{C.}},
  \bauthor{\bsnm{Knoblauch},~\bfnm{Jeremias}\binits{J.}} \AND
  \bauthor{\bsnm{Dunson},~\bfnm{David~B}\binits{D.~B.}}
(\byear{2023}).
\btitle{{Robustifying likelihoods by optimistically re-weighting data}}.
\bjournal{arXiv}.
\end{barticle}
\endbibitem

\bibitem[\protect\citeauthoryear{Dobriban}{2020}]{Dobriban_PA-for-FA_2020}
\begin{barticle}[author]
\bauthor{\bsnm{Dobriban},~\bfnm{Edgar}\binits{E.}}
(\byear{2020}).
\btitle{Permutation Methods for Factor Analysis and {{PCA}}}.
\bjournal{The Annals of Statistics}
\bvolume{48}.
\bdoi{10.1214/19-AOS1907}
\end{barticle}
\endbibitem

\bibitem[\protect\citeauthoryear{Donoho}{2000}]{Donoho:2000}
\begin{barticle}[author]
\bauthor{\bsnm{Donoho},~\bfnm{David}\binits{D.}}
(\byear{2000}).
\btitle{High-Dimensional Data Analysis: The Curses and Blessings of
  Dimensionality}.
\bjournal{AMS Math Challenges Lecture}
\bpages{1-32}.
\end{barticle}
\endbibitem

\bibitem[\protect\citeauthoryear{Dunson}{2000}]{Dunson:2000}
\begin{barticle}[author]
\bauthor{\bsnm{Dunson},~\bfnm{D.~B.}\binits{D.~B.}}
(\byear{2000}).
\btitle{{Bayesian latent variable models for clustered mixed outcomes}}.
\bjournal{Journal of the Royal Statistical Society: Series B (Statistical
  Methodology)}
\bvolume{62}
\bpages{355--366}.
\bdoi{10.1111/1467-9868.00236}
\end{barticle}
\endbibitem

\bibitem[\protect\citeauthoryear{F{\'e}votte and
  Dobigeon}{2015}]{Fevotte_NonlinearHyperspectralUnmixing_2015}
\begin{barticle}[author]
\bauthor{\bsnm{F{\'e}votte},~\bfnm{C{\'e}dric}\binits{C.}} \AND
  \bauthor{\bsnm{Dobigeon},~\bfnm{Nicolas}\binits{N.}}
(\byear{2015}).
\btitle{Nonlinear Hyperspectral Unmixing with Robust Nonnegative Matrix
  Factorization}.
\bjournal{IEEE Transactions on Image Processing}
\bvolume{24}
\bpages{4810--4819}.
\bdoi{10.1109/TIP.2015.2468177}
\end{barticle}
\endbibitem

\bibitem[\protect\citeauthoryear{Flynt, Dean and
  Nugent}{2019}]{Flynt_sARI_2019}
\begin{barticle}[author]
\bauthor{\bsnm{Flynt},~\bfnm{Abby}\binits{A.}},
  \bauthor{\bsnm{Dean},~\bfnm{Nema}\binits{N.}} \AND
  \bauthor{\bsnm{Nugent},~\bfnm{Rebecca}\binits{R.}}
(\byear{2019}).
\btitle{{{sARI}}: A Soft Agreement Measure for Class Partitions Incorporating
  Assignment Probabilities}.
\bjournal{Advances in Data Analysis and Classification}
\bvolume{13}
\bpages{303--323}.
\bdoi{10.1007/s11634-018-0346-x}
\end{barticle}
\endbibitem

\bibitem[\protect\citeauthoryear{Fr{\"u}hwirth-Schnatter}{2006}]{Fruhwurth:2006}
\begin{bbook}[author]
\bauthor{\bsnm{Fr{\"u}hwirth-Schnatter},~\bfnm{S.}\binits{S.}}
(\byear{2006}).
\btitle{Finite Mixture and Markov Switching Models}.
\bseries{Springer Series in Statistics}.
\bpublisher{Springer New York}.
\end{bbook}
\endbibitem

\bibitem[\protect\citeauthoryear{Fu, Huang and
  Sidiropoulos}{2018}]{Fu_IdentifiabilityNMF_2018}
\begin{barticle}[author]
\bauthor{\bsnm{Fu},~\bfnm{Xiao}\binits{X.}},
  \bauthor{\bsnm{Huang},~\bfnm{Kejun}\binits{K.}} \AND
  \bauthor{\bsnm{Sidiropoulos},~\bfnm{Nicholas~D.}\binits{N.~D.}}
(\byear{2018}).
\btitle{On {{Identifiability}} of {{Nonnegative Matrix Factorization}}}.
\bjournal{IEEE Signal Processing Letters}
\bvolume{25}
\bpages{328--332}.
\bdoi{10.1109/LSP.2018.2789405}
\end{barticle}
\endbibitem

\bibitem[\protect\citeauthoryear{Gibbs and Su}{2002}]{Gibbs:2002}
\begin{barticle}[author]
\bauthor{\bsnm{Gibbs},~\bfnm{Alison~L.}\binits{A.~L.}} \AND
  \bauthor{\bsnm{Su},~\bfnm{Francis~Edward}\binits{F.~E.}}
(\byear{2002}).
\btitle{On Choosing and Bounding Probability Metrics}.
\bjournal{International Statistical Review}
\bvolume{70}
\bpages{419--435}.
\end{barticle}
\endbibitem

\bibitem[\protect\citeauthoryear{Gorsky, Chan and Ma}{2020}]{Gorsky:2020}
\begin{barticle}[author]
\bauthor{\bsnm{Gorsky},~\bfnm{Shai}\binits{S.}},
  \bauthor{\bsnm{Chan},~\bfnm{Cliburn}\binits{C.}} \AND
  \bauthor{\bsnm{Ma},~\bfnm{Li}\binits{L.}}
(\byear{2020}).
\btitle{Coarsened mixtures of hierarchical skew normal kernels for flow
  cytometry analyses}.
\bjournal{arXiv}.
\end{barticle}
\endbibitem

\bibitem[\protect\citeauthoryear{Grabski, Street and Irizarry}{2023}]{grabski}
\begin{barticle}[author]
\bauthor{\bsnm{Grabski},~\bfnm{Isabella~N.}\binits{I.~N.}},
  \bauthor{\bsnm{Street},~\bfnm{Kelly}\binits{K.}} \AND
  \bauthor{\bsnm{Irizarry},~\bfnm{Rafael~A.}\binits{R.~A.}}
(\byear{2023}).
\btitle{Significance analysis for clustering with single-cell {RNA}-sequencing
  data}.
\bjournal{Nature Methods}
\bvolume{20}
\bpages{1196--1202}.
\bdoi{10.1038/s41592-023-01933-9}
\end{barticle}
\endbibitem

\bibitem[\protect\citeauthoryear{Gretton et~al.}{2012}]{Gretton:2012}
\begin{barticle}[author]
\bauthor{\bsnm{Gretton},~\bfnm{Arthur}\binits{A.}},
  \bauthor{\bsnm{Borgwardt},~\bfnm{Karsten~M.}\binits{K.~M.}},
  \bauthor{\bsnm{Rasch},~\bfnm{Malte~J.}\binits{M.~J.}},
  \bauthor{\bsnm{Sch{{\"o}}lkopf},~\bfnm{Bernhard}\binits{B.}} \AND
  \bauthor{\bsnm{Smola},~\bfnm{Alexander}\binits{A.}}
(\byear{2012}).
\btitle{A Kernel Two-Sample Test}.
\bjournal{Journal of Machine Learning Research}
\bvolume{13}
\bpages{723-773}.
\end{barticle}
\endbibitem

\bibitem[\protect\citeauthoryear{Guha, Ho and Nguyen}{2021}]{Guha:2021}
\begin{barticle}[author]
\bauthor{\bsnm{Guha},~\bfnm{Aritra}\binits{A.}},
  \bauthor{\bsnm{Ho},~\bfnm{Nhat}\binits{N.}} \AND
  \bauthor{\bsnm{Nguyen},~\bfnm{XuanLong}\binits{X.}}
(\byear{2021}).
\btitle{{On posterior contraction of parameters and interpretability in
  Bayesian mixture modeling}}.
\bjournal{Bernoulli}
\bvolume{27}
\bpages{2159 -- 2188}.
\bdoi{10.3150/20-BEJ1275}
\end{barticle}
\endbibitem

\bibitem[\protect\citeauthoryear{Horn}{1965}]{Horn_ParallelAnalysis_1965}
\begin{barticle}[author]
\bauthor{\bsnm{Horn},~\bfnm{John~L.}\binits{J.~L.}}
(\byear{1965}).
\btitle{A Rationale and Test for the Number of Factors in Factor Analysis}.
\bjournal{Psychometrika}
\bvolume{30}
\bpages{179--185}.
\bdoi{10.1007/BF02289447}
\end{barticle}
\endbibitem

\bibitem[\protect\citeauthoryear{Huang and
  Yao}{2012}]{Huang_MixtureRegressionModels_2012}
\begin{barticle}[author]
\bauthor{\bsnm{Huang},~\bfnm{Mian}\binits{M.}} \AND
  \bauthor{\bsnm{Yao},~\bfnm{Weixin}\binits{W.}}
(\byear{2012}).
\btitle{Mixture of {{Regression Models With Varying Mixing Proportions}}: {{A
  Semiparametric Approach}}}.
\bjournal{Journal of the American Statistical Association}
\bvolume{107}
\bpages{711--724}.
\bdoi{10.1080/01621459.2012.682541}
\end{barticle}
\endbibitem

\bibitem[\protect\citeauthoryear{Hubert and Arabie}{1985}]{ari}
\begin{barticle}[author]
\bauthor{\bsnm{Hubert},~\bfnm{L.}\binits{L.}} \AND
  \bauthor{\bsnm{Arabie},~\bfnm{P.}\binits{P.}}
(\byear{1985}).
\btitle{{Comparing partitions}}.
\bjournal{Journal of classification}
\bvolume{2}
\bpages{193--218}.
\end{barticle}
\endbibitem

\bibitem[\protect\citeauthoryear{Islam
  et~al.}{2022}]{islam_sigprofiler_extractor_2022}
\begin{barticle}[author]
\bauthor{\bsnm{Islam},~\bfnm{S.~M.~Ashiqul}\binits{S.~M.~A.}},
  \bauthor{\bsnm{{D{\'i}az-Gay}},~\bfnm{Marcos}\binits{M.}},
  \bauthor{\bsnm{Wu},~\bfnm{Yang}\binits{Y.}},
  \bauthor{\bsnm{Barnes},~\bfnm{Mark}\binits{M.}},
  \bauthor{\bsnm{Vangara},~\bfnm{Raviteja}\binits{R.}},
  \bauthor{\bsnm{Bergstrom},~\bfnm{Erik~N.}\binits{E.~N.}},
  \bauthor{\bsnm{He},~\bfnm{Yudou}\binits{Y.}},
  \bauthor{\bsnm{Vella},~\bfnm{Mike}\binits{M.}},
  \bauthor{\bsnm{Wang},~\bfnm{Jingwei}\binits{J.}},
  \bauthor{\bsnm{Teague},~\bfnm{Jon~W.}\binits{J.~W.}},
  \bauthor{\bsnm{Clapham},~\bfnm{Peter}\binits{P.}},
  \bauthor{\bsnm{Moody},~\bfnm{Sarah}\binits{S.}},
  \bauthor{\bsnm{Senkin},~\bfnm{Sergey}\binits{S.}},
  \bauthor{\bsnm{Li},~\bfnm{Yun~Rose}\binits{Y.~R.}},
  \bauthor{\bsnm{Riva},~\bfnm{Laura}\binits{L.}},
  \bauthor{\bsnm{Zhang},~\bfnm{Tongwu}\binits{T.}},
  \bauthor{\bsnm{Gruber},~\bfnm{Andreas~J.}\binits{A.~J.}},
  \bauthor{\bsnm{Steele},~\bfnm{Christopher~D.}\binits{C.~D.}},
  \bauthor{\bsnm{Otlu},~\bfnm{Bur{\c c}ak}\binits{B.}},
  \bauthor{\bsnm{Khandekar},~\bfnm{Azhar}\binits{A.}},
  \bauthor{\bsnm{Abbasi},~\bfnm{Ammal}\binits{A.}},
  \bauthor{\bsnm{Humphreys},~\bfnm{Laura}\binits{L.}},
  \bauthor{\bsnm{Syulyukina},~\bfnm{Natalia}\binits{N.}},
  \bauthor{\bsnm{Brady},~\bfnm{Samuel~W.}\binits{S.~W.}},
  \bauthor{\bsnm{Alexandrov},~\bfnm{Boian~S.}\binits{B.~S.}},
  \bauthor{\bsnm{Pillay},~\bfnm{Nischalan}\binits{N.}},
  \bauthor{\bsnm{Zhang},~\bfnm{Jinghui}\binits{J.}},
  \bauthor{\bsnm{Adams},~\bfnm{David~J.}\binits{D.~J.}},
  \bauthor{\bsnm{Martincorena},~\bfnm{I{\~n}igo}\binits{I.}},
  \bauthor{\bsnm{Wedge},~\bfnm{David~C.}\binits{D.~C.}},
  \bauthor{\bsnm{Landi},~\bfnm{Maria~Teresa}\binits{M.~T.}},
  \bauthor{\bsnm{Brennan},~\bfnm{Paul}\binits{P.}},
  \bauthor{\bsnm{Stratton},~\bfnm{Michael~R.}\binits{M.~R.}},
  \bauthor{\bsnm{Rozen},~\bfnm{Steven~G.}\binits{S.~G.}} \AND
  \bauthor{\bsnm{Alexandrov},~\bfnm{Ludmil~B.}\binits{L.~B.}}
(\byear{2022}).
\btitle{Uncovering Novel Mutational Signatures by de Novo Extraction with
  {{SigProfilerExtractor}}}.
\bjournal{Cell Genomics}
\bvolume{2}.
\bdoi{10.1016/j.xgen.2022.100179}
\end{barticle}
\endbibitem

\bibitem[\protect\citeauthoryear{Jaspers, Kom{\'a}rek and
  Aerts}{2018}]{Jaspers_BayesianEstimationMixtureCovariate_2018}
\begin{barticle}[author]
\bauthor{\bsnm{Jaspers},~\bfnm{Stijn}\binits{S.}},
  \bauthor{\bsnm{Kom{\'a}rek},~\bfnm{Arno{\v s}t}\binits{A.}} \AND
  \bauthor{\bsnm{Aerts},~\bfnm{Marc}\binits{M.}}
(\byear{2018}).
\btitle{Bayesian Estimation of Multivariate Normal Mixtures with
  Covariate-Dependent Mixing Weights, with an Application in Antimicrobial
  Resistance Monitoring}.
\bjournal{Biometrical Journal}
\bvolume{60}
\bpages{7--19}.
\bdoi{10.1002/bimj.201600253}
\end{barticle}
\endbibitem

\bibitem[\protect\citeauthoryear{Ji
  et~al.}{2016}]{Ji_Estimatng_Vegetation_Fractional_Cover_2016}
\begin{barticle}[author]
\bauthor{\bsnm{Ji},~\bfnm{Cuicui}\binits{C.}},
  \bauthor{\bsnm{Jia},~\bfnm{Y.}\binits{Y.}},
  \bauthor{\bsnm{Li},~\bfnm{Xuan}\binits{X.}} \AND
  \bauthor{\bsnm{Wang},~\bfnm{Jw}\binits{J.}}
(\byear{2016}).
\btitle{Research on Linear and Nonlinear Spectral Mixture Models for Estimating
  Vegetation Fractional Cover of {{Nitraria}} Bushes}.
\bjournal{National Remote Sensing Bulletin}
\bvolume{20}
\bpages{1402--1412}.
\bdoi{10.11834/jrs.20166020}
\end{barticle}
\endbibitem

\bibitem[\protect\citeauthoryear{Kingma and Welling}{2014}]{Kingma2014}
\begin{binproceedings}[author]
\bauthor{\bsnm{Kingma},~\bfnm{Diederik~P}\binits{D.~P.}} \AND
  \bauthor{\bsnm{Welling},~\bfnm{Max}\binits{M.}}
(\byear{2014}).
\btitle{Auto-Encoding Variational Bayes}.
\bseries{International Conference on Learning Representations}.
\end{binproceedings}
\endbibitem

\bibitem[\protect\citeauthoryear{Kingma and Welling}{2019}]{Kingma2019VAEs}
\begin{barticle}[author]
\bauthor{\bsnm{Kingma},~\bfnm{Diederik~P.}\binits{D.~P.}} \AND
  \bauthor{\bsnm{Welling},~\bfnm{Max}\binits{M.}}
(\byear{2019}).
\btitle{An Introduction to Variational Autoencoders}.
\bjournal{Foundations and Trends in Machine Learning}
\bvolume{12}
\bpages{307--392}.
\bdoi{10.1561/2200000056}
\end{barticle}
\endbibitem

\bibitem[\protect\citeauthoryear{Kinker et~al.}{2020}]{Kinker_Pan_Cancer_2020}
\begin{barticle}[author]
\bauthor{\bsnm{Kinker},~\bfnm{Gabriela~S.}\binits{G.~S.}},
  \bauthor{\bsnm{Greenwald},~\bfnm{Alissa~C.}\binits{A.~C.}},
  \bauthor{\bsnm{Tal},~\bfnm{Rotem}\binits{R.}},
  \bauthor{\bsnm{Orlova},~\bfnm{Zhanna}\binits{Z.}},
  \bauthor{\bsnm{Cuoco},~\bfnm{Michael~S.}\binits{M.~S.}},
  \bauthor{\bsnm{McFarland},~\bfnm{James~M.}\binits{J.~M.}},
  \bauthor{\bsnm{Warren},~\bfnm{Allison}\binits{A.}},
  \bauthor{\bsnm{Rodman},~\bfnm{Christopher}\binits{C.}},
  \bauthor{\bsnm{Roth},~\bfnm{Jennifer~A.}\binits{J.~A.}},
  \bauthor{\bsnm{Bender},~\bfnm{Samantha~A.}\binits{S.~A.}},
  \bauthor{\bsnm{Kumar},~\bfnm{Bhavna}\binits{B.}},
  \bauthor{\bsnm{Rocco},~\bfnm{James~W.}\binits{J.~W.}},
  \bauthor{\bsnm{Fernandes},~\bfnm{Pedro~ACM}\binits{P.~A.}},
  \bauthor{\bsnm{Mader},~\bfnm{Christopher~C.}\binits{C.~C.}},
  \bauthor{\bsnm{{Keren-Shaul}},~\bfnm{Hadas}\binits{H.}},
  \bauthor{\bsnm{Plotnikov},~\bfnm{Alexander}\binits{A.}},
  \bauthor{\bsnm{Barr},~\bfnm{Haim}\binits{H.}},
  \bauthor{\bsnm{Tsherniak},~\bfnm{Aviad}\binits{A.}},
  \bauthor{\bsnm{{Rozenblatt-Rosen}},~\bfnm{Orit}\binits{O.}},
  \bauthor{\bsnm{Krizhanovsky},~\bfnm{Valery}\binits{V.}},
  \bauthor{\bsnm{Puram},~\bfnm{Sidharth~V.}\binits{S.~V.}},
  \bauthor{\bsnm{Regev},~\bfnm{Aviv}\binits{A.}} \AND
  \bauthor{\bsnm{Tirosh},~\bfnm{Itay}\binits{I.}}
(\byear{2020}).
\btitle{Pan-Cancer Single Cell {{RNA-seq}} Uncovers Recurring Programs of
  Cellular Heterogeneity}.
\bjournal{Nature genetics}
\bvolume{52}
\bpages{1208--1218}.
\bdoi{10.1038/s41588-020-00726-6}
\end{barticle}
\endbibitem

\bibitem[\protect\citeauthoryear{Kiselev et~al.}{2016}]{sc3}
\begin{barticle}[author]
\bauthor{\bsnm{Kiselev},~\bfnm{Vladimir~Yu.}\binits{V.~Y.}},
  \bauthor{\bsnm{Kirschner},~\bfnm{Kristina}\binits{K.}},
  \bauthor{\bsnm{Schaub},~\bfnm{Michael~T.}\binits{M.~T.}},
  \bauthor{\bsnm{Andrews},~\bfnm{Tallulah}\binits{T.}},
  \bauthor{\bsnm{Chandra},~\bfnm{Tamir}\binits{T.}},
  \bauthor{\bsnm{Natarajan},~\bfnm{Kedar~N}\binits{K.~N.}},
  \bauthor{\bsnm{Reik},~\bfnm{Wolf}\binits{W.}},
  \bauthor{\bsnm{Barahona},~\bfnm{Mauricio}\binits{M.}},
  \bauthor{\bsnm{Green},~\bfnm{Anthony~R}\binits{A.~R.}} \AND
  \bauthor{\bsnm{Hemberg},~\bfnm{Martin}\binits{M.}}
(\byear{2016}).
\btitle{SC3 {\textendash} consensus clustering of single-cell RNA-Seq data}.
\bjournal{bioRxiv}.
\bdoi{10.1101/036558}
\end{barticle}
\endbibitem

\bibitem[\protect\citeauthoryear{Kotliar
  et~al.}{2019}]{Kotliar_Identify_Cell_Idendity_Activity_NMF_2019}
\begin{barticle}[author]
\bauthor{\bsnm{Kotliar},~\bfnm{Dylan}\binits{D.}},
  \bauthor{\bsnm{Veres},~\bfnm{Adrian}\binits{A.}},
  \bauthor{\bsnm{Nagy},~\bfnm{M~Aurel}\binits{M.~A.}},
  \bauthor{\bsnm{Tabrizi},~\bfnm{Shervin}\binits{S.}},
  \bauthor{\bsnm{Hodis},~\bfnm{Eran}\binits{E.}},
  \bauthor{\bsnm{Melton},~\bfnm{Douglas~A}\binits{D.~A.}} \AND
  \bauthor{\bsnm{Sabeti},~\bfnm{Pardis~C}\binits{P.~C.}}
(\byear{2019}).
\btitle{Identifying Gene Expression Programs of Cell-Type Identity and Cellular
  Activity with Single-Cell {{RNA-Seq}}}.
\bjournal{eLife}
\bvolume{8}
\bpages{e43803}.
\bdoi{10.7554/eLife.43803}
\end{barticle}
\endbibitem

\bibitem[\protect\citeauthoryear{Lee and Seung}{2000}]{Lee-Seung_multdiv_2000}
\begin{binproceedings}[author]
\bauthor{\bsnm{Lee},~\bfnm{Daniel}\binits{D.}} \AND
  \bauthor{\bsnm{Seung},~\bfnm{H.~Sebastian}\binits{H.~S.}}
(\byear{2000}).
\btitle{Algorithms for Non-negative Matrix Factorization}.
In \bbooktitle{Advances in Neural Information Processing Systems}
(\beditor{\bfnm{T.}\binits{T.}~\bsnm{Leen}},
  \beditor{\bfnm{T.}\binits{T.}~\bsnm{Dietterich}} \AND
  \beditor{\bfnm{V.}\binits{V.}~\bsnm{Tresp}}, eds.)
\bvolume{13}.
\bpublisher{MIT Press}.
\end{binproceedings}
\endbibitem

\bibitem[\protect\citeauthoryear{Levitin
  et~al.}{2019}]{Levitin_DeNovo_Gene_Signature_Identification_2019}
\begin{barticle}[author]
\bauthor{\bsnm{Levitin},~\bfnm{Hanna~Mendes}\binits{H.~M.}},
  \bauthor{\bsnm{Yuan},~\bfnm{Jinzhou}\binits{J.}},
  \bauthor{\bsnm{Cheng},~\bfnm{Yim~Ling}\binits{Y.~L.}},
  \bauthor{\bsnm{Ruiz},~\bfnm{Francisco~JR}\binits{F.~J.}},
  \bauthor{\bsnm{Bush},~\bfnm{Erin~C}\binits{E.~C.}},
  \bauthor{\bsnm{Bruce},~\bfnm{Jeffrey~N}\binits{J.~N.}},
  \bauthor{\bsnm{Canoll},~\bfnm{Peter}\binits{P.}},
  \bauthor{\bsnm{Iavarone},~\bfnm{Antonio}\binits{A.}},
  \bauthor{\bsnm{Lasorella},~\bfnm{Anna}\binits{A.}},
  \bauthor{\bsnm{Blei},~\bfnm{David~M}\binits{D.~M.}} \AND
  \bauthor{\bsnm{Sims},~\bfnm{Peter~A}\binits{P.~A.}}
(\byear{2019}).
\btitle{De Novo Gene Signature Identification from Single-cell {{RNA}}-seq with
  Hierarchical {{Poisson}} Factorization}.
\bjournal{Molecular Systems Biology}
\bvolume{15}
\bpages{e8557}.
\bdoi{10.15252/msb.20188557}
\end{barticle}
\endbibitem

\bibitem[\protect\citeauthoryear{Li et~al.}{2020}]{pcawgSV:2020}
\begin{barticle}[author]
\bauthor{\bsnm{Li},~\bfnm{Yilong}\binits{Y.}},
  \bauthor{\bsnm{Roberts},~\bfnm{Nicola~D.}\binits{N.~D.}},
  \bauthor{\bsnm{Wala},~\bfnm{Jeremiah~A.}\binits{J.~A.}},
  \bauthor{\bsnm{Shapira},~\bfnm{Ofer}\binits{O.}},
  \bauthor{\bsnm{Schumacher},~\bfnm{Steven~E.}\binits{S.~E.}},
  \bauthor{\bsnm{Kumar},~\bfnm{Kiran}\binits{K.}},
  \bauthor{\bsnm{Khurana},~\bfnm{Ekta}\binits{E.}},
  \bauthor{\bsnm{Waszak},~\bfnm{Sebastian}\binits{S.}},
  \bauthor{\bsnm{Korbel},~\bfnm{Jan~O.}\binits{J.~O.}},
  \bauthor{\bsnm{Haber},~\bfnm{James~E.}\binits{J.~E.}},
  \bauthor{\bsnm{Imielinski},~\bfnm{Marcin}\binits{M.}}, \bauthor{\bsnm{{PCAWG
  Structural Variation Working Group}}},
  \bauthor{\bsnm{Weischenfeldt},~\bfnm{Joachim}\binits{J.}},
  \bauthor{\bsnm{Beroukhim},~\bfnm{Rameen}\binits{R.}},
  \bauthor{\bsnm{Campbell},~\bfnm{Peter~J.}\binits{P.~J.}} \AND
  \bauthor{\bsnm{{PCAWG Consortium}}}
(\byear{2020}).
\btitle{Patterns of somatic structural variation in human cancer genomes.}
\bjournal{Nature}
\bvolume{578}
\bpages{112--121}.
\bdoi{10.1038/s41586-019-1913-9}
\end{barticle}
\endbibitem

\bibitem[\protect\citeauthoryear{Lin and
  Zhang}{2017}]{Lin_RetrievingHydrousMinerals_2017}
\begin{barticle}[author]
\bauthor{\bsnm{Lin},~\bfnm{Honglei}\binits{H.}} \AND
  \bauthor{\bsnm{Zhang},~\bfnm{Xia}\binits{X.}}
(\byear{2017}).
\btitle{Retrieving the Hydrous Minerals on {{Mars}} by Sparse Unmixing and the
  {{Hapke}} Model Using {{MRO}}/{{CRISM}} Data}.
\bjournal{Icarus}
\bvolume{288}
\bpages{160--171}.
\bdoi{10.1016/j.icarus.2017.01.019}
\end{barticle}
\endbibitem

\bibitem[\protect\citeauthoryear{Liu}{2019}]{Liu_Support-Union_2019}
\begin{binproceedings}[author]
\bauthor{\bsnm{Liu},~\bfnm{Zhaoqiang}\binits{Z.}}
(\byear{2019}).
\btitle{Model {{Selection}} for {{Nonnegative Matrix Factorization}} by
  {{Support Union Recovery}}}.
In \bbooktitle{{{ICASSP}} 2019 - 2019 {{IEEE International Conference}} on
  {{Acoustics}}, {{Speech}} and {{Signal Processing}} ({{ICASSP}})}
\bpages{3407--3411}.
\bdoi{10.1109/ICASSP.2019.8683718}
\end{binproceedings}
\endbibitem

\bibitem[\protect\citeauthoryear{McKinnon}{2018}]{flow_cytometry}
\begin{barticle}[author]
\bauthor{\bsnm{McKinnon},~\bfnm{Katherine~M.}\binits{K.~M.}}
(\byear{2018}).
\btitle{Flow Cytometry: An Overview}.
\bjournal{Current Protocols in Immunology}
\bvolume{120}
\bpages{5.1.1-5.1.11}.
\bdoi{https://doi.org/10.1002/cpim.40}
\end{barticle}
\endbibitem

\bibitem[\protect\citeauthoryear{Miller and Dunson}{2019}]{Miller:2019}
\begin{barticle}[author]
\bauthor{\bsnm{Miller},~\bfnm{Jeffrey~W.}\binits{J.~W.}} \AND
  \bauthor{\bsnm{Dunson},~\bfnm{David~B.}\binits{D.~B.}}
(\byear{2019}).
\btitle{Robust Bayesian Inference via Coarsening}.
\bjournal{Journal of the American Statistical Association}
\bvolume{114}
\bpages{1113-1125}.
\bdoi{10.1080/01621459.2018.1469995}
\end{barticle}
\endbibitem

\bibitem[\protect\citeauthoryear{{Nik-Zainal}
  et~al.}{2012}]{Nik-zainal_MutationalProcessesMolding_2012}
\begin{barticle}[author]
\bauthor{\bsnm{{Nik-Zainal}},~\bfnm{Serena}\binits{S.}},
  \bauthor{\bsnm{Alexandrov},~\bfnm{Ludmil~B.}\binits{L.~B.}},
  \bauthor{\bsnm{Wedge},~\bfnm{David~C.}\binits{D.~C.}},
  \bauthor{\bsnm{Van~Loo},~\bfnm{Peter}\binits{P.}},
  \bauthor{\bsnm{Greenman},~\bfnm{Christopher~D.}\binits{C.~D.}},
  \bauthor{\bsnm{Raine},~\bfnm{Keiran}\binits{K.}},
  \bauthor{\bsnm{Jones},~\bfnm{David}\binits{D.}},
  \bauthor{\bsnm{Hinton},~\bfnm{Jonathan}\binits{J.}},
  \bauthor{\bsnm{Marshall},~\bfnm{John}\binits{J.}},
  \bauthor{\bsnm{Stebbings},~\bfnm{Lucy~A.}\binits{L.~A.}},
  \bauthor{\bsnm{Menzies},~\bfnm{Andrew}\binits{A.}},
  \bauthor{\bsnm{Martin},~\bfnm{Sancha}\binits{S.}},
  \bauthor{\bsnm{Leung},~\bfnm{Kenric}\binits{K.}},
  \bauthor{\bsnm{Chen},~\bfnm{Lina}\binits{L.}},
  \bauthor{\bsnm{Leroy},~\bfnm{Catherine}\binits{C.}},
  \bauthor{\bsnm{Ramakrishna},~\bfnm{Manasa}\binits{M.}},
  \bauthor{\bsnm{Rance},~\bfnm{Richard}\binits{R.}},
  \bauthor{\bsnm{Lau},~\bfnm{King~Wai}\binits{K.~W.}},
  \bauthor{\bsnm{Mudie},~\bfnm{Laura~J.}\binits{L.~J.}},
  \bauthor{\bsnm{Varela},~\bfnm{Ignacio}\binits{I.}},
  \bauthor{\bsnm{McBride},~\bfnm{David~J.}\binits{D.~J.}},
  \bauthor{\bsnm{Bignell},~\bfnm{Graham~R.}\binits{G.~R.}},
  \bauthor{\bsnm{Cooke},~\bfnm{Susanna~L.}\binits{S.~L.}},
  \bauthor{\bsnm{Shlien},~\bfnm{Adam}\binits{A.}},
  \bauthor{\bsnm{Gamble},~\bfnm{John}\binits{J.}},
  \bauthor{\bsnm{Whitmore},~\bfnm{Ian}\binits{I.}},
  \bauthor{\bsnm{Maddison},~\bfnm{Mark}\binits{M.}},
  \bauthor{\bsnm{Tarpey},~\bfnm{Patrick~S.}\binits{P.~S.}},
  \bauthor{\bsnm{Davies},~\bfnm{Helen~R.}\binits{H.~R.}},
  \bauthor{\bsnm{Papaemmanuil},~\bfnm{Elli}\binits{E.}},
  \bauthor{\bsnm{Stephens},~\bfnm{Philip~J.}\binits{P.~J.}},
  \bauthor{\bsnm{McLaren},~\bfnm{Stuart}\binits{S.}},
  \bauthor{\bsnm{Butler},~\bfnm{Adam~P.}\binits{A.~P.}},
  \bauthor{\bsnm{Teague},~\bfnm{Jon~W.}\binits{J.~W.}},
  \bauthor{\bsnm{J{\"o}nsson},~\bfnm{G{\"o}ran}\binits{G.}},
  \bauthor{\bsnm{Garber},~\bfnm{Judy~E.}\binits{J.~E.}},
  \bauthor{\bsnm{Silver},~\bfnm{Daniel}\binits{D.}},
  \bauthor{\bsnm{Miron},~\bfnm{Penelope}\binits{P.}},
  \bauthor{\bsnm{Fatima},~\bfnm{Aquila}\binits{A.}},
  \bauthor{\bsnm{Boyault},~\bfnm{Sandrine}\binits{S.}},
  \bauthor{\bsnm{Langer{\o}d},~\bfnm{Anita}\binits{A.}},
  \bauthor{\bsnm{Tutt},~\bfnm{Andrew}\binits{A.}},
  \bauthor{\bsnm{Martens},~\bfnm{John W.~M.}\binits{J.~W.~M.}},
  \bauthor{\bsnm{Aparicio},~\bfnm{Samuel A. J.~R.}\binits{S.~A. J.~R.}},
  \bauthor{\bsnm{Borg},~\bfnm{{\AA}ke}\binits{{\AA}.}},
  \bauthor{\bsnm{Salomon},~\bfnm{Anne~Vincent}\binits{A.~V.}},
  \bauthor{\bsnm{Thomas},~\bfnm{Gilles}\binits{G.}},
  \bauthor{\bsnm{{B{\o}rresen-Dale}},~\bfnm{Anne-Lise}\binits{A.-L.}},
  \bauthor{\bsnm{Richardson},~\bfnm{Andrea~L.}\binits{A.~L.}},
  \bauthor{\bsnm{Neuberger},~\bfnm{Michael~S.}\binits{M.~S.}},
  \bauthor{\bsnm{Futreal},~\bfnm{P.~Andrew}\binits{P.~A.}},
  \bauthor{\bsnm{Campbell},~\bfnm{Peter~J.}\binits{P.~J.}},
  \bauthor{\bsnm{Stratton},~\bfnm{Michael~R.}\binits{M.~R.}} \AND
  \bauthor{\bsnm{{Breast Cancer Working Group of the International Cancer
  Genome Consortium}}}
(\byear{2012}).
\btitle{Mutational Processes Molding the Genomes of 21 Breast Cancers}.
\bjournal{Cell}
\bvolume{149}
\bpages{979--993}.
\bdoi{10.1016/j.cell.2012.04.024}
\end{barticle}
\endbibitem

\bibitem[\protect\citeauthoryear{Paninski}{2003}]{Paninski:2003}
\begin{barticle}[author]
\bauthor{\bsnm{Paninski},~\bfnm{Liam}\binits{L.}}
(\byear{2003}).
\btitle{Estimation of Entropy and Mutual Information}.
\bjournal{Neural Computation}
\bvolume{15}
\bpages{1191-1253}.
\bdoi{10.1162/089976603321780272}
\end{barticle}
\endbibitem

\bibitem[\protect\citeauthoryear{Pelizzola, Laursen and
  Hobolth}{2023}]{Pelizzola_NegBin-NMF_2023}
\begin{barticle}[author]
\bauthor{\bsnm{Pelizzola},~\bfnm{Marta}\binits{M.}},
  \bauthor{\bsnm{Laursen},~\bfnm{Ragnhild}\binits{R.}} \AND
  \bauthor{\bsnm{Hobolth},~\bfnm{Asger}\binits{A.}}
(\byear{2023}).
\btitle{Model Selection and Robust Inference of Mutational Signatures Using
  {{Negative Binomial}} Non-Negative Matrix Factorization}.
\bjournal{BMC Bioinformatics}
\bvolume{24}
\bpages{187}.
\bdoi{10.1186/s12859-023-05304-1}
\end{barticle}
\endbibitem

\bibitem[\protect\citeauthoryear{Prabhakaran et~al.}{2016}]{Prabhakaran:2016}
\begin{binproceedings}[author]
\bauthor{\bsnm{Prabhakaran},~\bfnm{Sandhya}\binits{S.}},
  \bauthor{\bsnm{Azizi},~\bfnm{Elham}\binits{E.}},
  \bauthor{\bsnm{Carr},~\bfnm{Ambrose}\binits{A.}} \AND
  \bauthor{\bsnm{Pe’er},~\bfnm{Dana}\binits{D.}}
(\byear{2016}).
\btitle{Dirichlet Process Mixture Model for Correcting Technical Variation in
  Single-Cell Gene Expression Data}.
In \bbooktitle{Proceedings of The 33rd International Conference on Machine
  Learning}
(\beditor{\bfnm{Maria~Florina}\binits{M.~F.}~\bsnm{Balcan}} \AND
  \beditor{\bfnm{Kilian~Q.}\binits{K.~Q.}~\bsnm{Weinberger}}, eds.).
\bseries{Proceedings of Machine Learning Research}
\bvolume{48}
\bpages{1070--1079}.
\bpublisher{PMLR}, \baddress{New York, New York, USA}.
\end{binproceedings}
\endbibitem

\bibitem[\protect\citeauthoryear{Rajabi and
  Ghassemian}{2015}]{Rajabi_SpectralUnmixingHyperspectral_2015}
\begin{barticle}[author]
\bauthor{\bsnm{Rajabi},~\bfnm{Roozbeh}\binits{R.}} \AND
  \bauthor{\bsnm{Ghassemian},~\bfnm{Hassan}\binits{H.}}
(\byear{2015}).
\btitle{Spectral {{Unmixing}} of {{Hyperspectral Imagery}} Using {{Multilayer
  NMF}}}.
\bjournal{IEEE Geoscience and Remote Sensing Letters}
\bvolume{12}
\bpages{38--42}.
\bdoi{10.1109/LGRS.2014.2325874}
\end{barticle}
\endbibitem

\bibitem[\protect\citeauthoryear{Risso
  et~al.}{2018}]{Risso_General_Flexible_Signal_Extract_2018}
\begin{barticle}[author]
\bauthor{\bsnm{Risso},~\bfnm{Davide}\binits{D.}},
  \bauthor{\bsnm{Perraudeau},~\bfnm{Fanny}\binits{F.}},
  \bauthor{\bsnm{Gribkova},~\bfnm{Svetlana}\binits{S.}},
  \bauthor{\bsnm{Dudoit},~\bfnm{Sandrine}\binits{S.}} \AND
  \bauthor{\bsnm{Vert},~\bfnm{Jean-Philippe}\binits{J.-P.}}
(\byear{2018}).
\btitle{A General and Flexible Method for Signal Extraction from Single-Cell
  {{RNA-seq}} Data}.
\bjournal{Nature Communications}
\bvolume{9}
\bpages{284}.
\bdoi{10.1038/s41467-017-02554-5}
\end{barticle}
\endbibitem

\bibitem[\protect\citeauthoryear{Rohe and Zeng}{2023}]{rohe2023vintage-7f4}
\begin{barticle}[author]
\bauthor{\bsnm{Rohe},~\bfnm{Karl}\binits{K.}} \AND
  \bauthor{\bsnm{Zeng},~\bfnm{Muzhe}\binits{M.}}
(\byear{2023}).
\btitle{Vintage factor analysis with Varimax performs statistical inference}.
\bjournal{Journal of the Royal Statistical Society Series B: Statistical
  Methodology}
\bvolume{85}
\bpages{1037--1060}.
\bdoi{10.1093/jrsssb/qkad029}
\end{barticle}
\endbibitem

\bibitem[\protect\citeauthoryear{Rousseeuw}{1987}]{siluet_coef}
\begin{barticle}[author]
\bauthor{\bsnm{Rousseeuw},~\bfnm{Peter~J.}\binits{P.~J.}}
(\byear{1987}).
\btitle{Silhouettes: A graphical aid to the interpretation and validation of
  cluster analysis}.
\bjournal{Journal of Computational and Applied Mathematics}
\bvolume{20}
\bpages{53-65}.
\bdoi{https://doi.org/10.1016/0377-0427(87)90125-7}
\end{barticle}
\endbibitem

\bibitem[\protect\citeauthoryear{Schwarz}{1978}]{Schwarz:1978}
\begin{barticle}[author]
\bauthor{\bsnm{Schwarz},~\bfnm{Gideon}\binits{G.}}
(\byear{1978}).
\btitle{Estimating the Dimension of a Model}.
\bjournal{The Annals of Statistics}
\bvolume{6}
\bpages{461--464}.
\end{barticle}
\endbibitem

\bibitem[\protect\citeauthoryear{Seplyarskiy
  et~al.}{2021}]{Seplyarskiy_PopulationSequencingData_2021}
\begin{barticle}[author]
\bauthor{\bsnm{Seplyarskiy},~\bfnm{Vladimir~B}\binits{V.~B.}},
  \bauthor{\bsnm{Soldatov},~\bfnm{Ruslan~A}\binits{R.~A.}},
  \bauthor{\bsnm{Koch},~\bfnm{Evan}\binits{E.}},
  \bauthor{\bsnm{McGinty},~\bfnm{Ryan~J}\binits{R.~J.}},
  \bauthor{\bsnm{Goldmann},~\bfnm{Jakob~M}\binits{J.~M.}},
  \bauthor{\bsnm{Hernandez},~\bfnm{Ryan~D}\binits{R.~D.}},
  \bauthor{\bsnm{Barnes},~\bfnm{Kathleen}\binits{K.}},
  \bauthor{\bsnm{Correa},~\bfnm{Adolfo}\binits{A.}},
  \bauthor{\bsnm{Burchard},~\bfnm{Esteban~G}\binits{E.~G.}},
  \bauthor{\bsnm{Ellinor},~\bfnm{Patrick~T}\binits{P.~T.}},
  \bauthor{\bsnm{McGarvey},~\bfnm{Stephen~T}\binits{S.~T.}},
  \bauthor{\bsnm{Mitchell},~\bfnm{Braxton~D}\binits{B.~D.}},
  \bauthor{\bsnm{Vasan},~\bfnm{Ramachandran~S}\binits{R.~S.}},
  \bauthor{\bsnm{Redline},~\bfnm{Susan}\binits{S.}},
  \bauthor{\bsnm{Silverman},~\bfnm{Edwin}\binits{E.}},
  \bauthor{\bsnm{Weiss},~\bfnm{Scott~T}\binits{S.~T.}},
  \bauthor{\bsnm{Arnett},~\bfnm{Donna~K}\binits{D.~K.}},
  \bauthor{\bsnm{Blangero},~\bfnm{John}\binits{J.}},
  \bauthor{\bsnm{Boerwinkle},~\bfnm{Eric}\binits{E.}},
  \bauthor{\bsnm{He},~\bfnm{Jiang}\binits{J.}},
  \bauthor{\bsnm{Montgomery},~\bfnm{Courtney}\binits{C.}},
  \bauthor{\bsnm{Rao},~\bfnm{D~C}\binits{D.~C.}},
  \bauthor{\bsnm{Rotter},~\bfnm{Jerome~I}\binits{J.~I.}},
  \bauthor{\bsnm{Taylor},~\bfnm{Kent~D}\binits{K.~D.}},
  \bauthor{\bsnm{Brody},~\bfnm{Jennifer~A}\binits{J.~A.}},
  \bauthor{\bsnm{Chen},~\bfnm{Yii-Der~Ida}\binits{Y.-D.~I.}},
  \bauthor{\bsnm{{de Las Fuentes}},~\bfnm{Lisa}\binits{L.}},
  \bauthor{\bsnm{Hwu},~\bfnm{Chii-Min}\binits{C.-M.}},
  \bauthor{\bsnm{Rich},~\bfnm{Stephen~S}\binits{S.~S.}},
  \bauthor{\bsnm{Manichaikul},~\bfnm{Ani~W}\binits{A.~W.}},
  \bauthor{\bsnm{Mychaleckyj},~\bfnm{Josyf~C}\binits{J.~C.}},
  \bauthor{\bsnm{Palmer},~\bfnm{Nicholette~D}\binits{N.~D.}},
  \bauthor{\bsnm{Smith},~\bfnm{Jennifer~A}\binits{J.~A.}},
  \bauthor{\bsnm{Kardia},~\bfnm{Sharon L~R}\binits{S.~L.~R.}},
  \bauthor{\bsnm{Peyser},~\bfnm{Patricia~A}\binits{P.~A.}},
  \bauthor{\bsnm{Bielak},~\bfnm{Lawrence~F}\binits{L.~F.}},
  \bauthor{\bsnm{O'Connor},~\bfnm{Timothy~D}\binits{T.~D.}},
  \bauthor{\bsnm{Emery},~\bfnm{Leslie~S}\binits{L.~S.}},
  \bauthor{\bsnm{Gilissen},~\bfnm{Christian}\binits{C.}},
  \bauthor{\bsnm{Wong},~\bfnm{Wendy S~W}\binits{W.~S.~W.}},
  \bauthor{\bsnm{Kharchenko},~\bfnm{Peter~V}\binits{P.~V.}} \AND
  \bauthor{\bsnm{Sunyaev},~\bfnm{Shamil}\binits{S.}}
(\byear{2021}).
\btitle{Population Sequencing Data Reveal a Compendium of Mutational Processes
  in Human Germline}.
\bjournal{Science (New York, N.Y.)}
\bvolume{373}
\bpages{1030--1035}.
\bdoi{10.1126/science.aba7408}
\end{barticle}
\endbibitem

\bibitem[\protect\citeauthoryear{Shmueli}{2010}]{Shmueli:2010}
\begin{barticle}[author]
\bauthor{\bsnm{Shmueli},~\bfnm{Galit}\binits{G.}}
(\byear{2010}).
\btitle{{To Explain or to Predict?}}
\bjournal{Statistical Science}
\bvolume{25}
\bpages{289 -- 310}.
\bdoi{10.1214/10-sts330}
\end{barticle}
\endbibitem

\bibitem[\protect\citeauthoryear{Simon-Gabriel and
  Sch\"{o}lkopf}{2018}]{SimonGabriel:2018}
\begin{barticle}[author]
\bauthor{\bsnm{Simon-Gabriel},~\bfnm{Carl-Johann}\binits{C.-J.}} \AND
  \bauthor{\bsnm{Sch\"{o}lkopf},~\bfnm{Bernhard}\binits{B.}}
(\byear{2018}).
\btitle{{Kernel Distribution Embeddings - Universal Kernels, Characteristic
  Kernels and Kernel Metrics on Distributions.}}
\bjournal{Journal of Machine Learning Research}
\bvolume{19}
\bpages{1 -- 29}.
\end{barticle}
\endbibitem

\bibitem[\protect\citeauthoryear{Simon-Gabriel et~al.}{2023}]{Simon:2020}
\begin{barticle}[author]
\bauthor{\bsnm{Simon-Gabriel},~\bfnm{C.~J.}\binits{C.~J.}},
  \bauthor{\bsnm{Barp},~\bfnm{A.}\binits{A.}},
  \bauthor{\bsnm{Sch{\"o}lkopf},~\bfnm{B.}\binits{B.}} \AND
  \bauthor{\bsnm{Mackey},~\bfnm{L.}\binits{L.}}
(\byear{2023}).
\btitle{Metrizing Weak Convergence with Maximum Mean Discrepancies}.
\bjournal{Journal of Machine Learning Research}
\bvolume{24}.
\end{barticle}
\endbibitem

\bibitem[\protect\citeauthoryear{Spiegelhalter
  et~al.}{2002}]{Spiegelhalter:2002}
\begin{barticle}[author]
\bauthor{\bsnm{Spiegelhalter},~\bfnm{David~J.}\binits{D.~J.}},
  \bauthor{\bsnm{Best},~\bfnm{Nicola~G.}\binits{N.~G.}},
  \bauthor{\bsnm{Carlin},~\bfnm{Bradley~P.}\binits{B.~P.}} \AND
  \bauthor{\bsnm{Van Der~Linde},~\bfnm{Angelika}\binits{A.}}
(\byear{2002}).
\btitle{Bayesian measures of model complexity and fit}.
\bjournal{Journal of the Royal Statistical Society: Series B (Statistical
  Methodology)}
\bvolume{64}
\bpages{583-639}.
\end{barticle}
\endbibitem

\bibitem[\protect\citeauthoryear{Sriperumbudur
  et~al.}{2010}]{Sriperumbudur:2010}
\begin{barticle}[author]
\bauthor{\bsnm{Sriperumbudur},~\bfnm{Bharath~K}\binits{B.~K.}},
  \bauthor{\bsnm{Gretton},~\bfnm{A.}\binits{A.}},
  \bauthor{\bsnm{Fukumizu},~\bfnm{K.}\binits{K.}},
  \bauthor{\bsnm{Sch\"{o}lkopf},~\bfnm{Bernhard}\binits{B.}} \AND
  \bauthor{\bsnm{Lanckriet},~\bfnm{G~R~G}\binits{G.~R.~G.}}
(\byear{2010}).
\btitle{{Hilbert Space Embeddings and Metrics on Probability Measures}}.
\bjournal{Journal of Machine Learning Research}
\bvolume{11}
\bpages{1517 -- 1561}.
\end{barticle}
\endbibitem

\bibitem[\protect\citeauthoryear{Stevens et~al.}{2019}]{Stevens:2019}
\begin{barticle}[author]
\bauthor{\bsnm{Stevens},~\bfnm{Elizabeth}\binits{E.}},
  \bauthor{\bsnm{Dixon},~\bfnm{Dennis~R.}\binits{D.~R.}},
  \bauthor{\bsnm{Novack},~\bfnm{Marlena~N.}\binits{M.~N.}},
  \bauthor{\bsnm{Granpeesheh},~\bfnm{Doreen}\binits{D.}},
  \bauthor{\bsnm{Smith},~\bfnm{Tristram}\binits{T.}} \AND
  \bauthor{\bsnm{Linstead},~\bfnm{Erik}\binits{E.}}
(\byear{2019}).
\btitle{Identification and analysis of behavioral phenotypes in autism spectrum
  disorder via unsupervised machine learning}.
\bjournal{International Journal of Medical Informatics}
\bvolume{129}
\bpages{29-36}.
\bdoi{https://doi.org/10.1016/j.ijmedinf.2019.05.006}
\end{barticle}
\endbibitem

\bibitem[\protect\citeauthoryear{Séjourné, Vialard and
  Peyré}{2022}]{unbsinkhorn}
\begin{bmisc}[author]
\bauthor{\bsnm{Séjourné},~\bfnm{Thibault}\binits{T.}},
  \bauthor{\bsnm{Vialard},~\bfnm{François-Xavier}\binits{F.-X.}} \AND
  \bauthor{\bsnm{Peyré},~\bfnm{Gabriel}\binits{G.}}
(\byear{2022}).
\btitle{Faster Unbalanced Optimal Transport: Translation invariant Sinkhorn and
  1-D Frank-Wolfe}.
\end{bmisc}
\endbibitem

\bibitem[\protect\citeauthoryear{Thorndike}{1953}]{Thorndike_1953}
\begin{barticle}[author]
\bauthor{\bsnm{Thorndike},~\bfnm{Robert~L.}\binits{R.~L.}}
(\byear{1953}).
\btitle{Who Belongs in the Family?}
\bjournal{Psychometrika}
\bvolume{18}
\bpages{267–276}.
\bdoi{10.1007/BF02289263}
\end{barticle}
\endbibitem

\bibitem[\protect\citeauthoryear{Tibshirani, Walther and
  Hastie}{2001}]{gap_stats}
\begin{barticle}[author]
\bauthor{\bsnm{Tibshirani},~\bfnm{Robert}\binits{R.}},
  \bauthor{\bsnm{Walther},~\bfnm{Guenther}\binits{G.}} \AND
  \bauthor{\bsnm{Hastie},~\bfnm{Trevor}\binits{T.}}
(\byear{2001}).
\btitle{Estimating the Number of Clusters in a Data Set via the Gap Statistic}.
\bjournal{Journal of the Royal Statistical Society. Series B (Statistical
  Methodology)}
\bvolume{63}
\bpages{411--423}.
\end{barticle}
\endbibitem

\bibitem[\protect\citeauthoryear{van~der Vaart and Wellner}{1996}]{Vaart:1996}
\begin{bbook}[author]
\bauthor{\bparticle{van~der} \bsnm{Vaart},~\bfnm{AW}\binits{A.}} \AND
  \bauthor{\bsnm{Wellner},~\bfnm{J.}\binits{J.}}
(\byear{1996}).
\btitle{Weak Convergence and Empirical Processes: With Applications to
  Statistics}.
\bseries{Springer Series in Statistics}.
\bpublisher{Springer}.
\end{bbook}
\endbibitem

\bibitem[\protect\citeauthoryear{Villani}{2009}]{Villani:2009}
\begin{bbook}[author]
\bauthor{\bsnm{Villani},~\bfnm{C}\binits{C.}}
(\byear{2009}).
\btitle{{Optimal transport: old and new}}
\bvolume{338}.
\bpublisher{Springer}.
\end{bbook}
\endbibitem

\bibitem[\protect\citeauthoryear{Vinh, Epps and Bailey}{2010}]{ami}
\begin{barticle}[author]
\bauthor{\bsnm{Vinh},~\bfnm{Nguyen~Xuan}\binits{N.~X.}},
  \bauthor{\bsnm{Epps},~\bfnm{Julien}\binits{J.}} \AND
  \bauthor{\bsnm{Bailey},~\bfnm{James}\binits{J.}}
(\byear{2010}).
\btitle{Information Theoretic Measures for Clusterings Comparison: Variants,
  Properties, Normalization and Correction for Chance}.
\bjournal{Journal of Machine Learning Research}
\bvolume{11}
\bpages{2837--2854}.
\end{barticle}
\endbibitem

\bibitem[\protect\citeauthoryear{Walker and Hjort}{2001}]{Walker:2001}
\begin{barticle}[author]
\bauthor{\bsnm{Walker},~\bfnm{Stephen}\binits{S.}} \AND
  \bauthor{\bsnm{Hjort},~\bfnm{Nils~Lid}\binits{N.~L.}}
(\byear{2001}).
\btitle{On Bayesian consistency}.
\bjournal{Journal of the Royal Statistical Society: Series B (Statistical
  Methodology)}
\bvolume{63}
\bpages{811-821}.
\bdoi{https://doi.org/10.1111/1467-9868.00314}
\end{barticle}
\endbibitem

\bibitem[\protect\citeauthoryear{Wang and Blei}{2019}]{Wang:2019}
\begin{barticle}[author]
\bauthor{\bsnm{Wang},~\bfnm{Yixin}\binits{Y.}} \AND
  \bauthor{\bsnm{Blei},~\bfnm{David~M.}\binits{D.~M.}}
(\byear{2019}).
\btitle{Frequentist Consistency of Variational Bayes}.
\bjournal{Journal of the American Statistical Association}
\bvolume{114}
\bpages{1147-1161}.
\end{barticle}
\endbibitem

\bibitem[\protect\citeauthoryear{Wang, Kulkarni and Verdú}{2009}]{Wang:2009}
\begin{barticle}[author]
\bauthor{\bsnm{Wang},~\bfnm{Qing}\binits{Q.}},
  \bauthor{\bsnm{Kulkarni},~\bfnm{Sanjeev}\binits{S.}} \AND
  \bauthor{\bsnm{Verdú},~\bfnm{Sergio}\binits{S.}}
(\byear{2009}).
\btitle{Divergence Estimation for Multidimensional Densities Via
  -Nearest-Neighbor Distances}.
\bjournal{Information Theory, IEEE Transactions on}
\bvolume{55}
\bpages{2392 - 2405}.
\end{barticle}
\endbibitem

\bibitem[\protect\citeauthoryear{Wellner}{1981}]{Wellner:1981}
\begin{barticle}[author]
\bauthor{\bsnm{Wellner},~\bfnm{Jon~A.}\binits{J.~A.}}
(\byear{1981}).
\btitle{A Glivenko-Cantelli theorem for empirical measures of independent but
  non-identically distributed random variables}.
\bjournal{Stochastic Processes and their Applications}
\bvolume{11}
\bpages{309-312}.
\end{barticle}
\endbibitem

\bibitem[\protect\citeauthoryear{West}{2003}]{West:2003}
\begin{barticle}[author]
\bauthor{\bsnm{West},~\bfnm{Mike}\binits{M.}}
(\byear{2003}).
\btitle{{Bayesian Factor Regression Models in the “Large p, Small n”
  Paradigm}}.
\bjournal{Bayesian Statistics}
\bvolume{7}.
\end{barticle}
\endbibitem

\bibitem[\protect\citeauthoryear{Wu et~al.}{2024}]{Wu:2024}
\begin{barticle}[author]
\bauthor{\bsnm{Wu},~\bfnm{Bohan}\binits{B.}},
  \bauthor{\bsnm{Weinstein},~\bfnm{Eli~N}\binits{E.~N.}},
  \bauthor{\bsnm{Salehi},~\bfnm{Sohrab}\binits{S.}},
  \bauthor{\bsnm{Wang},~\bfnm{Yixin}\binits{Y.}} \AND
  \bauthor{\bsnm{Blei},~\bfnm{David~M}\binits{D.~M.}}
(\byear{2024}).
\btitle{{Adaptive Nonparametric Perturbations of Parametric Bayesian Models}}.
\bjournal{arXiv}.
\bdoi{10.48550/arxiv.2412.10683}
\end{barticle}
\endbibitem

\bibitem[\protect\citeauthoryear{Xue et~al.}{2024}]{Xue:2024}
\begin{barticle}[author]
\bauthor{\bsnm{Xue},~\bfnm{Catherine}\binits{C.}},
  \bauthor{\bsnm{Miller},~\bfnm{Jeffrey~W.}\binits{J.~W.}},
  \bauthor{\bsnm{Carter},~\bfnm{Scott~L.}\binits{S.~L.}} \AND
  \bauthor{\bsnm{Huggins},~\bfnm{Jonathan~H.}\binits{J.~H.}}
(\byear{2024}).
\btitle{{Robust discovery of mutational signatures using power posteriors}}.
\bjournal{bioRxiv}
\bpages{2024.10.23.619958}.
\bdoi{10.1101/2024.10.23.619958}
\end{barticle}
\endbibitem

\bibitem[\protect\citeauthoryear{Zhao and Lai}{2020}]{Zhao:2020}
\begin{binproceedings}[author]
\bauthor{\bsnm{Zhao},~\bfnm{Puning}\binits{P.}} \AND
  \bauthor{\bsnm{Lai},~\bfnm{Lifeng}\binits{L.}}
(\byear{2020}).
\btitle{Analysis of K Nearest Neighbor KL Divergence Estimation for Continuous
  Distributions}.
In \bbooktitle{2020 IEEE International Symposium on Information Theory (ISIT)}
\bpages{2562-2567}.
\bdoi{10.1109/ISIT44484.2020.9174033}
\end{binproceedings}
\endbibitem

\bibitem[\protect\citeauthoryear{Zhao and
  Tan}{2017}]{Zhao_ConvergenceAnalysisMU_2017}
\begin{bmisc}[author]
\bauthor{\bsnm{Zhao},~\bfnm{Renbo}\binits{R.}} \AND
  \bauthor{\bsnm{Tan},~\bfnm{Vincent Y.~F.}\binits{V.~Y.~F.}}
(\byear{2017}).
\btitle{A {{Unified Convergence Analysis}} of the {{Multiplicative Update
  Algorithm}} for {{Regularized Nonnegative Matrix Factorization}}}.
\bdoi{10.48550/arXiv.1609.00951}
\end{bmisc}
\endbibitem

\end{thebibliography}
